%% file: IPE_delta_modular.tex
\documentclass[sn-mathphys]{sn-jnl}
\jyear{2021}

\usepackage[utf8]{inputenc}
\usepackage[english]{babel}
\usepackage{color}
\usepackage{commath}

\theoremstyle{thmstyleone}%
\newtheorem{theorem}{Theorem}
\newtheorem{lemma}{Lemma}%
\newtheorem{corollary}{Corollary}

\theoremstyle{thmstyletwo}%
\newtheorem{remark}{Remark}%

\theoremstyle{thmstylethree}%
\newtheorem{definition}{Definition}%

\input{grib_defs}

\DeclareMathOperator{\BILPSF}{BILP-SF}

\DeclareMathOperator{\BILPCF}{BILP-CF}
\DeclareMathOperator{\ILPSF}{ILP-SF}
\DeclareMathOperator{\ILPCF}{ILP-CF}

\renewcommand{\GribanovAdd}[1]{{#1}}

\begin{document}

\title[On $\Delta$-Modular Integer Linear Problems In The Canonical Form]{On $\Delta$-Modular Integer Linear Problems In The Canonical Form And Equivalent Problems
}

\author*[1]{\fnm{Dmitry} \sur{Gribanov}}\email{dimitry.gribanov@gmail.com}

\author[2]{\fnm{Ivan} \sur{Shumilov}}\email{ivan.a.shumilov@gmail.com}

\author[1]{\fnm{Dmitry} \sur{Malyshev}}\email{dsmalyshev@rambler.ru}

\author[3]{\fnm{Panos} \sur{Pardalos}}\email{panos.pardalos@gmail.com}

\affil*[1]{\orgdiv{Laboratory of Algorithms and Technologies for Network Analysis}, \orgname{HSE University}, \orgaddress{\street{136 Rodionova Ulitsa}, \city{Nizhny Novgorod}, \postcode{603093}, \country{Russian Federation}}}

\affil[2]{\orgname{Lobachevsky State University of Nizhny Novgorod}, \orgaddress{\street{23 Gagarina Avenue}, \city{Nizhny Novgorod}, \postcode{603950}, \country{Russian Federation}}}

\affil[3]{\orgdiv{Department of Industrial and Systems Engineering}, \orgname{University of Florida}, \orgaddress{\street{401 Weil Hall}, \city{Gainesville}, \postcode{116595}, \state{Florida}, \country{USA}}}

\abstract{
Many papers in the field of integer linear programming (ILP, for short) are devoted to problems of the type $\max\{c^\top x \colon A x = b,\, x \in \ZZ^n_{\geq 0}\}$, where all the entries of $A,b,c$ are integer, parameterized by the number of rows of $A$ and $\|A\|_{\max}$. This class of problems is known under the name of ILP problems in the standard form, adding the word "bounded" if $x \leq u$, for some integer vector $u$. Recently, many new sparsity, proximity, and complexity results were obtained for bounded and unbounded ILP problems in the standard form.

In this paper, we consider ILP problems in the canonical form $$\max\{c^\top x \colon b_l \leq A x \leq b_r,\, x \in \ZZ^n\},$$ where $b_l$ and $b_r$ are integer vectors. We assume that the integer matrix $A$ has the rank $n$, $(n + m)$ rows, $n$ columns, and parameterize the problem by $m$ and $\Delta(A)$, where $\Delta(A)$ is the maximum of $n \times n$ sub-determinants of $A$, taken in the absolute value. We show that any ILP problem in the standard form can be polynomially reduced to some ILP problem in the canonical form, preserving $m$ and $\Delta(A)$, but the reverse reduction is not always possible. More precisely, we define the class of generalized ILP problems in the standard form, which includes an additional group constraint, and prove the equivalence to ILP problems in the canonical form.

We generalize known sparsity, proximity, and complexity bounds for ILP problems in the canonical form. Additionally, sometimes, we strengthen previously known results for ILP problems in the canonical form, and, sometimes, we give shorter proofs. Finally, we consider the special cases of $m \in \{0,1\}$. By this way, we give specialised sparsity, proximity, and complexity bounds for the problems on simplices, Knapsack problems and Subset-Sum problems.
}

\maketitle  


\section{Introduction}\label{intro_sec}

\subsection{Basic Definitions And Notations}

Let $A \in \mathbb{Z}^{m \times n}$ be an integer matrix. We denote by $A_{ij}$ the $ij$-th element of the matrix, by $A_{i*}$ its $i$-th row, and by $A_{*j}$ its $j$-th column. For sets $\IC \subseteq \{1,\dots,m\}$ and $\JC \subseteq \{1,\dots,n\}$, the symbol $A_{\IC \JC}$ denotes the sub-matrix of $A$, which is generated by all the rows with indices in $\IC$ and all the columns with indices in $\JC$. If $\IC$ or $\JC$ is replaced by $*$, then all the rows or columns are selected, respectively. Sometimes, we simply write $A_{\IC}$ instead of $A_{\IC*}$ and $A_{\JC}$ instead of $A_{*\JC}$, if this does not lead to confusion.

\GribanovAdd{
We denote by $I_{n \times n}$, $\BUnit_{m \times n}$, and $\BZero_{m \times n}$ the $n \times n$ identity matrix, the $m \times n$ all-one matrix, and the $m \times n$ all-zero matrix, respectively. Sometimes, we will ignore subscripts, if this does not lead to confusion. The same notation is used for all-one and all-zero vectors.
}

The maximum absolute value of entries of $A$ is denoted by $\|A\|_{\max} = \max_{i,j} \abs{A_{i\,j}}$. 

The $l_p$-norm of a vector $x$ is denoted by $\norm{x}_p$. The number of non-zero components of a vector $x$ is denoted by $\norm{x}_0=\abs{\{i\colon x_i \not= 0\}}$. The column, compounded by diagonal elements of a $n \times n$ matrix $A$, is denoted by $\diag(A) = (A_{11},\dots,A_{nn})^\top$. The adjugate matrix is denoted by $A^* = \det(A) A^{-1}$.

For $x \in \RR$, we denote by $\lfloor x \rfloor$, $\{x\}$, and $\lceil x \rceil$ the floor, fractional part, and ceiling of $x$, respectively. 

Let $S \in \ZZ_{\geq 0}^{n \times n}$ be a diagonal matrix and $v \in \ZZ^n$. We denote by $v \bmod S$ the vector, whose $i$-th component equals $v_i \bmod S_{i i}$. For $\MC \subseteq \ZZ^n$, we denote $\MC \bmod\, S = \{ v \bmod S \colon v \in \MC \}$. For example, the set $\ZZ^n \bmod\, S$ consists of $\det(S)$ elements. 

\begin{definition}
For a matrix $A \in \ZZ^{m \times n}$, by $$
\Delta_k(A) = \max\left\{\abs{\det (A_{\IC \JC})} \colon \IC \subseteq \intint m,\, \JC \subseteq \intint n,\, \abs{\IC} = \abs{\JC} = k\right\},
$$ we denote the maximum absolute value of determinants of all the $k \times k$ sub-matrices of $A$.  By $\Delta_{\gcd}(A,k)$ and $\Delta_{\lcm}(A,k)$, we denote the greatest common divisor and the least common multiplier of nonzero determinants of all the $k \times k$ sub-matrices of $A$, respectively. Additionally, let $\Delta(A) = \Delta_{\rank A}(A)$, $\Delta_{\gcd}(A) = \Delta_{\gcd}(A,\rank(A))$ and $\Delta_{\lcm}(A) = \Delta_{\lcm}(A,\rank(A))$.

The matrix $A$ with $\Delta(A) \leq \Delta$, for some $\Delta > 0$, is called \emph{$\Delta$-modular}. \GribanovAdd{Note that $\Delta_1(A) = \|A\|_{\max}$.}
\end{definition}

\begin{definition}
For a matrix $B \in \RR^{m \times n}$,
$\cone(B) = \{B t\colon t \in \RR_{\geq 0}^{n} \}$ is the \emph{cone spanned by columns of} $B$,
$\conv(B) = \{B t\colon t \in \RR_{\geq 0}^{n},\, \sum_{i=1}^{n} t_i = 1  \}$ is the \emph{convex hull spanned by columns of} $B$,
$\affh(B) = \{B t\colon t \in \RR^n,\, \sum_{i=1}^{n} t_i = 1\}$ is the \emph{affine hull spanned by columns of} $B$,
$\linh(B) = \{B t\colon t \in \RR^n \}$ is the \emph{linear hull spanned by columns of} $B$,
$\inth(B) = \{x = B t \colon t \in \ZZ^n\}$ is the \emph{lattice spanned by columns of} $B$.
\end{definition}

\begin{definition}
For $A \in \ZZ^{m \times n}$ and $b \in \ZZ^m$, we denote 
$$
\PC(A,b) = \{x \in \RR^n \colon A x \leq b\} \quad\text{and}\quad \PC_I(A,b) = \conv(\PC(A,b) \cap \ZZ^n).
$$
\end{definition}

\subsection{The Smith And Hermite Normal Forms}\label{SNF_subs}

Let $A \in \ZZ^{m \times n}$ be an integer matrix of rank $n$, \GribanovAdd{ assuming that the first $n$ rows of $A$ are linearly independent.} It is a known fact (see, for example, \cite{Schrijver,HNFOptAlg}) that there exists a unimodular matrix $Q \in \ZZ^{n \times n}$, such that $A = \binom{H}{B} Q$, where $B \in \ZZ^{(m-n) \times n}$ and $H \in \ZZ_{\geq 0}^{n \times n}$ is a lower-triangular matrix, such that $0 \leq H_{i j} < H_{i i}$, for any $i \in \intint n$ and $j \in \intint{i-1}$. The matrix $\binom{H}{B}$ is called the \emph{Hermite Normal Form} (or, shortly, the HNF) of the matrix $A$. Additionally, it was shown in \cite{FPT_Grib} that $\|B\|_{\max} \leq \Delta(A)$ and, consequently, $\|\binom{H}{B}\|_{\max} \leq \Delta(A)$.

It is a known fact (see, for example, \cite{Schrijver,SNFOptAlg,Zhendong}) that there exist unimodular matrices $P \in \ZZ^{m \times m}$ and $Q \in \ZZ^{n \times n}$, such that $A = P \dbinom{S}{\BZero_{d \times n}} Q$, where $d = m - n$ and $S \in \ZZ_{\geq 0}^{n \times n}$ is a diagonal non-degenerate matrix. Moreover, $\prod_{i = 1}^{k} S_{ii} = \Delta_{\gcd}(A,k)$, and, consequently, $S_{ii} \mid S_{(i+1) (i+1)}$, for $i \in \intint{n-1}$. The matrix $\dbinom{S}{\BZero_{d \times n}}$ is called the \emph {Smith Normal Form} (or, shortly, the SNF) of the matrix $A$.

Near-optimal polynomial-time algorithms for constructing the HNF and SNF of $A$ are given in \cite{HNFOptAlg,SNFOptAlg}. 

We will significantly use the following Lemma.
\begin{lemma}\label{par_enum_lm}
Let $A \in \ZZ^{n \times n}$, $p \in \QQ^n$, $\Delta = \abs{\det(A)} > 0$, and $\gamma > 0$. Let us consider the intersection of a parallelepiped with $\ZZ^n$:
$$
\PC = \ZZ^n \cap \{A x \colon \|x-p\|_\infty \leq \gamma \}.
$$
Then, $\abs{\PC} \leq (2 \gamma + 1)^n \cdot \Delta$.

All the points of $\PC$ can be enumerated by an algorithm with the arithmetic complexity bound:
$$
O\bigl((2 \gamma +1)^n \cdot n \cdot \Delta \cdot \min\{n, \log_2(\Delta)\}\bigr).
$$
\end{lemma}
We omit a proof, because it is completely similar to the proofs from \cite{WidthSimplex_Grib} or \cite{EffectiveCounting}.

\subsection{The Problems Under Consideration And Motivation Of This Paper}\label{problems_subs}

The key research object of many works in the field of ILP is the following ILP problem in the standard form:  
\begin{gather}
    c^\top x \to \max\notag\\
    \begin{cases}
    A x = b\\
    0 \leq x \leq u\\
    x \in \ZZ^n.
    \end{cases}\label{CLASSIC-ILP-SF}\tag{CLASSIC-ILP-SF}
\end{gather}
Here, $A \in \ZZ^{m \times n}$, $\rank(A) = m$, $b \in \ZZ^m$, $c\in \ZZ^n$, and $u \in \bigl(\ZZ_{\geq 0} \cup \{+\infty\}\bigr)^n$. We call this problem and similar problems \emph{unbounded}, if $u_i = +\infty$, for all $i \in \intint n$, and \emph{bounded} in the opposite case. One of the most popular and actual research directions is investigation of properties of the \ref{CLASSIC-ILP-SF} problem, parameterized by $m$ and $\Delta_1 := \Delta_1(A) = \|A\|_{\max}$ or by $m$ and $\Delta := \Delta(A)$. Some interesting properties that arise in many papers in this field are: sparsity bounds for integer vertices coordinates, proximity bounds that estimate the distance between optimal vertices of ILP and the related LP  problems, bounds on the number of integer vertices, the complexity of integer programming algorithms and etc. We list some interesting recent results for the unbounded variant of the \ref{CLASSIC-ILP-SF} problem. Let $z$ be an arbitrary vertex of the integer polyhedron of an unbounded  \ref{CLASSIC-ILP-SF} problem. Assuming $\Delta_{\gcd}(A) = 1$, it was showed in \cite{SupportIPSolutions} that 
$$
\norm{z}_0 \leq m + \log_2 \sqrt{\det(A A^\top)} \leq 2m \cdot \log_2(2 \sqrt{m} \cdot \Delta_1).
$$ Additionally, it was showed in \cite{SupportIPSolutions} that number of integer vertices is bounded by
$$
(n \cdot m \cdot \Delta_1)^{O(m^2 \cdot \log(\sqrt{m} \cdot \Delta_1))}.
$$
The paper \cite{NewBoundsForFixedM} strengthens both bounds. More precisely, it gives the bound
$$
\|z\|_0 \leq \frac{3}{2} m \cdot \log_2(O( \sqrt{m} \cdot \Delta_1) ),
$$ for the integer vertex support (see \cite[Theorem~30]{NewBoundsForFixedM}) and the bound
$$
(n \cdot m \cdot \log(m \Delta_1))^{O(m \cdot \log(\sqrt{m} \Delta_1))}
$$ for the number of integer vertices.

Let $v$ be an optimal vertex solution of the LP relaxation of the problem. The paper \cite{ProximityViaSparsity} gives sparsity and proximity bounds with respect to $m$ and $\Delta$. More precisely, it was shown there that for any integer vertex $z$ we have
$$
\|z\|_0 \leq 2m \cdot \log_2(\sqrt{2m} \cdot \Delta^{1/m}).
$$ and,
for some optimal integer vertex solution $z'$, we have
$$
\|v - z'\|_1 \leq 3m^2 \cdot \log_2(2 \sqrt{2 m} \cdot \Delta^{1/m}) \cdot \Delta.
$$

The paper \cite{DiscConvILP} gives a dynamic programming algorithm with the arithmetic complexity bound
\begin{multline*}
    O(n m) + O(\eta)^{2m} \cdot \frac{\log(\|b\|_{\infty} + \Delta_1)}{ \log(\Delta_1)} = \\
    = O(nm) + O(\sqrt{m} \Delta_1)^{2m} \cdot \frac{\log(\|b\|_{\infty} + \Delta_1)}{ \log(\Delta_1)}.
\end{multline*}
Here, $\eta$ is the hereditary discrepancy bound for the matrix $A$. Due to \cite{HerDisc} and \cite{SixDeviations_Spencer}, it is known that $\eta = O(\sqrt{m} \Delta_1)$. Moreover, it was shown in \cite{DiscConvILP} that there are no algorithms with dependence $\Delta_1^{2 m - \varepsilon}$, for any $\varepsilon > 0$, otherwise the strong exponential time hypothesis (SETH, for short) does not hold. Similarly, due to \cite{TightLowerBoundsM}, the dependence $m^{O(m)}$ cannot be reduced, otherwise SETH does not hold.

Now, we list some results for the bounded \ref{CLASSIC-ILP-SF} problem. The following proximity bound is given in \cite{SteinitzILP}:
$$
\|z - v\|_1 \leq m \cdot (2 m \cdot \Delta_1 + 1)^m.
$$ It is noted in \cite{ProximityViaSparsity} that the previous bound can be extended to work with the $\Delta$ parameter:
$$
\|z - v\|_1 \leq m \cdot \Delta \cdot (2 m + 1)^m.
$$ Additionally, the paper \cite{SteinitzILP} gives a dynamic programming algorithm with the arithmetic complexity bound
$$
n \cdot O(m)^{(m+1)^2} \cdot O(\Delta_1)^{m(m+1)} \cdot \log^2(m \Delta_1).
$$

The proximity bound that is polynomial both in $m$ and $\Delta$ is given in \cite{ModularDiffColumns}:
$$
\|z-v\|_1 = O(m^3 \cdot \Delta^3).
$$ More precisely, the paper \cite{ModularDiffColumns} contains a much stronger result. It was shown there that any matrix $A \in \ZZ^{m \times n}$ without pairs $(u,v)$ of columns, such that $u = \pm v$, can have at most 
$$
1/2 \cdot (m^2 + m) \cdot \Delta^2 
$$ columns.

Finally, we mention some results about properties of the unbounded \ref{CLASSIC-ILP-SF} problem in an average case. Assuming that $A \in \ZZ^{m \times n}$ and $c \in \ZZ^n$ are fixed, $\Delta_{\gcd}(A) = 1$, and the vector $b \in \ZZ^m$ is taken randomly, it was shown in \cite{SparsityAverage} and \cite{DistributionsILP} that for almost all $b$ we have:
\begin{enumerate}
    \item $\|z\|_0 \leq m + \log_2(\Delta)$,
    \item $\|z-v\|_1 \leq (m+1)\cdot(\Delta-1)$,
    \item $\|z-v\|_\infty \leq \Delta-1$.
\end{enumerate}

There are many results for the \ref{CLASSIC-ILP-SF} problem in both bounded and unbounded variants. In our paper, we are going to work with problems in the canonical form. The reason is that, as it will be shown later, the problems in the canonical form are slightly more general and, sometimes, more suitable to obtain results in the area of ILP.

One of the main research objects of the current work is the following ILP problem in the canonical form:
\begin{definition}
Let $A \in \ZZ^{(n + m)\times n}$, $rank(A) = n$, $c \in \ZZ^n$, $b_r \in \ZZ^{n+m}$ and $b_l \in \bigl(\ZZ^{n+m} \cup \{-\infty\}\bigr)^n$. \emph{The bounded integer linear programming problem in the canonical form} is formulated as follows:
\begin{gather}
    c^\top x \to \max \notag\\
    \begin{cases}
    b_l \leq A x \leq b_r\\
    x \in \ZZ^n.
    \end{cases}\label{BILP-CF}\tag{$\BILPCF$}
\end{gather}

\emph{The classical ILP problem in the canonical form} corresponds to the case $(b_l)_i = -\infty$, for $i \in \intint{n+m}$, and $b := b_r$:
\begin{gather}
    c^\top x \to \max \notag\\
    \begin{cases}
    A x \leq b\\
    x \in \ZZ^n.
    \end{cases}\label{ILP-CF}\tag{$\ILPCF$}
\end{gather}
\end{definition}

Why is it important to work with these problems? We will show later that the \ref{ILP-CF} and \ref{BILP-CF} problems are more general, than bounded or unbounded problems of the type \ref{CLASSIC-ILP-SF}. More precisely, the paper \cite{BoroshTreybigProof} (see also \cite{CountingFixedM}) gives a way how to reduce the \ref{CLASSIC-ILP-SF} problem to the \ref{ILP-CF} problem, but the inverse reduction, preserving values of the parameters $m$ and $\Delta(A)$, is not always possible (see Remark \ref{reduction_impossibility} for a counterexample). Consequently, any fact that is true for the \ref{BILP-CF} or \ref{ILP-CF} problems can be easily restated for the bounded and unbounded variants of the \ref{CLASSIC-ILP-SF} problem. However, due to Lemmas \ref{ILPCF_to_ILPSF_lm} and \ref{ILPSF_to_ILPCF_lm} of this paper, the \ref{ILP-CF} and \ref{BILP-CF} problems are polynomially-equivalent to generalized problems in the standard form with an additional group constraint. We define these problems here:
\begin{definition}
Let $A \in \ZZ^{m \times n}$ and $G \in \ZZ^{(n-m) \times n}$, such that $\dbinom{A}{G}$ is an integer $n \times n$ unimodular matrix.
Let, additionally, $S \in \ZZ^{(n-m) \times (n-m)}$ be  some  matrix  reduced  to  the  SNF, $g \in \ZZ^{(n-m)}$, $b \in \ZZ^m$, $c \in \ZZ^n$ and $u \in \bigl( \ZZ_{\geq 0} \cup \{+\infty\} \bigr)^n$.
\emph{The generalized bounded integer linear programming problem in the standard form} is formulated as follows:
\begin{gather}
    c^\top x \to \min \notag\\
    \begin{cases}
    A x = b\\
    G x \equiv g \pmod{S\cdot\ZZ^n}\\
    0 \leq x \leq u\\
    x \in \ZZ^n.
    \end{cases}\label{BILP-SF}\tag{$\BILPSF$}
\end{gather}

The unbounded case of \ref{BILP-SF} corresponds to the case $u_i = +\infty$, for $i \in \intint n$:
\begin{gather}
    c^\top x \to \min \notag\\
    \begin{cases}
    A x = b\\
    G x \equiv g \pmod{S \cdot \ZZ^n}\\
    x \in \ZZ_{\geq 0}^n.
    \end{cases}\label{ILP-SF}\tag{$\ILPSF$}
\end{gather}
\end{definition}

It is not straightforward clear how to define a linear programming relaxation of the \ref{BILP-SF} problem. So, we emphasise it as a stand-alone definition:
\begin{definition}\label{BLPSF_def}
In terms of the previous definition, the problem 
\begin{gather*}
    c^\top x \to \min \notag\\
    \begin{cases}
    A x = b\\
    0 \leq x \leq u\\
    x \in \RR^n.
    \end{cases}\label{BLP-SF}
\end{gather*} is called \emph{the linear programming relaxation} of the \ref{BILP-SF} problem.
\end{definition}

\begin{remark}\label{classic_ILP_rm}
We note that the \ref{CLASSIC-ILP-SF} problem in the standard form can be polynomially reduced to the \ref{BILP-SF} and \ref{ILP-SF} problems.

First of all, we can assume that $\Delta_{\gcd}(A) = 1$. Indeed, let $A = P\, \bigl(S\,\BZero\bigr)\, Q$, where $\bigl(S\,\BZero\bigr) \in \ZZ^{m \times n}$ be the SNF of $A$, and $P \in \ZZ^{m \times m}$, $Q \in \ZZ^{n \times n}$ be unimodular matrices. We multiply rows of the original system $A x = b,\, x \geq \BZero$ by the matrix $(P S)^{-1}$. After this step, the original system is transformed to the equivalent system $\bigl(I_{m \times m}\,\BZero\bigr)\,Q\,x = b^\prime$, $x \geq \BZero$. In the last formula $b^\prime \in \ZZ^m$, because in the opposite case the system is integrally infeasible. Clearly, the matrix $\bigl(I_{m \times m}\,\BZero\bigr)$ is the SNF of $\bigl(I_{m \times m}\,\BZero\bigr)\,Q$, so its $\Delta_{\gcd}(\cdot)$ is equal to $1$.

Since $\Delta_{\gcd}(A) = 1$, the columns of $A^\top$ form a primitive basis of some sub-lattice of $\ZZ^n$. Hence, it can be extended to a full basis of $\ZZ^n$. Let the columns of $G^\top \in \ZZ^{(n-m) \times n}$ form this extension. Consequently, $\dbinom{A}{G}$ is a $n \times n$ integral unimodular matrix, and the \ref{BILP-SF} problem
\begin{gather*}
    c^\top x \to \min\\
    \begin{cases}
    A x = b\\
    G x \equiv \BZero \pmod{I \cdot \ZZ^n}\\
    0 \leq x \leq u\\
    x \in \ZZ^n,
    \end{cases}
\end{gather*}
is equivalent to the original \ref{CLASSIC-ILP-SF} problem (here, $I$ is the $(n-m) \times (n-m)$ identity matrix). 
\end{remark}

We note some interesting results, devoted to the last problems. It is shown in \cite{BoroshTreybigProof} that the \ref{ILP-SF} problem has an integer feasible solution $z$ with 
$$
\norm{z}_\infty \leq \Delta_{ext} \cdot \abs{\det(S)},
$$ where $\Delta_{ext} = \Delta(\bigl(A\,b\bigr))$ is the maximal absolute value of $m \times m$ minors of the extended matrix $\bigl(A\,b\bigr)$. Previously, this inequality in a weaker form was conjectured by I.~Borosh and T.~Treybig in \cite{BoroshTreybig1,BoroshTreybig2}. For the \ref{ILP-CF} problem, it means the existence of an integer feasible solution $z$ with
$$
\norm{b - A z}_\infty \leq \Delta_{ext}.
$$

Due to \cite{CountingFixedM}, the number of integer feasible solutions in the \ref{ILP-CF} problem can be computed by an algorithm with the arithmetic complexity bound 
$$
O(n^{m + O(1)} \cdot n^{\log_2(\Delta)}).
$$ A similar result is formulated for the unbounded \ref{CLASSIC-ILP-SF} problem.

Due to \cite{FPT_Grib}, the \ref{ILP-CF} problem can be solved by an algorithm with the arithmetic complexity bound 
\begin{equation}\label{old_ILPCF_complexity}
O\bigl((m + \log(\Delta)) \cdot n^{2(m+1)} \cdot \Delta^{2(m+1)}\bigr).    
\end{equation}

Let us consider the variant of the \ref{ILP-CF} problem, such that the matrix $A \in \ZZ^{(n+m)\times n}$ does not have any degenerate $n \times n$ sub-matrix, or, in other words, there are no $n \times n$ minors that are equal $0$. Due to \cite{NonDegenerateMinors}, if $n \geq f(\Delta)$, for some function $f(\Delta)$ that depends only from $\Delta$, then the parameter $m$ can not be greater than $1$. The last result was improved in \cite{FPT_Grib}, more precisely, it was shown that $f(\Delta) \leq \Delta \cdot (2 \Delta + 1)^2 + \log_2(\Delta)$. Note that for $n \leq f(\Delta)$ the problem can be solved by $O(n)^n$-complexity algorithm due to \cite{DadushFDim}, see also \cite{DadushDis}.

Consequently, due to the last complexity bound, the \ref{ILP-CF} problem without $0$-valued $n \times n$ sub-determinants can be solved by an algorithm with arithmetic complexity bound 
\begin{equation}\label{old_ILPCF_0Minor_complexity}
    [n \leq f(\Delta)]\cdot f(\Delta)^{f(\Delta)} + [n > f(\Delta)] \cdot \log_2(\Delta) \cdot n^4 \cdot \Delta^4.
\end{equation}
\GribanovAdd{Here $[\cdot]$ is the K.~Iverson's notation, see \cite{ConcreteMathematics}, i.e. $[true] = 1$ and $[false] = 0$.}

The last complexity bounds \eqref{old_ILPCF_complexity} and \eqref{old_ILPCF_0Minor_complexity} will be sufficiently improved in our work, see Subsection \ref{results_part1_subs}.

The paper \cite{IntVertexEnumSimplex} gives a polynomial-time algorithm to enumerate integer vertices of a simplex, defined by a $\Delta$-modular system $A x \leq b$. Consequently, if $n \geq f(\Delta)$, a similar algorithm exists for any $n$-dimensional polyhedron, defined by a $\Delta$-modular system. The last result will be generalized in our work for any polyhedron of the \ref{ILP-CF} problem, see Subsection \ref{results_part1_subs}.

Finally, let us consider known results on the number of vertices in the integer polyhedron $\PC_I$ of the \ref{ILP-CF} problem. Let $\xi(n,k)$ denote the maximum number of vertices in $n$-dimensional polyhedra with $k$ facets. \GribanovAdd{Due to the seminal result \cite{MaxFacesTh} of P.~McMullen, the value of $\xi(n,k)$ equals the number of vertices of polytopes that are dual to cyclic polytopes with $k$ vertices.}

\GribanovAdd{
Consequently, due to \cite[Section~4.7]{Grunbaum},
\begin{equation*}\label{cyclic_poly_faces_num}
    \xi(n,k) = \begin{cases}
    \frac{k}{k-s} \binom{k-s}{s},\text{ for }n = 2s\\
    2\binom{k-s-1}{s},\text{ for }n = 2s+1\\
    \end{cases} = O\left(\frac{k}{n}\right)^{n/2}.
\end{equation*}

Due to \cite{IntVertEstimates_VesChir} (see also \cite{IntVerticesSurveyPart1,IntVerticesSurveyPart2}),
\begin{multline}\label{int_vert_bound_delta_ext_intro}
    \abs{\vertex(\PC_I)} \leq (n+1)^{n+1} \cdot n! \cdot \xi(n,k) \cdot \log_2^{n-1}(2 \sqrt{n+1} \cdot \Delta_{ext}) = \\
    = k^{\frac{n}{2}} \cdot O(n)^{\frac{3}{2}n+1.5} \cdot \log^{n-1}(n \cdot \Delta_{ext}),
\end{multline}
} Here $k$ is the number of lines in $A x \leq b$ (in our notation $k = n + m$) and $\Delta_{ext} = \Delta(\bigl(A\,b\bigr))$.

Let $\phi$ be the encoding length of the system $A x \leq b$. Due to \cite[Chapter~3.2, Theorem~3.2]{Schrijver}, we have $\Delta_{ext} \leq 2^\phi$. In notation with $\phi$, the last bound \eqref{int_vert_bound_delta_ext_intro} becomes
\GribanovAdd{
$$
k^{\frac{n}{2}} \cdot O(n)^{\frac{3}{2}n+1.5} \cdot (\phi + \log n)^{n-1},
$$
}which is better, than a more known bound
\begin{equation}\label{int_vert_cook_intro}
    k \cdot \binom{k-1}{n-1} \cdot (5 n^2 \cdot \phi + 1)^{n-1} = k^n \cdot \Omega(n)^{n-1} \cdot \phi^{n-1},
\end{equation} due to W.~Cook \cite{IntVert_Cook}, because \eqref{int_vert_bound_delta_ext_intro} depends on $k$ as $k^{n/2}$ and $k \geq n$.

Due to \cite{IntVerticesSurveyPart2}, we can combine the previous inequality \eqref{int_vert_bound_delta_ext_intro} with the sensitivity result of W.~Cook et~al. \cite{Sensitivity_Tardos} to construct a bound that depends on $\Delta$ instead of $\Delta_{ext}$:
\GribanovAdd{
\begin{multline}\label{int_vert_bound_delta_intro}
\abs{\vertex(\PC_I)} \leq (n+1)^{n+1} \cdot n! \cdot \xi(n,k) \cdot \xi(n,2k) \cdot \log_2^{n-1}(2 \cdot (n+1)^{2.5} \cdot \Delta^2) = \\
= k^n \cdot O(n)^{n+1.5} \cdot \log^{n-1}(n \cdot \Delta),
\end{multline}
} which again in some sense is better than \eqref{int_vert_cook_intro}, because \eqref{int_vert_bound_delta_intro} depends only from the encoding length of $A$, while \eqref{int_vert_cook_intro} depends on the encoding length of $(A\,b)$.

One of our goals is to generalise results of the papers \cite{SupportIPSolutions,NewBoundsForFixedM,ProximityViaSparsity,DiscConvILP,SteinitzILP,DistributionsILP}, presented in beginning of the current Subsection. Additionally, we strengthen known results and give more elegant proofs. For example, in our opinion, the proofs of sparsity and proximity results are more elegant, when we work with problems in the canonical form. On the other hand, for developing dynamic programming ILP algorithms, it is more convenient to work with the \ref{BILP-SF} and \ref{ILP-SF} problems in the standard form.

The special case of the \ref{BILP-SF} and \ref{ILP-SF} problems, when $m = 0$, deserves special interest. For $m = 0$, these problems have only group constraints $G x \equiv g \pmod{S\cdot \ZZ^n}$. Next, we give more general definition of the group minimization problem.
\begin{definition}
Let $\GC$ be a finite Abelian group, $g_0,g_1,\dots,g_n \in \GC$, $c \in \ZZ^n_{\geq 0}$ and $u \in \bigl(\ZZ_{\geq 0} \cup \{+\infty\}\bigr)^n$. \emph{The bounded group minimization problem} can be formulated as follows:
\begin{gather}
    c^\top x \to \min\notag\\
    \begin{cases}
    \sum\limits_{i = 1}^n x_i\,g_i = g_0\\
    0 \leq x \leq u\\
    x \in \ZZ^n.
    \end{cases}\tag{B-GROUP-PROB}\label{bn_group_prob}
\end{gather}

The unbounded case of the \ref{bn_group_prob} problem corresponds to the case $u = +\infty$:
\begin{gather}
    c^\top x \to \min\notag\\
    \begin{cases}
    \sum\limits_{i = 1}^n x_i\,g_i = g_0\\
    x \in \ZZ^n_{\geq 0}.
    \end{cases}\tag{U-GROUP-PROB}\label{un_group_prob}
\end{gather}
\end{definition}

Properties of the \ref{un_group_prob} problem  were described in seminal work \cite{GomoryRelation} of R.~Gomory (see also \cite[Chapter~19]{HuBook}). The second goal of this paper is to consider special cases of the \ref{BILP-CF}, \ref{BILP-SF}, \ref{ILP-CF}, \ref{ILP-SF} problems with $m \in \{0,1\}$. This class of problems includes ILP problems on simplices and parallelepipeds, bounded and unbounded Knapsack problems, bounded and unbounded Subset-Sum problems. 
We study properties of the last problems with respect to their relation to group minimization problems.

\GribanovAdd{Good surveys on other aspects of parametric ILP are given in \cite{FiveMiniatures,AlgorithmicTheoryILP,CountingFixedM}. }

\subsection{Normalization For Problems In The Canonical Form}\label{normalization_subs}

Here, we give some polynomial-time preprocessing procedures that can help to make the \ref{BILP-CF} and \ref{ILP-CF} problems more convenient for analysis.

Let us consider the \ref{BILP-CF} problem:
\begin{gather*}
    c^\top x \to \max\\
    \begin{cases}
    b_l \leq A x \leq b_r\\
    x \in \ZZ^n,
    \end{cases}
\end{gather*}
where $A$ be an integral $(n+m)\times n$ matrix of rank $n$. Let $v$ be an optimal solution of the relaxed LP problem and $\BC$ be the corresponding base, e.g. $v = A_{\BC}^{-1} b_{\BC}$. W.l.o.g. we can assume that $\BC = \intint{n}$, let additionally $\NotBC = \intint{n+m} \setminus \BC$.

Since the HNF can be computed by a polynomial-time algorithm, and since the problem is invariant with respect to any unimodular changes of variables, we can assume that $A_{\BC}$ has already been reduced to the HNF. Using additional permutations of rows and columns, we can continue a transformation of $A_{\BC}$, such that it will have the following triangle form:
\begin{equation} \label{HNF}
A_{\BC} = \begin{pmatrix}
1            & 0                   & \dots         & 0           & 0               & 0      & \dots & 0\\
0            & 1                   & \dots         & 0           & 0              &0       & \dots & 0\\
\hdotsfor{8} \\
0            &        0            & \dots         & 1           & 0           & 0          & \dots & 0\\
A_{s+1\,1}  &   A_{s+2\,2}        & \dots        & A_{s+1\,s}  & A_{s+1\,s+1} & 0           & \dots & 0\\
\hdotsfor{8} \\
A_{n\,1}  &   A_{n\,2}        & \hdotsfor{5}  & A_{n\,n}\\
\end{pmatrix}.
\end{equation}

Here, $s$ is the number of $1$'s on the diagonal. Hence, $A_{i\,i} \geq 2$, for $i \in \intint[(s+1)]{n}$. Let, additionally, $t = n - s$ be the number of diagonal elements that are not equal to $1$, $\Delta = \Delta(A)$, and $\delta = \abs{\det(A_{\BC})}$.

The following properties follow from the HNF structure:
\begin{itemize}
\item[1)] $0 \leq A_{i\,j} < A_{i\,i}$, for any $i \in \intint n$ and $j \in \intint{(i-1)}$,

\item[2)] $\Delta \geq \delta = \prod_{i=s+1}^n A_{i\,i}$, and, hence, $t \leq \log_2(\Delta)$,

\item[3)] since $A_{i\,i} \geq 2$, for $i \in \intint[(s+1)]n$, we have \[\sum\limits_{i=s+1}^n A_{i\,i} \leq \frac{\delta}{2^{t-1}} + 2(t-1) \leq \delta.\]
\end{itemize}

Using integer translations, we can transform $b_{\BC}$, such that $\BZero \leq (b_l)_{\BC} < \diag(A_{\BC})$, so the first $s$ components of $b_{\BC}$ are equal to $0$.
Let $H$ be the sub-matrix of $A_{\BC}$, which is located in the rows of $A_{\BC}$ right after the $s \times s$ identity matrix. In other words, $H$ can be characterized by the following representation of $A$: 
\begin{equation*}
A_{\BC} = \begin{pmatrix}
I_{s \times s} & \BZero_{s \times t}\\
H & T\\
\end{pmatrix},\quad\text{where}\quad
T = \begin{pmatrix}
A_{s+1\,s+1} & 0           & \dots & 0\\
A_{s+2\,s+1} & A_{s+2\,s+2} & \dots & 0\\
\hdotsfor{4}\\
A_{n \, s+1} &   A_{n\,s+2}        & \dots  & A_{n\,n}\\
\end{pmatrix}.
\end{equation*}

Without loss of generality, we can assume that columns of $H$ are lexicographically sorted. Indeed, any permutation of the first $s$ variables of the system $A x \leq b$ can be compensated by a permutation of the first $s$ rows.

\begin{definition}\label{normalization_def}
A system of the \ref{BILP-CF} or \ref{ILP-CF} problems is called \emph{$v$-normalized} or \emph{$\BC$-normalized}, if the matrix $A$ and the vector $b_l$ have the form, described in this Subsection.

\end{definition}

\begin{remark}[The computational complexity of the normalization]\label{norm_complexity_rm}
It can be easily seen that $v$-normalization of a system of the \ref{ILP-CF} problem can be done by a polynomial-time algorithm. Indeed, two most complex parts of the normalization are searching of an optimal solution of the LP relaxation and computing the HNF for $A_{\BC}$.

Due to \cite{Khachiyan}, the computational complexity of the LP problem is polynomial. More efficient algorithms can be found in \cite{GlobalOpt,Karmarkar,IntPointBook}. 
The information about the HNF is given in Subsection \ref{SNF_subs}.
\end{remark}

The following Lemma helps to characterise all the elements of the matrix $A$ for a normalized system. More precisely, it characterises elements of the matrix $A$ with the property that $A_{\BC}$ has already been reduced to the HNF. 
\begin{lemma}[\cite{FPT_Grib}]\label{HNF_elem_lm}
The following inequality holds:

\[\|A_{\NotBC}\|_{\max} \leq \frac{\Delta}{\delta} ( \frac{\delta}{2^{t-1}} + t -1) \leq \Delta.\]

Hence, $\|A\|_{\max} \leq \Delta$.
\end{lemma}

The following Lemma, that will be significantly used in our work, helps to characterise elements of $A_{\BC}^{-1}$ and vectors of the form $y = A_{\BC}^{-1} x$.
\begin{lemma}\label{adj_lm}
The adjugate matrix $A_{\BC}^*$ has the form
$$
\begin{pmatrix}
\delta            & 0                   & \dots         & 0           & 0               & 0      & \dots & 0\\
0            & \delta                   & \dots         & 0           & 0              &0       & \dots & 0\\
\hdotsfor{8} \\
0            &        0            & \dots         & \delta           & 0           & 0          & \dots & 0\\
*  &   *        & \dots        & *  & \delta / A_{s+1\,s+1} & 0           & \dots & 0\\
\hdotsfor{8} \\
*  &   *        & \hdotsfor{5}  & \delta / A_{n\,n}\\
\end{pmatrix}.
$$ More precisely, $(A_{\BC}^*)_{i\,i} = \delta / A_{i\,i}$, $\|A_{\BC}^*\|_{\max} \leq \delta^2/2$, and the first $s$ rows of $A_{\BC}^*$ have the form $\bigl(\delta \cdot I_{s \times s}\; \BZero\bigr)$.

Consequently, let $\|y\|_0 \leq \alpha$, $\|y\|_1 \leq \beta$, for some $y \in \ZZ^n$ and $\alpha,\beta \geq 0$. Then,
$$
\|A_{\BC}^{-1} y\|_0 \leq \alpha + \log_2(\delta), \quad \|A_{\BC}^{-1} y\|_{\infty} \leq \frac{\delta \beta}{2}.
$$
\end{lemma}
The proof of Lemma \ref{adj_lm} could be found in Appendix Section \ref{appendix_sec}.

\subsection{Reductions Between The \ref{BILP-CF} And \ref{BILP-SF} Problems }\label{reductions_subs}

In this Subsection, we describe how pairs of the \ref{BILP-CF} and \ref{BILP-SF} problems or pairs of the \ref{ILP-CF} and \ref{ILP-SF} problems can be polynomially reduced to each other. The following Theorem was proved in \cite{MinimalDistance_Veselov}. Other proofs also can be found in \cite{BlueBook,CountingFixedM,ABCModular}.
\begin{theorem}[\cite{MinimalDistance_Veselov}]\label{perp_matricies_th}
Let $A \in \ZZ^{n \times m}$, $B \in \ZZ^{n \times (n-m)}$, $\rank A = m$, $\rank B = n-m$, and $A^\top B = \BZero$. Then, for any $\BC \subseteq \intint n$, $\abs{\BC} = m$, and $\NotBC = \intint{n} \setminus \BC$, the following equality holds:
$$
\Delta_{\gcd}(B) \cdot \abs{\det (A_{\BC *})} = \Delta_{\gcd}(A) \cdot \abs{\det (B_{\NotBC  *})}.
$$
\end{theorem}

\begin{remark}
The result of this Theorem was strengthened in \cite{MinorsOfOrtMatricies}. Namely, it was shown that the matrices $A$ and $B$ have the same diagonal of their Smith Normal Forms modulo of $\gcd$-like multipliers.
\end{remark}

\begin{lemma}\label{ILPCF_to_ILPSF_lm}
For any instance of the \ref{BILP-CF} problem, there exists an equivalent instance of the \ref{BILP-SF} problem:
\begin{gather}
    \hat c^\top x \to \min \notag\\
    \begin{cases}
    \hat A x = \hat b\\
    G x \equiv g \pmod{S \cdot \ZZ^n}\\
    0 \leq x \leq u\\
    x \in \ZZ^{n + m}.
    \end{cases}
\end{gather}
with $\hat A \in \ZZ^{m \times (n+m)}$, $\rank(\hat A) = m$, $\hat b \in \ZZ^m$, $G \in \ZZ^{n \times (n + m)}$, $g \in \ZZ^n$, $\hat c \in \ZZ^{n+m}$ and $u \in \bigl( \ZZ_{\geq 0} \cup \{+\infty\} \bigr)^{n + m}$.

Moreover, the following properties hold:
\begin{enumerate}
    \item $\hat A \cdot A = \BZero_{m \times n}$, $\Delta(\hat A) = \Delta(A) / \Delta_{\gcd}(A)$;
    \item $\abs{\det(S)} = \Delta_{\gcd}(A)$;
    \item The map $\hat x = b_r - A x$ is a bijection between integer solutions of both problems;
    \item \GribanovAdd{If the original relaxed LP problem is bounded, then we can assume that $\hat c \geq \BZero$.}
\end{enumerate}
\end{lemma}
\begin{proof}

Let $\BC$ be the set of row-indices of some $n \times n$ base sub-matrix of $A$ and $\NotBC = \intint{n+m} \setminus \BC$. 
The original system can be rewritten in the following way:
\begin{gather*}
    - \hat c^\top \hat x \to \max \\
    \begin{cases}
    \hat x = b_r - A x\\
    x \in \ZZ^n \\
    0 \leq \hat x \leq u\\
    \hat x \in \ZZ^{n+m},\\
    \end{cases}
\end{gather*}
where $u = b_r - b_l$, $\hat c_{\BC}^\top = \abs{\det(A_{\BC})} \cdot c^\top A^{-1}_{\BC}$ and $\hat c_{\NotBC} = \BZero$. \GribanovAdd{Here, we set $u_i = +\infty$ iff $(b_l)_i = -\infty$}.

Let $A = P^{-1} \dbinom{S}{\BZero} Q^{-1}$, where $\dbinom{S}{\BZero} \in \ZZ^{(n+m) \times n}$ is the SNF of $A$ and $P^{-1} \in \ZZ^{(n+m) \times (n+m)}$, $Q^{-1} \in \ZZ^{n \times n}$ are unimodular matrices. After the change of variables $x = Q x^\prime$, the system transforms to
\begin{gather*}
    \hat c^\top \hat x \to \min \\
    \begin{cases}
    P_{\BC} \hat x = P_{\BC} (b_r)_{\BC} - S x\\
    P_{\NotBC} \hat x = P_{\NotBC} (b_r)_{\NotBC} - \BZero \\
    x \in \ZZ^n \\
    0 \leq \hat x \leq u\\
    \hat x \in \ZZ^{n+m}.\\
    \end{cases}
\end{gather*}

The last system is equivalent to 
\begin{gather*}
    \hat c^\top \hat x \to \min \\
    \begin{cases}
    P_{\BC} \hat x = P_{\BC} (b_r)_{\BC} \pmod{S \cdot \ZZ^n}\\
    P_{\NotBC} \hat x = P_{\NotBC} (b_r)_{\NotBC} \\
    0 \leq \hat x \leq u\\
    \hat x \in \ZZ^{n+m}.
    \end{cases}
\end{gather*}

Putting $\hat A := P_{\NotBC} \in \ZZ^{m \times (n+m)}$, $\hat b = P_{\NotBC} b_l$, $G := P_{\BC}$ and $g := P_{\BC} b_l$, we transform the original problem to a problem of the type \ref{BILP-SF}.

Now, $\hat A A = P_{\NotBC} P^{-1} \dbinom{S}{\BZero} Q^{-1} = \BZero$ and, by Theorem \ref{perp_matricies_th}, $\Delta(\hat A) = \Delta(A) / \Delta_{\gcd}(A)$. Since, $\dbinom{S}{\BZero}$ is the SNF of $A$, $\abs{\det(S)} = \Delta_{\gcd}(A)$.

Finally, \GribanovAdd{if the original relaxed LP problem is bounded that we can chose $\BC$ to be an optimal base. Then, due to duality, $\hat c \geq \BZero$.}


\end{proof}

\begin{lemma}\label{ILPSF_to_ILPCF_lm}
For any instance of the \ref{BILP-SF} problem with $d = n - m$, there exists an equivalent instance of the \ref{BILP-CF} problem: 
\begin{gather}
    \hat c^\top x \to \max \notag\\
    \begin{cases}
    \hat b_l \leq \hat A x \leq \hat b_r\\
    x \in \ZZ^d
    \end{cases}
\end{gather}
with $\hat A \in \ZZ^{(d+m) \times d}$, $\rank(\hat A) = d$, $\hat c \in \ZZ^d$, $\hat b_r \in \ZZ^{d+m}$ and $\hat b_l \in \bigl(\ZZ \cup \{-\infty\}\bigr)^{d+m}$.

Moreover, the following properties hold:
\begin{enumerate}
    \item $A \cdot \hat A = \BZero_{m \times d}$, $\Delta(\hat A) = \Delta(A) \cdot \abs{\det(S)}$;
    \item $\Delta_{\gcd}(\hat A) = \abs{\det(S)}$;
    \item The map $x = \hat b_r - \hat A \hat x$ is a bijection between integer solutions of both problems.
\end{enumerate}
\end{lemma}

\begin{proof}
We can rewrite the initial problem as follows:
\begin{gather*}
    c^\top x \to \min \\
    \begin{cases}
    \dbinom{A}{G} x = \dbinom{b}{g} + \dbinom{\BZero}{S} \hat x\\
    0 \leq x \leq u\\
    x \in \ZZ^n\\
    \hat x \in \ZZ^d.
    \end{cases}
\end{gather*}
Let $P = \dbinom{A}{G}$, substituting $x = P^{-1} \left(\dbinom{b}{g} + \dbinom{\BZero}{S} \hat x\right)$, the problem becomes:
\begin{gather*}\label{ILPSF_to_ILPCF_lm_representation}
    - \hat c^\top \hat x \to \max \\
    \begin{cases}
    P^{-1}\dbinom{b}{g} - u \leq - P^{-1} \dbinom{\BZero}{S} \hat x \leq P^{-1}\dbinom{b}{g}\\
    \hat x \in \ZZ^d.
    \end{cases}
\end{gather*}
Here, we have eliminated the restriction $x \in \ZZ^n$, because $P$ is unimodular, and this condition follows from the integrality of $\hat x$.

Putting $\hat A := -P^{-1} \dbinom{\BZero}{S}$, $\hat b_l := P^{-1}\dbinom{b}{g} - u$, and $\hat b_r = P^{-1}\dbinom{b}{g}$, we transform the original problem to a problem of the type \ref{BILP-CF}. \GribanovAdd{Here, we set $(b_l)_i = -\infty$ iff $u_i = +\infty$}.

Now, $A \cdot \hat A = A \cdot \dbinom{A}{G}^{-1} \cdot \dbinom{\BZero}{S} = \bigl(I_{m \times m}\, \BZero\bigr) \cdot \dbinom{\BZero}{S} = \BZero$. Since $\dbinom{\BZero}{S}$ is the SNF of $\hat A$, we have $\Delta_{\gcd}(\hat A) = \abs{\det (S)}$. By Theorem \ref{perp_matricies_th}, we have $$\Delta(\hat A) = \Delta(A) \cdot \Delta_{\gcd}(\hat A) = \Delta(A) \cdot \abs{\det (S)}.$$
\end{proof}

\GribanovAdd{
\begin{remark}\label{reduction_impossibility}
Lemmas \ref{ILPCF_to_ILPSF_lm} and \ref{ILPSF_to_ILPCF_lm} give reductions between the \ref{ILP-CF}, \ref{BILP-CF}, \ref{ILP-SF}, \ref{BILP-SF} problems and vise-verse. A weaker result was shown in the paper \cite{CountingFixedM} that any unbounded \ref{CLASSIC-ILP-SF} problem can be reduced to the \ref{ILP-CF} problem. Now, let us show that the reverse reduction from the \ref{ILP-CF} problem to an unbounded \ref{CLASSIC-ILP-SF} problem, preserving the parameters $m$ and $\Delta(A)$, is not always possible.

Let us consider the $2$-dimensional simplex $\PC = \{ x \in \RR^2 \colon 
A x \leq b
\}$, where $A = \begin{pmatrix}
1 & 1 \\
1 & -1 \\
-3 & -1 \\
\end{pmatrix}$ and $b = (1, 1, 1)^\top$. Let $\SC$ be the set of slacks of the system $A x \leq b$, i.e. $\SC = \{y = b - A x \colon x \in \PC \}$. By a direct enumeration, we can see that
\begin{gather*}
\PC \cap \ZZ^2 = \left\{\binom{0}{0}, \binom{0}{1}, \binom{1}{0}, \binom{-1}{2}, \binom{0}{-1}\right\}, \quad\text{and} \\
\SC \cap \ZZ^3 = \left\{ \trinom{1}{1}{1}, \trinom{0}{2}{2}, \trinom{0}{0}{4}, \trinom{0}{4}{0}, \trinom{2}{0}{0} \right\}.
\end{gather*}

Due to Lemma \ref{ILPCF_to_ILPSF_lm}, the set $\SC \cap \ZZ^3$ can be expressed by the following system
\begin{equation}\label{standard_2simplex_system}
\begin{cases}
(2, 1, 1) y = 4\\
(3, 0 , 1) y \equiv 0 \pmod{2}\\
y \in \ZZ^3_{\geq 0}.
\end{cases}
\end{equation}

For the sake of contradiction, assume that there is another system of the type \ref{CLASSIC-ILP-SF}
\begin{equation}\label{classic_2simplex_system}
\begin{cases}
a^\top y = a_0\\
y \in \ZZ^3_{\geq 0},
\end{cases}
\end{equation}
representing the set $\SC \cap \ZZ^3$, where $a \in \ZZ^3$ and $a_0 \in \ZZ$. 
\end{remark}

The plane $a^\top y = a_0$ is equivalent to the plane $(2, 1, 1) y = 4$, because any plane can be uniquely determined by any $3$ affine-independent points. For example, we can choose the points $(0,0,4)^\top, (0,4,0)^\top, (2,0,0)^\top$ from $\SC \cap \ZZ^3$. Hence, we can assume that the eqution $a^\top y = a_0$ coincides with $(2, 1, 1) y = 4$. Finally, note that the system \eqref{classic_2simplex_system} contains integer points that do not belong to $\SC \cap \ZZ^3$. For example, we can choose a point $(0,1,3)^\top$.
}

\section{Description Of The Results And Related Papers}\label{results_sec}

The rest of the paper consist of two parts. The first part presents results about sparsity, proximity, the number of integer vertices and the computational complexity of the \ref{BILP-SF}, \ref{BILP-CF}, \ref{ILP-SF}, and \ref{ILP-CF} problems, parameterized by $m$ and $\Delta = \Delta(A)$. In the second part, we consider special cases of these problems for $m \in \{0,1\}$ and group minimization problems. The case $m = 0$ corresponds to square systems of linear inequalities. For example, the ILP problem on a parallelepiped, defined by a system of linear inequalities, is exactly the \ref{BILP-CF} problem with $m = 0$. Due to \cite[Paragraph~3.3., p. 42--43]{BlueBook}, the last problem is equivalent to the group minimization problem \ref{bn_group_prob}. Moreover, this case is important, because, as it was shown in \cite{IntegralityNumber}, for any fixed $A$ and $c$, and $b$ chosen  randomly, for almost all $b$ the \ref{ILP-SF} problem is equivalent to the problem, induced by some square sub-system of the original system. For the case $m = 1$, we have the following examples: the \ref{ILP-CF} problem with $m=1$ corresponds to the ILP problem on simplices, the \ref{ILP-SF} and \ref{BILP-SF} problems with $m = 1$ correspond to the unbounded and bounded Knapsack problems with some additional group constraints. Using these observations, we give new properties of the ILP problem on simplices, unbounded Knapsack problem,  unbounded Subset-Sum problem. Additionally, we study an average case of the \ref{ILP-CF} problem. 

Next, we give a detailed description of our results. In the following text, we denote by $v$ and $z$ an optimal vertex solution of the LP relaxation and an optimal integer vertex solution of the original problem, respectively. Additionally, we denote $\Delta = \Delta(A)$, $\Delta_1 = \Delta_1(A)$ and $\Delta^* = \Delta \cdot \abs{\det(S)}$.

\subsection{Results Of The First Part Of This Paper}\label{results_part1_subs}

\begin{enumerate}
    \item We generalize the sparsity result of the paper \cite{SupportIPSolutions} for the unbounded \ref{ILP-SF} and \ref{ILP-CF} problems with respect to $m$, $\Delta$, and $\Delta^*$, instead of $m$ and $\Delta_1$. More precisely, we give the following sparsity bounds for any vertex solution $z$ of the  \ref{ILP-CF} and \ref{ILP-SF} problems, respectively:
    \begin{multline*}
    \|b - A z\|_0 \leq m + \log_2 \sqrt{\det(A^\top A)} \quad\text{and}\quad\\ \|z\|_0 \leq m + \log_2 \sqrt{\det(A A^\top)} + \log_2 \abs{\det(S)}.    
    \end{multline*}
    
    These can be translated to the bounds
    \begin{multline*}
    \|b - A z\|_0 \leq c \cdot m + \log_2(\Delta) + \frac{m}{2} \cdot \log_2 \left( \log_2 \sqrt{2 e} + \frac{\log_2(\Delta)}{m} \right) \quad\text{and}\quad\\ 
    \|z\|_0 \leq c \cdot m + \log_2(\Delta^*) + \frac{m}{2} \cdot \log_2 \left( \log_2 \sqrt{2 e} + \frac{\log_2(\Delta^*)}{m} \right),
    \end{multline*}
    where $c = \log_2\sqrt{\frac{2 e^2}{e - \log_2 e}} + \frac{1}{2} \leq 2.27$. See Subsection \ref{sparsity_subs}, Theorem \ref{ILP_CF_sparsity_th} and Corollary \ref{ILP_SF_sparsity_cor}.
    
    Due to Remark \ref{classic_ILP_rm}, the unbounded \ref{CLASSIC-ILP-SF} problem in the standard form corresponds to the case $S = I$. For this case, our bound gives
    \GribanovAdd{
    \begin{multline*}
    \|z\|_0 \leq c \cdot m + \log_2(\Delta) + \frac{m}{2} \cdot \log_2 \left( \log_2 \sqrt{2 e} + \frac{\log_2(\Delta)}{m} \right) = \\
    = O(m + \log(\Delta)),
    \end{multline*}
    }
    which is stronger, than the bound
    $$
    \|z\|_0 \leq 2m \cdot \log_2\sqrt{2m} + 2 \cdot \log_2(\Delta)
    $$ due to \cite{ProximityViaSparsity} (see Remark \ref{ILP_CF_sparsity_rm} for details).
    
    Next, we can use the Hadamard's inequality to achieve the sparsity bound with respect to $m$ and $\Delta_1$ for the unbounded  \ref{CLASSIC-ILP-SF} problem:
    $$
    \|z\|_0 \leq m \cdot \log_2\left( \const \cdot \Delta_1 \cdot \sqrt{m \cdot \log_2\bigl(\Delta \cdot \sqrt{2 e m}\bigr)} \right).
    $$
    The last bound is better in terms of a power of $m$ and $\Delta_1$ inside the logarithm, than the previous state of the art bound
    $$
    \|z\|_0 \leq \frac{3}{2} \cdot m \cdot \log_2\left( \const \cdot \Delta_1 \cdot \sqrt{m} \right)
    $$ due to \cite{NewBoundsForFixedM}. \GribanovAdd{ For $m = 1$, our bound becomes 
    $$
    \|z\|_0 \leq \log_2\left( 2^c \cdot \Delta_1 \cdot \sqrt{\log_2\bigl(\Delta \cdot \sqrt{2 e}\bigr)} \right),
    $$ which is slightly worse, than the corresponding bound 
    $$
        \|z\|_0 \leq \log_2 \left( \sqrt{6} \cdot \Delta_1 \cdot \sqrt{\log_2(2 \Delta_1)} \right),\quad\text{due to the paper \cite{NewBoundsForFixedM}.}
    $$
    }
    
   
    \item Following to the works \cite{SupportIPSolutions} and \cite{NewBoundsForFixedM}, we give bounds for the number of integer vertices for polyhedra of the \ref{ILP-CF} and \ref{ILP-SF} problems with respect to $m$, $\Delta$, and $\Delta^*$, instead of $m$ and $\Delta_1$. More precisely, we give the following bounds for the \ref{ILP-CF} and \ref{ILP-SF} problems, respectively:
    \begin{multline*}
    \abs{\vertex(\PC_I)} = (n+m)^s \cdot O(s)^{s+1} \cdot \log^{s-1}(s \cdot \Delta) \quad\text{and}\quad\\ 
    \abs{\vertex(\PC_I)} = n^s \cdot O(s)^{s+1} \cdot \log^{s-1}(s \cdot \Delta^*).
    \end{multline*}
    Here $\PC_I$ is an integer polyhedron, related to the problem, and $s$ is the sparsity parameter. Due to the sparsity results of this work, we have $s = O(m + \log(\Delta))$ for both problems. Consequently, both bounds can be simplified to
    $$
    \abs{\vertex(\PC_I)} = \bigl(n \cdot m \cdot \log(\Delta) \bigr)^{O(m + \log(\Delta))},
    $$ which is a polynomial for fixed $m$ and $\Delta$ ($\Delta$ must be replaced by $\Delta^*$ for the \ref{ILP-SF} problem). Note that all bounds are constructive. See Subsection \ref{int_vertex_subs}, Theorem \ref{ILPCF_int_vert_num_th} and Corollary \ref{ILPSF_int_vert_num_cor}.
    
    Next, we can use the Hadamard’s inequality to achieve bounds with respect to $\Delta_1$ for the unbounded \ref{CLASSIC-ILP-SF} problem:
    $$
    \abs{\vertex(\PC_I)} = n^s \cdot O(s)^{s+1} \cdot O(m)^{s-1} \cdot \log^{s-1}(m \cdot \Delta_1).
    $$
    The last bound is better in terms of the dependence on $n$ and $\log(m \cdot \Delta_1)$, than the bound
    \begin{multline*}
        \abs{\vertex(\PC_I)} \leq \binom{n}{m} \cdot s \cdot n^s \cdot \log_2^{s}\bigl(m\cdot(2m \Delta_1+1)^m\bigr) = \\
    = n^{m + s} \cdot s \cdot O(m)^{s-m} \cdot \log^{s}(m \cdot \Delta_1),
    \end{multline*}
    due to \cite{NewBoundsForFixedM}.
    
    Noting that $s = O(m \cdot \log(m \cdot \Delta_1))$, both bounds give
    $$
    \abs{\vertex(\PC_I)}= \bigl(n \cdot m \cdot \log(m \cdot \Delta_1) \bigr)^{O(m \cdot \log(m \cdot \Delta_1))},
    $$ which is a polynomial for fixed $m$ and $\Delta_1$.

    \item We generalize the main proximity result of the work \cite{ProximityViaSparsity} for the unbounded \ref{ILP-SF} and \ref{ILP-CF} problems. More precisely, we give the following proximity bounds for any vertex solution of the  \ref{ILP-CF} and \ref{ILP-SF} problems, respectively:
    \begin{multline*}
    \|A(z - v)\|_1 = O(m^2 \cdot \Delta \cdot \log \sqrt[m]{\Delta}) \quad\text{and}\quad\\ 
    \|z - v\|_1 = O(m^2 \cdot \Delta^* \cdot \log \sqrt[m]{\Delta^*}).
    \end{multline*} Here, we mean that for any optimal vertex solution $v$ of the relaxed LP problem there exists an optimal integer vertex solution $z$ that satisfies the given bounds. See Subsection \ref{un_prox_subs}, Lemma \ref{sparsity_via_proximity_lm}, Theorem \ref{ILP_CF_proximity_th} and Corollary \ref{ILP_SF_proximity_cor}.
    
    For the unbounded \ref{CLASSIC-ILP-SF} problem, it gives the same bound as in the paper \cite{ProximityViaSparsity}:
    $$
    \|z - v\|_1 = O(m^2 \cdot \Delta \cdot \log_2 \sqrt[m]{\Delta});
    $$
    
    \item We generalize the main proximity result of the work \cite{SteinitzILP} for the bounded \ref{BILP-SF} and \ref{BILP-CF} problems with respect to $m$, $\Delta$, and $\Delta^*$, instead of $m$ and $\Delta_1$. More precisely, we give the following proximity bounds for any integer vertex solution $z$ of the  \ref{BILP-CF} and \ref{BILP-SF} problems, respectively:
    \GribanovAdd{
    \begin{multline*}
    \|b - A z\|_1 \leq m\cdot(2 m + 1)^m \cdot \Delta \quad\text{and}\quad \\
    \|z-v\|_1 \leq m\cdot(2 m + 1)^m \cdot \Delta^*.    
    \end{multline*}
    }
    Here, we mean that for any optimal vertex solution $v$ of the relaxed LP problem there exists an optimal integer solution $z$ that satisfies the given bounds. See Subsection \ref{bn_prox_subs}, Theorem \ref{BILP_SF_proximity_th} and Corollary \ref{BILP_CF_proximity_cor}. 
    
    \GribanovAdd{
    Substituting $S = I$ it gives the proximity bound, 
    $$
    \|z-v\|_1 \leq m\cdot(2 m + 1)^m \cdot \Delta,
    $$ previously given in the work \cite{ProximityViaSparsity}, for the bounded \ref{CLASSIC-ILP-SF} problem.
    }
    
    \item We generalize the dynamic programming algorithm of the paper \cite{DiscConvILP} for the unbounded \ref{ILP-SF} and \ref{ILP-CF} problems with respect to $m$, $\Delta$, and $\Delta^*$, instead of $m$ and $\Delta_1$. We show the existence of algorithms with the following arithmetic complexity bounds for the \ref{ILP-CF} and \ref{ILP-SF} problems ,  respectively:
    \begin{multline*}
    O(\log m)^{2 m^2} \cdot m^{m+1} \cdot \Delta^2 \cdot \log(\Delta_{\gcd}) \cdot \rho
    \quad\text{and}\quad 
    \\O(\log m)^{2 m^2} \cdot m^{m+1} \cdot (\Delta^*)^2 \cdot \log \abs{\det(S)} \cdot \rho,    
    \end{multline*}
    where $\|z\|_1 \leq (6/5)^\rho$, for some optimal integer solution $z$. As it was noted in \cite{DiscConvILP}, after solving the corresponding relaxed LP problems, we can take $\rho = O(\log(m \Delta))$ or $\rho = O(\log(m \Delta^*))$, respectively.
    
    For fixed $m$, it gives the complexity bounds
    $$
    O\bigl(\Delta^2 \cdot \log(\Delta) \cdot \log(\Delta_{\gcd})\bigr) \quad\text{and}\quad O\bigl((\Delta^*)^2 \cdot \log(\Delta^*) \cdot \log \abs{\det(S)}\bigr),
    $$ where $\Delta_{\gcd} = \Delta_{\gcd}(A)$.
    
    See Subsection \ref{unbounded_complexity_subs}, Theorem \ref{ILP_SF_complexity_th}, Corollary \ref{ILP_SF_complexity_cor} and Corollary \ref{ILP_CF_complexity_cor}.
    
     For the unbounded \ref{CLASSIC-ILP-SF} problem, it gives slightly worse complexity bound, than $O(\Delta^2)$ that is given in the paper \cite{DiscConvILP}, because we are not able to use the $(\min,+)$-convolution trick here.
     
     \item We generalize the dynamic programming algorithm of the paper \cite{SteinitzILP} for the bounded \ref{BILP-SF} and \ref{BILP-CF} problems with respect to the parameters $m$, $\Delta$, and $\Delta^*$, instead of $m$ and $\Delta_1$. We show the existence of algorithms with the following arithmetic complexity bounds for the \ref{BILP-CF} and \ref{BILP-SF} problems, respectively:
     \GribanovAdd{
      \begin{multline*}
      T_{LP} + n \cdot O(\log m)^{m^2} \cdot (\chi + m)^m \cdot \Delta \cdot \bigl( m + \log(\Delta_{\gcd})\bigr)
      \quad\text{and}\quad 
    \\T_{LP} + n \cdot O(\log m)^{m^2} \cdot (\chi + m)^m \cdot \Delta^* \cdot \bigl( m + \log(\Delta_{\gcd})\bigr),   
    \end{multline*}
    }
    where $T_{LP}$ is the complexity of a linear programming algorithm and $\chi$ is a $l_1$-proximity bound.
    
    Taking $\chi$, as it was stated in the previous results, we obtain algorithms with the complexity bounds
    \GribanovAdd{
    \begin{multline*}
      T_{LP} + n \cdot O(\log m)^{m^2} \cdot O(m)^{m^2 + m +1} \cdot \Delta^{m+1} \cdot \log (\Delta_{\gcd})
      \quad\text{and}\quad 
    \\T_{LP} + n \cdot O(\log m)^{m^2} \cdot O(m)^{m^2+m+1} \cdot (\Delta^*)^{m+1} \cdot \log \abs{\det(S)}.   
    \end{multline*}
    }
    
    Due to \cite{LinearLP_Megiddo}, the LP relaxation of the \ref{ILP-SF} problem can be solved by a linear-time algorithm for fixed $m$. Consequently, for fixed $m$, it gives an algorithm with the arithmetic complexity bounds
    \begin{equation*}
      O\bigl(n \cdot \Delta^{m+1} \cdot \log (\Delta_{\gcd}) \bigr)
      \quad\text{and}\quad 
      O\bigl(n \cdot (\Delta^*)^{m+1} \cdot \log \abs{\det(S)} \bigr)
    \end{equation*} for the \ref{BILP-CF} and \ref{BILP-SF} problems.
    
    \GribanovAdd{
    For the unbounded \ref{CLASSIC-ILP-SF} problem with respect to $\Delta_1$, it gives the complexity bound 
    $$
    T_{LP} + n \cdot O(\log m)^{m^2} \cdot O(m)^{m^2 + m +1} \cdot (\Delta_1)^{m(m+1)},
    $$ that is slightly better, than the state of the art bound due to the work \cite{SteinitzILP}    
    $$
    T_{LP} + n \cdot O(\log m)^{m^2} \cdot O(m)^{(m+1)^2} \cdot (\Delta_1)^{m(m+1)} \cdot \log^2(m \cdot \Delta_1);
    $$
    }
    
    See Subsection \ref{bounded_complexity_subs}, Theorem \ref{BILP_SF_complexity_th} and Corollary \ref{BILP_CF_complexity_cor}. 
    
    \GribanovAdd{
        Additionally, our results can be applied to the bounded $m$-dimensional $\Delta$-modular Knapsack problem:
        \begin{gather*}
            c^\top x \to \max \\
            \begin{cases}
            A x \leq b \\
            \BZero \leq x \leq u\\
            x \in \ZZ^n,
            \end{cases}
        \end{gather*}
        where $A \in \ZZ_{\geq 0}^{m \times n}$, $b \in \ZZ^m_{\geq 0}$, $c,u \in \ZZ^n_{\geq 0}$.
        
        For fixed $m$, it gives $O(n \cdot \Delta^{m+1})$-arithmetic complexity algorithm for this problem. Surprisingly, dependence on $\Delta$ can be significantly reduced for an approximate version of the problem. More precisely, due to \cite{MultyKnapsack_Grib}, the problem admits an FPTAS with the arithmetic complexity bound ($m$ being fixed):
        $$
        O\bigl(n\cdot(1/\varepsilon)^{m+3} \cdot \Delta \bigr).
        $$
    }
    
    \item Taking $m = 1$ and $S = I$ in the previous result, we obtain an algorithm with the arithmetic complexity bound
    $$
    O(n \cdot \Delta^2)
    $$ for the bounded Knapsack problem. Here, $\Delta$ is the maximal weight of an item. It outperforms the state of the art bound $O(n \cdot \Delta^2 \cdot \log^2(\Delta))$, given in the paper \cite{SteinitzILP}.
    
    \GribanovAdd{
    Very recently, an $O(n + \Delta^3)$-complexity algorithm was given in \cite{KnapsackSubsetSum_SmallItems}, which outperforms our algorithm for $\Delta = o(n)$.
    }
    
    \item Finally, as a reader could see, our sparsity and proximity bounds for the \ref{BILP-CF} and \ref{ILP-CF} problems in the canonical form are stated for slack variables, e.g. all the bounds looks like $\|b - A z\| \leq \dots$ or $\|A(z - v)\| \leq \dots$. However, if we additionally assume that the problems are normalised (the normalization can be done by a polynomial-time algorithm, see Subsection \ref{normalization_subs}), then we can give direct bounds for $\|z\|$ and $\|z - v\|$.
    
    \GribanovAdd{
    More precisely, for the \ref{ILP-CF} problem, we have the following sparsity and proximity bounds:
    \begin{gather*}
    \|z\|_0 \leq \|b - A z\|_0 + \log_2(\delta),\\
    \|z - v\|_\infty = O\bigl(m^2 \cdot \Delta \cdot \delta \cdot \log \sqrt[m]{\Delta}\bigr),\\
    \|z\|_\infty = O\bigl(m^2 \cdot \Delta \cdot \delta \cdot \log \sqrt[m]{\Delta}\bigr).
    \end{gather*}
    Here, $\|b - A z\|_0$ can be taken from the result 1 (we skip it to make the formula shorter) and $\delta = \abs{\det(A_{\BC})}$, where the system $A x \leq b$ is assumed to be $\BC$-normalized.
    
    Similarly, for the bounded \ref{BILP-CF} problem, we have:
    \begin{gather*}
        \|z - v\|_\infty \leq m \cdot (2 m + 1)^m \cdot \Delta \cdot \delta,\\
        \|z\|_\infty \leq m \cdot (2 m + 1)^m \cdot \Delta \cdot \delta.
    \end{gather*}
    }
    
\end{enumerate}

\subsection{Results Of The Second Part Of This Paper}\label{results_part2_subs}

    \begin{enumerate}
        \item Taking $m = 0$ in formulation of Theorem \ref{BILP_SF_complexity_th},
        we give an algorithm for the \ref{bn_group_prob} problem with the group-operation complexity bound
        $$
        O(n \cdot \Delta \cdot \log(\Delta)).
        $$ Here, we denote $\Delta = \abs{\GC}$ and assume that $n \leq \Delta$. Additionally, due to Theorem \ref{BILP_SF_proximity_th}, we have that the \ref{bn_group_prob} problem has an optimal solution $z^*$ with $\|z^*\|_1 \leq \Delta - 1$.
        
        \GribanovAdd{
        As a straightforward application of this result, together with the main result of the paper \cite{WidthConv_Grib}, it gives an $O\Bigl(n^4 \cdot \Delta \cdot \bigl(n + \Delta \cdot \log(\Delta)\bigr)  \Bigr)$-arithmetical complexity algorithm to compute the width of an empty lattice simplex, defined by a convex hull of integer points.
        }
        
        Due to Theorem \ref{cyclic_group_complexity_th}, for the \ref{un_group_prob} problem with a cyclic Abelian group, we use the $(\min,+)$-convolution trick from the paper \cite{DiscConvILP} to give an algorithm with the group-operation complexity bound
        $$
        O\bigl(T_{(\min,+)}(2 \Delta) \cdot \log(\Delta) + \Delta \cdot \log(\Delta)  \bigr).
        $$ Due to \cite{3SumViaAdditioveComb} and \cite{APSPViaCircuitComplexity}, it is known that 
        $$
        T_{(\min,+)}(n) = \frac{n^2}{2^{\Omega(\sqrt{\log n})}}.
        $$ Consequently, our complexity bound becomes
        $$
        \Delta^2 \cdot \frac{1}{2^{\Omega(\sqrt{\log(\Delta)})}},
        $$ which is better, than the classical bound $O(n \cdot  \Delta \cdot \log (\Delta))$, due to \cite{GomoryRelation} (see also \cite[Chapter~19]{HuBook}), for example, when $n = \Theta(\Delta)$;

        \item We generalize the result of R.~Gomory \cite{GomoryRelation} that characterises vertices of the \ref{un_group_prob} problem's polyhedron. Now, it works with faces of any dimension. 
        
        More precisely, let $z \in \ZZ^n_{\geq 0}$ be any vertex of the \ref{un_group_prob} problem's polyhedra, denoted by $\PC$. Due to \cite{GomoryRelation},
        $$
        \prod_{i=1}^n (1 + z_i) \leq \abs{\GC},
        $$ where $\GC$ is the corresponding Abelian group. Consequently, $\|z\|_0 \leq \log_2(\Delta)$.
        
        Due to Theorem \ref{faces_Gomory_th}, we show that if $z \in \FC \cap \ZZ^n$, for a $d$-dimensional face of $\PC$, then there exists a set of indices $\JC \subseteq \intint n$, such that $\abs{\JC} \geq n - d$ and
        $$
        \prod_{i \in \JC} (1 + z_i) \leq \abs{\GC}.
        $$ Consequently, $\|z\|_0 \leq d + \log_2(\Delta)$.

        \item In Subsection \ref{group_subs}, we study properties of unbounded ILP problems with square matrices. In other words, we study the \ref{ILP-CF} problem with $m = 0$. Due to \cite[Paragraph~3.3, p. 42--43]{BlueBook}, this problem corresponds to the group minimization problem \ref{un_group_prob}. Consequently, many properties of such problems can be easily deduced from the seminal work \cite{GomoryRelation} of R.~Gomory. 
        
        By this way, we show the following fact (see Theorem \ref{simplex_feasibility_th} of Subsection \ref{group_app_subs}). Let us consider any $n$-dimensional simplex $\PC = \PC(A,b)$ with $A \in \ZZ^{(n+1) \times n}$, $b \in \ZZ^{n+1}$, and $\delta$ be the minimal sub-determinant of $A$, taken by an absolute value. Let $\PC_I = \PC_I(A,b)$. Then, for any vertex $v_{\BC}$ of the simplex corresponding to a base $\BC$, there exists a vertex $z$ of $\PC_I$ with the properties 
        $$
        \|b - A z\|_0 \leq 1 + \log_2(\Delta_{\BC}),\quad \|A(z-v_{\BC})\|_1 \leq \Delta_{\BC} + \Delta - 2, \quad\|A(z - v_{\BC})\|_\infty \leq \Delta_{\BC} - 1,
        $$ and $z$ lies on a $d$-dimensional face of $\PC$ with $d \leq \log_2(\Delta_{\BC})$. Here, $\Delta_{\BC} = \abs{\det(A_{\BC})}$ and $\Delta = \Delta(A)$.
        
        Moreover, an integer solution $z$ (not necessarily optimal) satisfying the given inequalities can be found by an algorithm with the arithmetic complexity bound
        $$
        O\bigl(n \cdot  \Delta_{\BC} \cdot \log ( \Delta_{\BC})\bigr);
        $$
        
        Consequently, there exists an integer vertex $z$ such that
        $$
        \|b - A z\|_0 \leq 1 + \log_2(\Delta_{\min}),\quad \|A(z-v)\|_1 \leq \Delta_{\min} + \Delta - 2, \quad\|A(z - v_{\BC})\|_\infty \leq \Delta_{\min} - 1,
        $$ where $\Delta_{\min}$ is minimal by an absolute value $n \times n$ minor of $A$ and $v$ is some vertex.
        
        Let us consider the unbounded Knapsack problem
        \begin{gather*}\label{unbounded_knapsack_prob}
            c^\top x \to \max\\
            \begin{cases}
            w^\top x = W\\
            x \in \ZZ_{\geq 0}^n,
            \end{cases}
        \end{gather*}
        where $c,w \in \ZZ^n_{>0}$ and $W \in \ZZ_{> 0}$.

        For a unbounded Knapsack problem polyhedron, it guarantees (see Theorem \ref{subset_sum_th}) that for any vertex $v$ of the relaxation there exists an integer vertex $z$ with the properties
        $$
        \|z\|_0 \leq 1 + \log_2 (w_{\max}), \quad \|z-v\|_1 \leq 2\cdot(w_{\max} - 1),\quad \|z - v\|_\infty \leq w_{\max} - 1
        $$ that is equivalent to the main result of the paper \cite{DistancesKnap}.
        
        Additionally, there exists a vertex $z$, such that 
        $$
        \|z\|_0 \leq 1 + \log_2 (w_{\min}), \quad \|z-v\|_1 \leq w_{\min} + w_{\max} - 2,\quad \|z - v\|_\infty \leq w_{\min} - 1.
        $$ Here, $v$ is some vertex of the relaxation, $w_{\min}$ and $w_{\max}$ are the minimal and maximal weights of items, respectively. 
        
        The last sparsity fact is known for Knapsack problem polyhedra from the paper \cite{SparceSemigroups}, but $z$ is not assumed to be a vertex of $\PC_I$ in \cite{SparceSemigroups}. Our proof is very simple, it uses a very basic facts from theory of polyhedra and completely based of results of R.~Gomory \cite{GomoryRelation}. Recently, a very close result that is also based on the work of R.~Gomory \cite{GomoryRelation} has been given in \cite{DistanceSparsityTransference}, but it again considers a Knapsack problem polyhedron instead of a general $\Delta$-modular simplex.
        
        \item \GribanovAdd{ 
        We treat the unbounded Subset-Sum problem as the integer feasibility problem for the unbounded Knapsack problem \eqref{unbounded_knapsack_prob}. Due to \cite{FastSimpleCoinProblem}, the unbounded Subset-Sum problem can be solved with $O(n \cdot w_{\min})$-complexity algorithm.
        
        Using a very simple reduction of the Subset-Sum problem to the cyclic group minimization problem, due to the previous results, we show that the Subset-Sum problem can be solved by an algorithm with the arithmetic complexity bound:
        $$
        w^2_{\min} \cdot \frac{1}{2^{\Omega(\sqrt{\log(w_{\min)}})}}.
        $$ 
        For any constant $c$ and $n = \Omega \bigl( \frac{w_{\min}}{\log^c(w_{\min})} \bigr)$, it outperforms the $O(n \cdot w_{\min})$-complexity algorithm, due to \cite{FastSimpleCoinProblem}, see Theorem \ref{subset_sum_th} of Subsection \ref{group_app_subs}.
        
        Very recently, the same result for the more general All-Targets and Frobenius problems has been obtained in the work \cite{CoinProblem_FineGrained}. Moreover, the algorithm from \cite{CoinProblem_FineGrained} uses only $O\bigl(\log \log(w_{\min})\bigr)$ calls to the $(\min,+)$-convolution oracle, when we use $O\bigl(\log(w_{\min})\bigr)$ calls. So, the algorithm from \cite{CoinProblem_FineGrained} is faster. However, our approach is different to the approach from \cite{CoinProblem_FineGrained} and we use a reduction to the more general cyclic group minimization problem that also can be used to solve the All-Targets and Frobenius problems with the same complexity bounds.
        
        Consider now the unbounded knapsack problem \eqref{unbounded_knapsack_prob}. Let $w_{opt}$ be the weight of an item with the optimal relative cost:
        $$
          \frac{c_j}{w_j} = \min_{i \in \intint n} \left\{ \frac{c_i}{w_i} \right\} \quad\text{and}\quad w_{opt} := w_j,
        $$
        or, in other words, $w_{opt}$ corresponds to the optimal solution of the relaxed LP Knapsack problem.
        
        Using the same ideas, we show that if $W \geq w^2_{opt}$, then the unbounded Knapsack problem can be solved by algorithms with the arithmetic complexity bounds
        $$
        O(n \cdot w_{opt}) \quad\text{and}\quad w^2_{opt} \cdot \frac{1}{2^{\Omega(\sqrt{\log(w_{opt})})}}.
        $$
        
        Due to \cite{KnapsackDPTrick}, the Knapsack problem can be solved by an $O(n \cdot W)$-complexity algorithm. Consequently, the unbounded Knapsack problem can be solved by an algorithm with the arithmetic complexity bound $O(n \cdot w^2_{opt})$.

        For $w_{opt} = o(w_{\max}/\sqrt{n})$, it outperforms the $O(w^2_{\max})$ current state of the art algorithm, due to \cite{DiscConvILP}, see Theorem \ref{knapsack_th} of Subsection \eqref{group_app_subs}.
        
        As an application of the previous result we can consider a variant of the unbounded Subset-Sum problem, when we need to maximize the number of coins used, which coincides with $l_1$-norm of a solution. It corresponds to the unbounded  knapsack problem \eqref{unbounded_knapsack_prob} with $c := \BUnit_n$. Since $w_{opt} = w_{\min}$ for the considered case, due to our previous result, this problem can be solved with $O(n \cdot w_{\min}^2)$-arithmetic complexity algorithm. Other $l_p$-norms, for $p \geq 2$, can also be considered to produce algorithms, which complexity depends on the parameters $n$ and $w_{\min}$. 
        }
        
        \item An $n$-dimensional simplex $\SC$ is called \emph{$\Delta$-modular} if it can be defined by a system $A x \leq b$ with $A \in \ZZ^{(n+1) \times n}$, $b \in \ZZ^{n}$, and $\Delta = \Delta(A)$. We call a simplex $\SC$ \emph{empty} if $\SC \cap \ZZ^n = \emptyset$.
        
        In Theorem \ref{empty_number_th} of Subsection \ref{empty_number_subs}, we show that for fixed $\Delta$ the number of empty $n$-dimensional $\Delta$-modular simplices modulo unimodular linear transformations and integer translations is bounded by a polynomial on $n$ of degree $\Delta-1$.
        
        More precisely, assuming $\Delta \leq n$, we show that this number is bounded by
        $$
        O\left(\frac{n}{\Delta}\right)^{\Delta - 1} \cdot \Delta^{\log_2(\Delta) + 2}.
        $$
        
        \item Our last result, given in Theorem \ref{ILP_CF_exp_th} of Subsection \ref{average_subs}, is about an average case of the \ref{ILP-CF} problem. More precisely, let us fix a matrix $A \in \ZZ^{(n+m) \times n}$ of rank $n$, choose random $b \in \ZZ^{n+m}$, and consider the polyhedron $\PC(b) = \PC(A,b)$, parameterized by $b$. Let, additionally, $\PC_I(b) = \PC_I(A,b)$.  It was shown in \cite{IntegralityNumber} that for almost all $b$ the following proposition is true: for any vertex $z$ of $\PC_I(b)$, there exists a feasible LP base $\BC$, such that $z$ is an integer vertex $z$ of the corner polyhedron $\PC(A_{\BC}, b_{\BC})$.
        
        We give a slightly improved probability analysis, than in \cite{IntegralityNumber}, and consider only the set of $b$ that correspond to nonempty $\PC(b)$. 
        
        Since the system $A_{\BC} x \leq b_{\BC}$ is square, due to \cite[Paragraph~3.3, p.~42--43]{BlueBook}, its solutions correspond to the group minimization problem \ref{un_group_prob}, described by R.~Gomory \cite{GomoryRelation} (see Section \ref{partial_cases_sec} for details). For example, for almost all $b$, it immediately gives an integer linear programming algorithm for the \ref{ILP-CF} problem with the arithmetic complexity bound
        $$
        O(n \cdot \Delta \cdot \log (\Delta) ),
        $$ where $\Delta = \Delta(A)$.
        
        Additionally, for almost all $b$, it follows that for any vertex $z$ of $\PC_I(b)$ there exists a feasible LP base $\BC$, such that:
        \begin{align*}
            &\|b - A z\|_0 \leq m + \log_2(\Delta_{\BC}), & \|A(z - v_{\BC})\|_0 \leq m + \log_2(\Delta_{\BC})\\
            &\|A(z - v_{\BC})\|_1 \leq (\Delta_{\BC} - 1) + m \cdot(\Delta - 1),& \|A(z - v_{\BC})\|_\infty \leq \Delta - 1,
        \end{align*}
        where $v_{\BC} = A_{\BC}^{-1} b_{\BC}$ and $\Delta_{\BC} = \abs{\det(A_{\BC})}$. Consequently, we can conclude that any vertex $z$ of $\PC_I(b)$ lies on a $d$-dimensional face of $\PC(b)$ with $d \leq \log_2( \Delta )$. If we normalize the system $A x \leq b$, then we can achieve similar bounds directly for $z$ and $z - v$.
        
        Due to Theorem \ref{ILP_SF_exp_th}, we can formulate a similar result for the \ref{ILP-SF} problem. For almost all the r.h.s. of the \ref{ILP-SF} problem, it gives:
         \begin{align*}
            &\|z\|_0 \leq m + \log_2(\Delta^*_{\BC}), & \|z - v_{\BC}\|_0 \leq m + \log_2(\Delta^*_{\BC})\\
            &\|z - v_{\BC}\|_1 \leq (\Delta^*_{\BC} - 1) + m \cdot(\Delta^* - 1),& \|z - v_{\BC}\|_\infty \leq \Delta^* - 1,
        \end{align*}
        where $\Delta^*_{\BC} = \Delta_{\BC} \cdot \abs{\det(S)}$ and $\Delta^* = \Delta \cdot \abs{\det(S)}$.
        
        The same bounds for the unbounded \ref{CLASSIC-ILP-SF} problem are given in \cite{DistributionsILP}. Moreover, the paper \cite{DistributionsILP} gives probability distributions for the proximity and sparsity bounds.
        
    \end{enumerate}

\subsection{Other Related Papers}\label{results_other_subs}

Here, we list some other results that, in our opinion, are related to the topic under consideration.

There are known some cases, when the \ref{ILP-CF} problem can be solved with a polynomial-time algorithm. It is well-known that if $\Delta(A) = 1$, then any optimal solution of the corresponding LP problem is integer. Hence, the \ref{ILP-CF} problem can be solved with any polynomial-time LP algorithm  (like in \cite{Khachiyan,Karmarkar,GlobalOpt,IntPointBook}).

The next natural step is to consider the \emph{bimodular} case, i.e. $\Delta(A) \leq 2$. The first paper that discovers fundamental properties of the bimodular ILP problem is \cite{BimodularVert}. Recently, using results of \cite{BimodularVert}, a strong polynomial-time solvability of the bimodular ILP problem was proved in \cite{BimodularStrong}. Very recently, this result was generalised in \cite{ABCModular}. More precisely, it was shown that if the set of all the minors of $A$ consists only from at most $3$ non-zero different absolute values, then a corresponding ILP can be solved by a polynomial-time algorithm. 

Unfortunately, not much is known about the computational complexity of \ref{ILP-CF}, for $\Delta(A) \geq 3$. V.~N.~Shevchenko \cite{BlueBook} conjectured that, for each fixed $\Delta = \Delta(A)$, the \ref{ILP-CF} problem can be solved with a polynomial-time algorithm. There are variants of this conjecture, where the augmented matrices $\binom{c^\top}{A}$ and $\bigl(A \, b\bigr)$ are considered \cite{AZ,BlueBook}. Very recently, it was shown in \cite{TwoNonZerosStrong} that the problem \ref{ILP-CF} is polynomial-time solvable for any fixed $\Delta$ if the matrix $A$ has at most $2$ non-zeros per row or per column. This result was obtained by the reduction of such integer linear programs to the independent set problem with small odd-cycle packing number. Previously, a weaker result, based on the same reduction, was known due to \cite{AZ}. It states that any integer linear program \ref{ILP-CF} with a $\{0,1\}$-matrix $A$, which has at most two non-zeros per row and a fixed value of $\Delta(\binom{\BUnit^\top}{A})$, can be solved by a linear-time algorithm. Analogues results were proven in \cite{DSet_Grib,ThreeGraph_Grib}, where it was shown that ILP formulations for the vertex-dominating and edge-dominating set problems are polynomial-time solvable if $\Delta(A)$ is fixed. 

Additionally,  we  note  that  due  to \cite{SubdeterminantsDiameter} there  are  no  polynomial-time  algorithms  for  the \ref{ILP-CF} problems with $\Delta(A) = \Omega(n^\varepsilon)$, for any $\varepsilon > 0$, unless  the ETH (the Exponential Time Hypothesis) is false. Consequently, with the same assumption, there are no algorithms for the $\Delta$-modular ILP problem \ref{ILP-CF} with the complexity bound $\poly(s)\cdot \Delta^{O(1)}$, where $s$ is an input size. Despite the fact that algorithms with complexities $\poly(s) \cdot \Delta^{f(\Delta)}$ or $s^{f(\Delta)}$ may still exist, it is interesting to consider existence of algorithms with a polynomial dependence on $\Delta$ in their complexities for some special cases of the $\Delta$-modular ILP problem. It is exactly what we do in the current paper while fixing the parameter $m$ in ILP formulations of our problems.

F.~Eisenbrand and S.~Vempala \cite{RandomEdgeLP} presented a randomized simplex-type linear programming algorithm, whose expected running time is strongly polynomial if all the minors of the constraint matrix are bounded in the absolute value by a fixed constant. Due to E.~Tardos results \cite{Strong_Tardos}, linear programs with the constraint matrices, whose all the minors are bounded in the absolute value by a fixed constant, can be solved in strongly polynomial time. N.~Bonifas et al. \cite{SubdeterminantsDiameter} showed that any polyhedron, defined by a totally $\Delta$-modular matrix (i.e., a matrix, whose all minors are at most $\Delta$ by an absolute value), has a diameter, bounded by a polynomial in $\Delta$ and the number of variables.

In \cite{WidthProceedings_Grib,Width_Grib}, it was shown that any lattice-free polyhedron $IP_{\leq}(A,b)$ has a relatively small width, i.e. the width is bounded by a function that is linear in the dimension and exponential in $\Delta(A)$. Interestingly, due to \cite{Width_Grib}, the width of any empty lattice simplex, defined by a system $A x \leq b$, can be estimated by $\Delta(A)$.
In \cite{WidthSimplex_Grib}, it has been shown that  the width of such simplices can be computed with a polynomial-time algorithm. In \cite{FPT_Grib}, for this problem, an FPT-algorithm was proposed. In \cite{WidthConv_Grib}, a similar FPT-algorithm was given for simplices, defined by the convex hull of columns of $\Delta$-modular matrices. We note that, due to \cite{IntroEmptyLatticeSimplicies}, this problem is NP-hard in the general case.

In the case, when the parameter $\Delta$ is not fixed and the dimension parameter $n$ is fixed, the problems \ref{ILP-CF} and \ref{ILP-SF} can be solved with a polynomial-time algorithm, due to the known work of Lenstra \cite{Lenstra}. A similar result for general convex sets, defined by separation hyperplane oracle, is presented in \cite{DadushFDim,DadushDis}. Wider families of sets, induced by the classes of convic and discrete convic functions, defined by the comparison oracle, are considered in the papers \cite{Convic,DConvic}. In \cite{ConvicComp}, it was shown that the integer optimization problem in such classes of sets can be solved with an algorithm having the oracle complexity $O(n)^n$, which meets the same complexity bound as in \cite{DadushFDim,DadushDis}. An algorithm with an close to optimal oracle complexity bound for the $2$-dimensional case is given in \cite{Convic2}. 

An interesting theory on demarcation of polynomial-time solvability and NP-completeness for graph problems is presented in \cite{Malyshev1,Malyshev2,Malyshev3,Malyshev4,Malyshev5}.

\section{Sparsity, Proximity, and Integer Vertices Bounds}\label{sparsity_and_proximity_sec}

\subsection{Sparsity Bounds for the \ref{ILP-CF} and \ref{ILP-SF} Problems}\label{sparsity_subs}

Firstly, we are going to present our sparsity bounds for the unbounded \ref{ILP-CF} and \ref{ILP-SF} problems. To do that we need the following technical lemmas.

\begin{lemma}\label{tight_delta_ineq}
Let $n,m,\Delta \geq 1$ be integer values that satisfy the inequality
$$
n \leq m + \frac{1}{2} m \cdot \log_2 \left( \frac{e \cdot n}{m} \right) + \log_2(\Delta),
$$ then
$$
n \leq c \cdot m + \log_2(\Delta) + \frac{m}{2} \cdot \log_2 \left(\log_2\sqrt{2 e} + \frac{\log_2(\Delta)}{m}\right),
$$ where 
$$ 
c =\log_2 \sqrt{\frac{2 e^2}{e - \log_2 e}} + \frac{1}{2} \leq 2.27.
$$
\end{lemma}
The proof of Lemma \ref{tight_delta_ineq} could be found in Appendix Section \ref{appendix_sec}.

\begin{lemma}\label{rough_delta_ineq}
Let $c_1,c_2 > 0$ and $n,m,\Delta \geq 1$ be integer values that satisfy the inequality
$$
n \leq c_1 \cdot m + \log_2(\Delta) + \frac{m}{2} \cdot \log_2 \left( c_2 + \frac{\log_2(\Delta)}{m}\right).
$$ Then, for any $k > 0$, we have
\begin{gather*}
    n \leq c_1^\prime \cdot m + c_2^\prime \cdot \log_2(\Delta),\text{ where}\\
    c_1^\prime = \left(c_1 + \log_2\sqrt{k + c_2} - (\ln 4)^{-1}\right) \quad\text{and}\quad
    c_2^\prime = \left(1 + \frac{1}{(k + c_2) \cdot \ln 4}\right).
\end{gather*}
\end{lemma}
The proof of Lemma \ref{rough_delta_ineq} could be found in Appendix Section \ref{appendix_sec}.

Now, we state one of the main results of the current Subsection.
\begin{theorem}\label{ILP_CF_sparsity_th}
Let $v$ be any integral vertex solution of the \ref{ILP-CF} problem, and let $u$ be the corresponding slack vector, e.g. $A v + u = b$, let $\Delta = \Delta(A)$. Then, the following inequalities hold 
\begin{enumerate}
    \item $\|u\|_0 \leq m + \log_2 (\det \inth(A))$;
    \item $\|u\|_0 \leq c \cdot m + \log_2(\Delta) + \frac{m}{2} \cdot \log_2 \left( \log_2 \sqrt{2 e} + \frac{\log_2(\Delta)}{m} \right)$, \\where $c = \log_2 \sqrt{\frac{2e^2}{e - \log_2 e}} + \frac{1}{2} \leq 2.27$.
\end{enumerate}

Additionally, let the problem be normalized, and let $\delta = \abs{\det(A_{\BC})}$. Then, 
$$
\|v\|_0 \leq \|u\|_0 + \log_2(\delta).
$$
\end{theorem}

\begin{proof}
Let $\SC = \supp(u)$ and $\ZC = \zeros(u)$. We have

$$
\begin{pmatrix}
A_{\ZC} \\
A_{\SC}
\end{pmatrix} v + \begin{pmatrix}
\BZero \\
u_S
\end{pmatrix} = \begin{pmatrix}
b_{\ZC} \\
b_{\SC}
\end{pmatrix}.
$$

There exists a unimodular matrix $Q \in \ZZ^{n \times n}$, such that $A_{\ZC} = \bigl(H \, \BZero\bigr) Q$, where $\bigl(H \, \BZero\bigr)$ is the HNF of $A_{\ZC}$ and $H \in \ZZ^{\abs{\ZC} \times r}$, where $r = \rank(H) = \rank(A_{\ZC})$. Let $y = Q v$, then
\begin{equation}\label{zeros_elimination_eq}
\begin{pmatrix}
H & \BZero \\
C & B
\end{pmatrix} y + \begin{pmatrix}
\BZero \\
u_{\SC}
\end{pmatrix} = \begin{pmatrix}
b_{\ZC} \\
b_{\SC}
\end{pmatrix},    
\end{equation}
where $(C \, B) = A_{\SC} Q^{-1}$ and $B \in \ZZ^{\abs{\SC} \times (n-r)}$. Note that $B$ has the full column rank and $\Delta(B) \leq \Delta$. 
Let us consider the last $\abs{\SC}$ equalities of the previous system:
\begin{equation}\label{zeros_eliminated_eq}
B z + u_{\SC} = b_{\SC} -  C y_{\intint r} = b_{\SC} - C H^{-1}_{\JC *} b_{\JC},    
\end{equation}
where $z = y_{\intint[r+1]{n}}$ is composed from the last $n-r$ components of $y$ and $\JC \subseteq \ZC$ is a set of indices such that a $r\times r$ sub-matrix $H_{\JC *}$ is non-degenerate..

Let $\lambda_1, \lambda_2, \dots, \lambda_{(n-r)}$ be the successive minima of $\inth(B)$. Due to the Minkowski's Second Theorem, see \cite[Paragraph~9, p.~59]{NumGeom_Lekk},
$$
\prod_{i = 1}^{n - r} \lambda_i \leq 2^{n - r} \frac{\det \inth(B) }{\vol(\BB_{\infty} \cap \linh(B))},
$$ where $\BB_{\infty}$ is the unit ball with respect to the $l_\infty$-norm. Due to the paper \cite{CubeSectionVol}, 
$$
\vol(\frac{1}{2}\BB_{\infty} \cap \linh(B)) \geq 1, \quad\text{so}\quad \vol(\BB_{\infty} \cap \linh(B)) \geq 2^{n-r},
$$ and
$$
\prod_{i = 1}^{n - r} \lambda_i \leq \det \inth(B), \quad \text{and} \quad \lambda_1 \leq (\det \inth(B))^{1/(n-r)}.
$$

Let us show that $n-r \leq \log_2 (\det \inth(B))$. Definitely, in the opposite case we have $n - r > \log_2 (\det \inth(B))$. Then, there exists a vector $t \in \ZZ^{n - r} \setminus \{0\}$, such that $\|B t\|_{\infty} \leq  1$. Hence, the vectors $z \pm t$ with slacks $u_{\SC} \pm Bt$ are feasible for the system \eqref{zeros_eliminated_eq}. The last fact contradicts that $v$ is a vertex.

Now, it can be seen that $\abs{\SC} \leq m + \log_2 (\det \inth(B))$. In the opposite case, $\abs{\SC} > m + \log_2 (\det \inth(B))$ and $\abs{\ZC} < n - \log_2 (\det \inth(B))$. But, $$n - r \geq n - \abs{\ZC} > \log_2 (\det \inth(B)),$$ that is contradiction. 

Finally, due to \eqref{zeros_elimination_eq}, we have $\det \inth(B) \leq \det \inth(A)$ that proves the first inequality. Using the Cauchy-Binet formula, we have 
$$
\abs{\SC} \leq m + \frac{1}{2} \log_2 \binom{\abs{\SC}}{n-r} + \log_2(\Delta).
$$

Now, $\abs{\SC} - n + r \leq \abs{\SC} + \abs{\ZC} - n \leq m$. Assuming $m \leq \abs{\SC}/2$,
$$
\abs{\SC} \leq m + \frac{1}{2} \log_2 \binom{\abs{\SC}}{m} + \log_2(\Delta) \leq m + \frac{1}{2} m \cdot \log_2 \left(\frac{e \cdot \abs{\SC}}{m}\right) + \log_2(\Delta),
$$ here we have used the inequality $\binom{n}{m} \leq \left( \frac{e \cdot n }{m} \right)^m$.

Lemma \ref{tight_delta_ineq} states that 
$$
\abs{\SC} \leq c \cdot m + \log_2(\Delta) + \frac{m}{2} \cdot \log_2 \left( \log_2\sqrt{2e} + \frac{\log_2(\Delta)}{m} \right).
$$

Finally, the inequality 
$$
\|v\|_0 \leq \|u\|_0 + \log_2(\delta)
$$
now follows from Lemma \ref{adj_lm}.
\end{proof}

\begin{remark}\label{ILP_CF_sparsity_rm}
We can deduce less accurate but rather simple bound from the previous Theorem bound using Lemma \ref{rough_delta_ineq}. More precisely, Lemma \ref{rough_delta_ineq} states that 
$$
\|u\|_0 \leq (\alpha + \beta_k)\cdot m + (1 + \gamma_k)\cdot \log_2(\Delta),
$$ where $k \geq 0$ is a parameter. When $k$ grows, $\gamma_k$ decreases with a linear speed to zero and $\beta_k$ increases with a logarithmic speed.

Taking $k = 3$, we will have 
$$
\|u\|_0 \leq 2.81 \cdot m + 1.18 \cdot \log_2(\Delta).
$$

\end{remark}

Now, using Lemma \ref{ILPSF_to_ILPCF_lm}, we state a similar result for the \ref{ILP-SF} problem.
\begin{corollary}\label{ILP_SF_sparsity_cor}
Let $v$ be any vertex integral solution of the \ref{ILP-SF} problem and $\Delta^* = \abs{\det(S)} \cdot \Delta(A)$. Then, the following inequalities hold:
\begin{enumerate}
    \item $\|v\|_0 \leq m + \log_2 \sqrt{\det (A A^\top) } + \log_2 \abs{\det(S)}$;
    \item $\|v\|_0 \leq c\cdot m + \log_2(\Delta^*) + \frac{m}{2} \cdot \log_2 \left( \log_2\sqrt{2e} + \frac{\log_2(\Delta^*)}{m} \right)$, \\
    where $c = \log_2 \sqrt{\frac{2 e^2}{e - \log_2 e}} + \frac{1}{2} \leq 2.27$.
\end{enumerate}
\end{corollary}




\subsection{Proximity Bounds for the \ref{ILP-CF} and \ref{ILP-SF} Problems}\label{un_prox_subs}

The main result of the paper \cite{ProximityViaSparsity} connects together sparsity and proximity properties of the unbounded \ref{CLASSIC-ILP-SF} problem. Our next theorem does the same for a more general \ref{ILP-CF} problem.

\begin{lemma}\label{sparsity_via_proximity_lm} Let us consider the  \ref{ILP-CF} problem. Let $\Delta = \Delta(A)$ and $x^*$ be an optimal solution of the LP relaxation. Then, there exists an optimal integer solution $z^*$ of \ref{ILP-CF}, such that
$$
\|A(x^* - z^*)\|_1 \leq (m+1) \cdot (m + s) \cdot \Delta,
$$ where $s$ is the sparsity parameter defined by
$$
s = \max\{\|y\|_0 \colon A z + y = b,\, z \in \vertex\bigl(\PC_I(A,b)\bigr)\}.
$$

Conversely, for any optimal integer vertex solution $z^*$ of the \ref{ILP-CF} problem, there exists an optimal solution $x^*$ of the LP relaxation, such that $$
\|A(x^* - z^*)\|_1 \leq (m+1) \cdot (m + s) \cdot \Delta.
$$

Additionally, if the problem is $x^*$-normalized, then 
$$
\|x^* - z^*\|_{\infty} \leq 1/2 \cdot (m+1) \cdot (m + s) \cdot \Delta \cdot \delta,
$$ where $\delta = \abs{\det(A_{\BC})}$.

\end{lemma}
\begin{proof} In our proof, we follow to the original proof from \cite{Sensitivity_Tardos} (see also \cite[Chapter~17.2]{Schrijver}) with usage of an idea given in the paper \cite{ProximityViaSparsity}.

Let $\hat z$ be an optimal integer vertex solution of the \ref{ILP-CF} problem 
and $\hat x$ be an optimal solution of the LP relaxation. Let $\SC = \supp(\hat z)$, and $\LC = \supp(\hat x)$, $\HC = \SC \cup \LC$ and $\BHC = \intint{n+m}\setminus \HC$. Clearly, $\abs{\HC} \leq s + m$.

Let us consider the following optimization problem:
\begin{gather}
    c^\top x \to \max \notag\\
    \begin{cases}
    A_{\HC} x \leq b_{\HC} \\
    A_{\BHC} x = b_{\BHC}\\
    x \in \ZZ^n.
    \end{cases}\tag{$ILP\!-\!A_H$}\label{ILP-CF-H}
\end{gather}
Clearly, $\hat x$ and $\hat z$ are LP and ILP optimal solutions for \eqref{ILP-CF-H}. Let $\DC = \{i \colon (A_{\HC})_i \hat x < (A_{\HC})_i \hat z\}$ and $\FC = \HC \setminus \DC$. As in \cite{Sensitivity_Tardos}, let us consider the cone $\KC$, defined by a system
\begin{equation}\label{K_sys}
\begin{cases}
A_{\BHC} x = \BZero\\
A_{\DC} x \leq \BZero\\
A_{\FC} x \geq \BZero\\
x \in \RR^n.
\end{cases}    
\end{equation}

The cone $\KC$ is pointed, because $\rank (A) = n$. Let $\KC = \cone(R)$, where $R \in \ZZ^{n \times t}$ is a matrix, whose columns, composed from the extreme rays of $\KC$.

Let us observe some properties of $\KC$:
\begin{enumerate}
    \item Since $A_{\DC} \hat x < A_{\DC} \hat z \leq b_{\DC}$, by complementary slackness, $c = (A_{\BHC \cup \FC})^\top y$, for some $y \geq 0$. Hence, $c^\top u \geq 0$, for any $u \in \KC$;
    \item By the definition, $\hat x - \hat z \in \KC$;
    \item Since any extreme ray of $\KC$ can be treated as some solution of a $(n-1)$-rank sub-system of the system \eqref{K_sys}, we can assume that any extreme ray $R_{* i}$ is a column of an adjugate of some $n \times n$ sub-matrix of the system, defining $\KC$;
    \item Let $B = A R$, then, for any $i$, $\|B_{* i}\|_0 \leq m+1$ and $\|B_{* i}\|_{\infty} \leq \Delta$. Definitely, the first inequality follows from the definition of $\KC$, and the second inequality follows from the previous property 3. Consequently, if we consider a cone of slacks $A \cdot \KC = \cone(B)$, then this cone is induced by extreme rays with at most $m + 1$ non-zero coordinates, and each coordinate is at most $\Delta$ in the absolute value.
\end{enumerate}

Now, from the property 2, there exists $\lambda \in \RR^t_{\geq 0}$, such that $\hat x - \hat z = R \lambda$. Let us set 
$$
z^* := \hat z + R \lfloor \lambda \rfloor = \hat x - R \{\lambda\}.
$$
Now 
\begin{gather}
    A_{\BHC} z^* = A_{\BHC} \hat z = b_{\BHC},\notag\\
    A_{\DC} z^* = A_{\DC}(\hat z + R \lfloor \lambda \rfloor) \leq A_{\DC} \hat z = b_{\DC}, \label{satisfability_check_proximity_th}\\
    A_{\FC} z^* = A_{\FC} (\hat x - R \{\lambda\} ) \leq b_{\FC}.\notag
\end{gather}

Hence, $z^*$ is an integral and feasible solution of the original system. By the property 1, $c^\top z^* \geq c^\top \hat z$, so $z^*$ is integer optimal. 

To prove the first statement, let us consider the cone of slacks $A \cdot \KC = \cone(B)$. Clearly, $A(\hat x - \hat z) = B \lambda$. Since $B_{\BHC} = \BZero$, by the Carath\'eodory's Theorem, we can assume that $\|\lambda\|_0 \leq m + s +1$. Consequently, by the property 4, we have
$$
\|A(\hat x - z^*)\|_1 = \|B  \{\lambda\} \|_1 \leq (m+1) \cdot (m+s+1) \cdot \Delta.
$$

Similarly, let
$$
x^* = \hat x - R \lfloor \lambda \rfloor = \hat z + R \{\lambda\}.
$$ Then $x^*$ satisfies to $A x \leq b$. Next, we have $c^\top x^* \geq c^\top \hat x$, because in the opposite case $\lfloor \lambda_i \rfloor \not= 0$ and $c^\top R_{* i} > 0$, for some $i$. Then, taking into account \eqref{satisfability_check_proximity_th}, the vector $\hat z + \lfloor \lambda_i \rfloor R_{* i}$ is an integer solution of $A x \leq b$ with $c^\top (\hat z + \lfloor \lambda_i \rfloor R_{* i}) > c^\top z$ that contradicts to the optimality of $\hat z$.

Again
$$
\|A(x^* - \hat z)\|_1 = \|B  \{\lambda\} \|_1 \leq (m+1) \cdot (m+s+1) \cdot \Delta
$$ that proves the second statement.

Finally, if the problem is $x^*$-normalized, we can use Lemma \ref{adj_lm} to estimate $$\|x^* - z^*\|_{\infty} = \|A_{\BC}^{-1}(A_{\BC}(x^* - z^*))\|_{\infty}.$$
\end{proof}

Now, we can use the sparsity results of Subsection \ref{sparsity_subs} to construct our proximity bounds for the \ref{ILP-CF} and \ref{ILP-SF} problems. 

\begin{theorem}\label{ILP_CF_proximity_th}
Let us consider the \ref{ILP-CF} problem. Let $x^*$ be an optimal vertex of the LP relaxation. Then, there exists an optimal integer solution $z^*$ of \ref{ILP-CF}, such that
$$
\|A(x^* - z^*)\|_1 = O( m^2 \cdot \Delta \cdot \log \sqrt[m]{\Delta} ).
$$ 

Conversely, for any optimal integer vertex solution $z^*$ of the \ref{ILP-CF} problem, there exists an optimal solution $x^*$ of the LP relaxation, such that $$
\|A(x^* - z^*)\|_1 = O( m^2 \cdot \Delta \cdot \log \sqrt[m]{\Delta} ).
$$

Additionally, if the problem is $x^*$-normalized, then 
\begin{gather*}
    \|x^* - z^*\|_{\infty} = O( m^2 \cdot \Delta \cdot \delta \cdot \log{\sqrt[m]{\Delta}} ),\\
    \|z^*\|_{\infty} = O( m^2 \cdot \Delta \cdot \delta \cdot \log{\sqrt[m]{\Delta}} ),
\end{gather*} 
where $\delta = \abs{\det(A_{\BC})}$.
\end{theorem}
\begin{proof}
Since $(z^* - x^*) = A_{\BC}^{-1}\bigl(A_{\BC} (z^* - x^*) \bigr)$, by Lemma \ref{adj_lm} and Remark \ref{ILP_CF_sparsity_th}, we have 
$$
\|x^* - z^*\|_{\infty} \leq \const \cdot m^2 \cdot \Delta \cdot \delta \cdot \log(2 \sqrt[m]{\Delta}).
$$

Observe that $\|x^*\|_{\infty} = \|A_{\BC}^{-1} b_{\BC}\|_{\infty} \leq \frac{\delta(\delta-1)}{2}$. Then,
\begin{multline*}
    \|z^*\|_{\infty} \leq \|z^* - x^*\|_{\infty} + \|x^*\|_{\infty} = \frac{\delta}{2} \cdot \bigl( O(m^2 \cdot \Delta \cdot \log\sqrt[m]{\Delta}) + \delta-1 \bigr) = \\
    = O\bigl( m^2 \cdot \Delta \cdot \delta \cdot \log\sqrt[m]{\Delta} \bigr).
\end{multline*}
\end{proof}

\begin{corollary}\label{ILP_SF_proximity_cor}
Let us consider the \ref{ILP-SF} problem and let $\Delta^* = \abs{\det(S)} \cdot \Delta(A)$. Let $x^*$ be an optimal solution of the LP relaxation. Then, there exists an optimal integer solution $z^*$ of the \ref{ILP-SF} problem, such that
$$
\|x^* - z^*\|_1 = O(m^2 \cdot \Delta^* \cdot \log \sqrt[m]{\Delta^*}).
$$

Conversely, for any optimal integer vertex solution $z^*$ of the \ref{ILP-SF} problem, there exists an optimal solution $x^*$ of the LP relaxation, such that $$
\|x^* - z^*\|_1 = O(m^2 \cdot \Delta^* \cdot \log \sqrt[m]{\Delta^*}).
$$
\end{corollary}

\subsection{Bounds On The Number of Integer Vertices For The \ref{ILP-CF} And \ref{ILP-SF} Problems}\label{int_vertex_subs}

Following to the works \cite{SupportIPSolutions},\cite{NewBoundsForFixedM} and using the sparsity bounds from Subsection \ref{sparsity_subs}, we are going to present upper bounds for the number of integer vertices in polyhedra of the \ref{ILP-CF} and \ref{ILP-SF} problems.

\begin{theorem}\label{ILPCF_int_vert_num_th}
Let $\PC_I$ be the integer polyhedron related to the problem \ref{ILP-CF} and $s = \max\{\|y\|_0 \colon A z + y = b,\, z \in \vertex(\PC_I)\}$. 
\begin{equation*}
    \text{Then, }\abs{\vertex(\PC_I)} = (n+m)^s \cdot O(s)^{s+1} \cdot \log^{s-1}(s \cdot \Delta).
\end{equation*}

Due to Remark \ref{ILP_CF_sparsity_rm}, $s = O(m + \log(\Delta))$, and consequently
$$
\abs{\vertex(\PC_I)} = \bigl( n \cdot m \cdot \log(\Delta) \bigr)^{O(m + \log(\Delta))}.
$$
The vertices of $\PC_I$ can be enumerated by an algorithm with the same arithmetic complexity, as in the last bound.

\end{theorem}
\begin{proof}
Let us consider an arbitrary vertex $z \in \vertex(\PC_I)$. Let $\ZC = \{i \colon b_i - A_{i *} z = 0 \}$ and $\SC = \intint{n+m} \setminus \ZC$. Since $\rank(A) = n$, $\rank(A_{\ZC}) \geq n - s$.

We have

$$
\begin{pmatrix}
A_{\ZC} \\
A_{\SC}
\end{pmatrix} z 
\;\emptyBinom{=}{<}\;
\begin{pmatrix}
b_{\ZC} \\
b_{\SC}
\end{pmatrix}.
$$

There exists a unimodular matrix $Q \in \ZZ^{n \times n}$, such that $A_{\ZC} = \bigl(H \, \BZero\bigr) Q$, where $\bigl(H \, \BZero\bigr)$ is the HNF of $A_{\ZC}$ and $H \in \ZZ^{\abs{\ZC} \times r}$, where $r = \rank(H) = \rank(A_{\ZC}) \geq n - s$. Let $z^\prime = Q z$, then
\begin{equation}\label{zeros_elimination_eq_vert}
\begin{pmatrix}
H & \BZero \\
C & B
\end{pmatrix} z^\prime 
\;\emptyBinom{=}{<}\;
\begin{pmatrix}
b_{\ZC} \\
b_{\SC}
\end{pmatrix},    
\end{equation}
where $(C \, B) = A_{\SC} Q^{-1}$ and $B \in \ZZ^{\abs{\SC} \times (n-r)}$. Note that $B$ has the full column rank and $\Delta(B) \leq \Delta$. 
Let us consider the last $\abs{\SC}$ inequalities of the previous system:
\begin{equation}\label{zeros_eliminated_ineq}
B z^* < b_{\SC} -  C z^\prime_{\intint r} = b_{\SC} - C H_{\JC *}^{-1} b_{\JC},    
\end{equation}
where $z^* = z^\prime_{\intint[r+1]{n}}$ is composed from the last $n-r$ components of $z^\prime$ and $\JC \subseteq \ZC$ is a set of indices such that a $r\times r$ sub-matrix $H_{\JC *}$ is non-degenerate.

Since $z$ is a vertex of the original polyhedron $\PC_I$, $z^*$ is an integer vertex of the polyhedron $\PC^*$, defined by 
$$
B x \leq b_{\SC} - C H_{\JC *}^{-1} b_{\JC}.
$$ 

Let us show that the map $z \to z^\prime \to z^*$ is an injection. Definitely, the first part $z^\prime = Q z$ of the map is a bijection and the first coordinates $z^\prime_{\intint r} = H_{\JC *}^{-1} b_{\JC}$ are completely determined by the set $\ZC$. Hence, different vertices $z$ of $\PC_I$ must be mapped to different integer vertices $z^* = z^\prime_{\intint[r+1]{n}}$ of $\PC^*$.

Note, that $\PC^*$ has at most $\abs{\SC} \leq s$ facets and its system has $n - r \leq n - (n-s) = s$ variables.

Now, we can apply the inequality \eqref{int_vert_bound_delta_intro} to estimate the number of integer vertices in $\PC^*$. Noting that there are at most $\binom{n+m}{s}$ ways to chose the set $\SC$, it gives
\begin{multline*}
    \binom{n+m}{s} \cdot (s+1)^{s + 1} \cdot s! \cdot \xi(s,s) \cdot \xi(s,2 s) \cdot O(1)^{s-1} \cdot \log^{s-1}(s \cdot \Delta) = \\
    =  (n+m)^s \cdot O(s)^{s+1} \cdot \log^{s-1}(s \cdot \Delta).
\end{multline*}

Finally, we note that the bounds \eqref{int_vert_bound_delta_ext_intro}, \eqref{int_vert_bound_delta_intro} are constructive, because they use an alternative method that is equivalent to the classical boxing technique, due to \cite{IntVertKnapsack_Hayes} (see also \cite[Chapter~4]{IntHullComplexity_Hartmann} and \cite[Chapter~3]{BlueBook}), applied to a triangulation of $\PC$. The resulting arithmetic complexity bound will be differ from the bound 
$$
\abs{\vertex(\PC_I)} = \bigl( n \cdot m \cdot \log(\Delta) \bigr)^{O(m + \log(\Delta))}
$$ only by some polynomial therm.
\end{proof}

Again, using Lemma \ref{ILPSF_to_ILPCF_lm}, we can formulate similar result for the \ref{ILP-SF} problem.
\begin{corollary}\label{ILPSF_int_vert_num_cor}
Let $\PC_I$ be an integer polyhedron, related to the \ref{ILP-SF} problem, $\Delta^* = \Delta(A) \cdot \abs{\det(S)}$, and $s = \max\{\|z\|_0 \colon z \in \vertex(\PC_I)\}$.
$$
\text{Then, } \abs{\vertex(\PC_I)} = n^s \cdot O(s)^{s+1} \cdot \log^{s-1}(s \cdot \Delta^*).
$$
Due to Remark \ref{ILP_CF_sparsity_rm} and Corollary \ref{ILP_SF_sparsity_cor}, $s = O(m + \log(\Delta))$, and consequently
$$
\abs{\vertex(\PC_I)} = \bigl( n \cdot m \cdot \log(\Delta^*) \bigr)^{O(m + \log(\Delta^*))}.
$$
The vertices of $\PC_I$ can be enumerated by an algorithm with the same arithmetic complexity, as in the last bound.
\end{corollary}




\subsection{Proximity Bounds for the \ref{BILP-CF} and \ref{BILP-SF} Problems}\label{bn_prox_subs}

The paper \cite{SteinitzILP} presents an elegant way how to apply the Steinitz's theorem to construct new proximity results for the \ref{CLASSIC-ILP-SF} problem. In the current Subsection, we generalize the original proof of \cite{SteinitzILP} to work with the \ref{BILP-CF} and \ref{BILP-SF} problems.

\begin{theorem}[E.~Steinitz \cite{SteinitzOriginal} (1913)]\label{Steinitz_th}
Let $\|\cdot\|$ be an arbitrary norm of $\RR^m$ and let $x_1, x_2, \dots, x_n \in \RR^m$, such that 
$$
\sum_{i = 1}^n x_i = \BZero \quad\text{and}\quad \|x_i\| \leq 1, \quad\text{for each $i$.}
$$
There exists a permutation $\pi \in S_n$, such that all partial sums satisfy
$$
\left\| \sum_{i = 1}^k x_{\pi(i)} \right\| \leq c(m), \quad\text{for all $k \in \intint{n}$.}
$$
Here, $c(m)$ is a constant, depending on $m$ only.
\end{theorem}

E.~Stenitz showed $c(m) \leq 2m$. It was later shown by Sevast'anov \cite{Sheduling_Sevast} that $c(m) \leq m$, a simplified proof is given in \cite{OnTheValueOfSteinitzConst}. A small survey about the Steinitz's constant is given in \cite{SteinitzILP}. 
It is conjectured that the Steinitz constant should be $O(\sqrt{m})$ for $l_\infty$-norm \cite{PowerOfLinearDependencies}. A proof of this conjecture or any asymptotic improvement would directly improve the bounds, provided in this Section.

\begin{definition}
A vector $y \in \ZZ^n$ is called \emph{a cycle} of $(z^* - x^*)$ if the following conditions are satisfied:
\begin{gather*}
    A y = \BZero, \quad G y \equiv \BZero \pmod{S \cdot \ZZ^n},\\
    \abs{y_i} \leq \abs{z_i^* - x_i^*} \quad\text{and}\quad y_i (z^*_i - x^*_i) \geq 0, \text{ for each $i$}.
\end{gather*}
\end{definition}

\begin{lemma}\label{no_cycle_lm}
Let $y$ be a cycle of $z^* - x^*$, then
\begin{enumerate}
    \item $z^* - y$ is an integer feasible solution of the \ref{BILP-SF} problem;
    \item $z^* + y$ is a feasible solution of the LP relaxation;
    \item One has that $c^\top y \leq 0$.
\end{enumerate}
\end{lemma}
Proof of this Lemma is completely similar to \cite[Lemma~5]{SteinitzILP}. The following Corollary is its straightforward consequence.
\begin{corollary}\label{no_cycle_cor}
Suppose that $x^*$ and $z^*$ additionally have the property that $z^* - x^*$ has the minimal $l_1$-norm. Then, there are no nontrivial cycles of $z^* - x^*$.
\end{corollary}

\GribanovAdd{
\begin{theorem}\label{BILP_SF_proximity_th}
Let us consider the \ref{BILP-SF} problem with $m \geq 1$. Let $x^*$ be an optimal solution of the LP relaxation and $\Delta = \Delta(A)$. Then, there exists an optimal integer solution $z^*$ of the \ref{BILP-SF} problem, such that
$$
\|x^* - z^*\|_1 \leq m \cdot (2 m + 1)^m \cdot \Delta \cdot \abs{\det (S)}.
$$

In the case $m = 0$ the following bound holds
$$
\|z^*\|_1 = \|x^* - z^*\|_1 \leq \abs{\det (S)} - 1.
$$
\end{theorem}
\begin{proof} Here we follow to the original proof of \cite{SteinitzILP} together with the idea, due to \cite{ProximityViaSparsity}, to work with a specialised vector norm. Our contribution here is a generalisation of the considered ideas to handle the group constraint $G x \equiv g \pmod{S \cdot \ZZ^n}$.

Let $z^*$ be an optimal integer solution such that $\|z^* - x^*\|_1$ is minimized. Let $B \in \ZZ^{m \times m}$ be a sub-matrix of $A$, such that $\abs{\det(B)} = \Delta$. W.l.o.g. we can assume that $B = A_{\intint m}$. Clearly, $A = B \bigl(I\, U\bigr)$, where $\Delta_i(U) \leq 1$, for all $i \in \intint m$, and consequently $\|U\|_{\max} \leq 1$. Let us define a special norm $\|x\|_B$, given by the formula
$$
\|x\|_B = \|B^{-1} x\|_\infty.
$$

Let us define $$
\lfloor x^* \rceil_i = \begin{cases}
\lceil x^* \rceil_i \quad \text{if $z^*_i > x^*_i$ and}\\
\lfloor x^* \rfloor_i \quad \text{if $z^*_i \leq x^*_i$}
\end{cases}
$$ and re-denote until the end of the proof $\{x^*\} = x^* - \lfloor x^* \rceil$.

Since $A (z^* - x^*) = 0$, it follows that
$$
A (z^* - \lfloor x^* \rceil) = -A\{x^*\}.
$$

We are going to apply the Steinitz's Lemma. First of all, we construct two sequences of columns
\begin{gather*}
    v_1, \dots, v_{t}, \quad\text{and}\quad \hat v_1, \dots, \hat v_{t},\\
\end{gather*}
where $t = \|z^* - \lfloor x^* \rceil\|_1$. The sequence $\{v_i\}$ is constructed in the following way: for each index $i$, append $\abs{(z^* - \lfloor x^* \rceil)_i}$ copies of $\sgn\bigl((z^* - \lfloor x^* \rceil)_i\bigr) \cdot A_i$, where $A_i$ is the $i$-th column of $A$. Now, we explain how to construct the sequence $\{\hat v_i\}$. The coordinates of each column $\hat v_i \in \ZZ^n$ consists from two parts: the first is $v_i$, the second is the corresponding column of $G$ modulo $S$, taken with the correct sign. More precisely, if $v_i = \pm a$, where $a$ is $j$-th column of $A$, then the last $n - m$ coordinates of $\hat v_i$ are $\pm G_{*j} \bmod S$.

Since $\{x^*\}$ has at most $m$ non-zero components, we have 
\begin{gather*}
    \|v_i\|_B = \|A_i\|_B = \|\bigl(I\, U\bigr)_i\|_{\infty} \leq 1, \quad\text{for each $i \in \intint{t}$.}
\end{gather*}

Let us show that the integer vector $A\{x^*\}$ can be represented as the sum 
$$
A\{x^*\} = w_1 + w_2 + \dots + w_m
$$ of rational vectors $w_i$ with $\|w_i\|_B \leq 1$. Definitely, let $y = \bigl(I\, U\bigr)\{x^*\}$. Since, $\|y\|_{\infty} \leq m$, we can represent it as $y = \sum_{i=1}^m y_i$, where $y_i \in \ZZ^n$ for $i \in \intint{m-1}$, $y_m = y - \sum_{i=1}^{m-1} y_i \in \QQ^n$ and $\|y_i\|_{\infty} \leq 1$, for each $i \in \intint m$. Now, we put $w_i = B y_i$. By the construction, $\|w_i\|_{B} \leq 1$ and $w_i \in  \ZZ^n$ for $i \in \intint{n-1}$. Finally, $w_n \in \ZZ^n$, because $w_n = A \{x^*\} - \sum_{i=1}^{n-1} w_i \in \ZZ^n$. 

Similarly, we construct columns $\hat w_i$ by appending $g \bmod S$ to the last components of $w_1$ and $\BZero_{n-m}$ to the last components of $w_i$, for $i \in \intint[2]{m}$. 

We append the columns $\{w_i\}$, $\{\hat w_i\}$ to the sequences $\{v_i\}$, $\{\hat v_i\}$, and obtain the sequences:
\begin{gather*}
    v_1, \dots, v_t, w_1, \dots, w_m,\\
    \hat v_1, \dots, \hat v_t, \hat w_1, \dots, \hat w_m.
\end{gather*}

 Observe that columns of the sequence $\{v_i, w_j\}$ sums up to $\BZero$ and have the property that each element of the sequence has $\|\cdot\|_B$-norm, bounded by $1$. Now, applying the Steinitz's Lemma, we rearrange the sequences in such a way
$$
u_1, \dots, u_{t+m}, \quad\text{and}\quad \hat u_1, \dots, \hat u_{t+m}
$$ that for each $j \in \intint{t+m}$ the partial sums $p_j = \sum_{i = 1}^j u_i$ (the partial sums $\{\hat p_j\}$ are defined in the same way) satisfies
$$
\|p_j\|_{B} \leq m.
$$

Let $\MC = \ZZ^m \cap \{B x \colon \|x\|_\infty \leq m \}$. Observe that elements of $\{\hat p_j\}$ are contained in the set 
\begin{equation}\tag{$\{\hat p_i\}$-set}\label{set_p}
\left\{ \binom{x}{y} \colon x \in \MC,\, y \in \ZZ^{n-m} \bmod\, S \right\}.
\end{equation}
Hence, due to Lemma \ref{par_enum_lm}, the sequence $\{\hat p_j\}$ has at most 
$$
(2 m + 1)^m \cdot \Delta \cdot \abs{\det(S)}
$$ different elements.

We will now argue that there are not indices $1 \leq s_1 < \dots < s_{m+1} \leq t+m$ with $\hat p_{s_1} = \hat p_{s_2} = \dots = \hat p_{s_{m+1}}$. Which implies that $t+m$ is bounded by $m$-times the number of points in the \eqref{set_p} set. 
Consequently,
\begin{multline*}
  \bigl\|z^* - x^*\bigr\|_1 \leq \bigl\|z^* - \lfloor x^* \rceil\bigr\|_1 + \bigl\|\{x^*\}\bigr\|_1 \leq t + m \leq \\
  \leq m \cdot (2 m + 1)^m \cdot \Delta \cdot \abs{\det(S)}.
\end{multline*}

Assume to this end that there exist indices $1 \leq s_1 < \dots < s_{m+1} \leq t + m$ with $\hat p_{s_1} = \hat p_{s_2} = \dots = \hat p_{s_{m+1}}$. Note that $\hat p_{t+m} = 0$, this yields a partition of the sequence $\{\hat u_i\}$ into $m+1$ pieces that sum up to zero, namely:
\begin{gather*}
    (\hat u_1, \dots, \hat u_{s_1}, \hat u_{s_{m+1}+1}, \dots, \hat u_{t+m})\\
    (\hat u_{s_i + 1}, \dots, \hat u_{s_{i+1}}), \quad\text{for $i \in \intint{m}$.}
\end{gather*}

One of these pieces does not contain an element of the sequence $\{\hat w_i\}$, so it consists from columns of $A$ or negatives thereof. This corresponds to a cycle $y$ of $z^*-x^*$, which, by the minimality of $\|x^* - z^*\|_1$ and Corollary \ref{no_cycle_cor}, is impossible.

Let us separately consider the case $m = 0$. In this case the polyhedron of the relaxed LP problem is just the positive cone $x \geq \BZero$. Hence, the relaxed LP problem has unique optimal solution $\BZero$. Taking $t = \|z^*\|_1$, we construct the sequence 
$$
v_1, v_2, \dots, v_{t+1}
$$ by the similar way: for each index $i$, append $z^*_i$ copies of $G_i \bmod S$ to $\{v_i\}$. After that, append $v_{t+1} = g \bmod S$ to $\{v_i\}$.

Consider the partial sums $p_j = \sum_{i = 1}^j v_i$. We will now argue that all elements of $\{p_i\}$ are different, which implies that $t + 1 \leq \Delta \cdot \abs{\det(S)}$ and $$
\|x^* - z^*\|_1 = \|z^*\|_1 = t \leq \Delta \cdot \abs{\det(S)} - 1.
$$

Assume to this end that there exist two indices $1 \leq s_1 < s_2 \leq t+1$ with $p_{s_1} = p_{s_2}$. This yields a partition of the sequence $\{v_i\}$ into two pieces that sum up to zero, namely:
\begin{equation*}
    (v_1, \dots, v_{s_1}, v_{s_2 + 1}, \dots, v_{t+1}), \quad\text{and }
    (v_{s_1+1}, \dots,v_{s_2}).
\end{equation*}

One of these pieces does not contain the element $v_{t+1}$, so it consists from columns of $G \bmod S$. This again corresponds to a cycle of $z^*-x^*$.
\end{proof}
}

\begin{remark}\label{ILP_SF_proximity_rm}
Note that the previous theorem holds even for $m = 0$. 


The case $m = 0$ corresponds to the group minimization problem \ref{bn_group_prob}. Hence, we have proven that the problem \ref{bn_group_prob} has an optimal solution $z^*$ with
$$
\|z^*\|_1 \leq \Delta - 1.
$$

Let us make one more remark. It can be easily seen that the original proof of \cite[Theorem~7]{SteinitzILP} for the \ref{CLASSIC-ILP-SF} problem with respect to the parameters $m$ and $\Delta_1 = \Delta_1(A)$ can be extended to work with the \ref{BILP-SF} problem, using the same technique as in Theorem \ref{BILP_SF_proximity_th}. For $m \geq 1$, it gives the following proximity bound:
$$
\|x^* - z^*\|_1 \leq m \cdot (2 \cdot m \cdot \Delta_1 + 1)^m \cdot \abs{\det(S)}.
$$
\end{remark}

\GribanovAdd{
\begin{corollary}\label{BILP_CF_proximity_cor}
Let us consider the \ref{BILP-CF} problem. Let $x^*$ be an optimal vertex solution of the LP relaxation and $\Delta = \Delta(A)$. Then, there exists an optimal integral solution $z^*$, such that
$$
\|A(z^* - x^*)\|_1 \leq m \cdot (2 m + 1)^m \cdot \Delta.
$$

Additionally, if the problem is $x^*$-normalized, then \begin{enumerate}
    \item $\|z^* - x^*\|_{\infty} \leq m \cdot (2 m + 1)^m \cdot \Delta \cdot \delta$,
    \item $\|z^*\|_{\infty} \leq m \cdot (2 m + 1)^m \cdot \Delta \cdot \delta$,
\end{enumerate} 
where $\delta = \abs{\det(A_{\BC})}$.
\end{corollary}
\begin{proof}

The inequality 
$$
\|A(z^* - x^*)\|_1 \leq m \cdot (2 m + 1)^m \cdot \Delta .
$$ directly follows from Lemma \ref{ILPCF_to_ILPSF_lm} and Theorem \ref{BILP_SF_proximity_th}.

Since $(z^* - x^*) = A_{\BC}^{-1}\bigl(A_{\BC} (z^* - x^*) \bigr)$, by Lemma \ref{adj_lm}, we have 
$$
\|z^* - x^*\|_{\infty} \leq \frac{\delta}{2} \bigl(m \cdot (2 m + 1)^m \cdot \Delta\bigr) \leq m \cdot (2 m + 1)^m \cdot \Delta \cdot \delta.
$$

Observe that $\|x^*\|_{\infty} = \|A_{\BC}^{-1} b_{\BC}\|_{\infty} \leq \frac{\delta(\delta-1)}{2}$. Then,
\begin{multline*}
    \|z^*\|_{\infty} \leq \|z^* - x^*\|_{\infty} + \|x^*\|_{\infty} \leq \frac{\delta}{2} \bigl( m\cdot (2 m + 1)^m \cdot \Delta + (\delta - 1) \bigr) \leq \\
    \leq m\cdot (2 m + 1)^m \cdot \Delta \cdot \delta.
\end{multline*}
\end{proof}
}

\section{Dynamic Programming Algorithms}\label{algorithms_sec}

In this Section, we will give dynamic programming algorithms for the \ref{BILP-SF} and \ref{ILP-SF} problems. After that we reduce the \ref{BILP-CF} and \ref{ILP-CF} problems to the problems in standard form.

\subsection{Algorithms For The Bounded Problems}\label{bounded_complexity_subs}

To provide algorithms for the \ref{BILP-CF} and \ref{BILP-SF} problems, we need the following simple Lemma and its Corollary.

\begin{lemma}\label{DP_width_lm}
Let $A \in \ZZ^{m \times n}$ and $B \in \ZZ^{m \times m}$ be a non-degenerate sub-matrix of $A$. Let, additionally, $\gamma \in \RR_{>0}$, $\Delta = \Delta(A)$, $\delta = \abs{\det(B)}$, and 
$$
\PC = \{y = A x \colon x \in \RR^n,\, \|x\|_1 \leq \gamma \},
$$
$$
\text{then}\quad \abs{\PC \cap \ZZ^m} \leq \left( 2 \gamma \cdot \frac{\Delta}{\delta} + 1 \right)^m \cdot \Delta.
$$

Points of $\PC \cap \ZZ^m$ can be enumerated by an algorithm with the arithmetic complexity bound:
$$
O\left(m^2 \cdot \Delta \cdot  \left(2 \gamma \cdot \frac{\Delta}{\delta} \right)^m \right).
$$
\end{lemma}
\begin{proof}
W.l.o.g. we can assume that the first $m$ columns of $A$ form the sub-matrix $B$. Let us consider a decomposition $A = B \bigl(I \; U\bigr)$, where $\bigl(I \; U\bigr)$ is a block-matrix, $I$ is the $m \times m$ identity matrix and the matrix $U$ is determined uniquely from this equality. Clearly, $\Delta(\bigl(I \; U\bigr)) = \frac{\Delta}{\delta}$, so $\Delta_k(U) \leq \frac{\Delta}{\delta}$, for all $k \in \intint m$. Let us consider the set 
$$
\NC = \{ y = B x \colon \|x\|_{\infty} \leq \gamma \cdot \frac{\Delta}{\delta} \}.
$$
Let us show that $\PC \subseteq \NC$. Definitely, if $y = A x$, for $\|x\|_1 \leq \gamma$, then $y = B \bigl(I \; U\bigr) x = B t$, for some $t \in [-\gamma,\gamma]^m \cdot \frac{\Delta}{\delta}$. 

Now, the proof follows from Lemma \ref{par_enum_lm}.

\end{proof}

\begin{corollary}\label{DP_width_cor}
Let $A \in \ZZ^{m \times n}$, $\gamma \in \RR_{>0}$, $\Delta = \Delta(A)$, and 
$$
\PC = \{y = A x \colon x \in \RR^n,\, \|x\|_1 \leq \gamma \},
$$
$$
\text{then}\quad \abs{\PC \cap \ZZ^m} \leq \left( 2 \gamma + 1 \right)^m \cdot \Delta.
$$

Points of $\PC \cap \ZZ^m$ can be enumerated by an algorithm with the arithmetic complexity bound:
$$
O(\log m)^{m^2} \cdot \Delta \cdot \gamma^m.
$$
\end{corollary}
\begin{proof}
W.l.o.g. we can assume that $\rank(A) = m$. Let us choose $B \in \ZZ^{m \times m}$, such that $\abs{\det(B)} = \Delta$, then the desired $\abs{M \cap \ZZ^m}$-bound follows from the previous Lemma \ref{DP_width_lm}. 

Due to \cite{SubdeterminantApprox}, we can compute a matrix $\hat B \in \ZZ^{m \times m}$, such that $\Delta = {O(\log m)}^m \cdot \delta$, where $\delta = \abs{\det (\hat B)}$, by a polynomial-time algorithm. Finally, we take the complexity bound of the previous Lemma \ref{DP_width_lm}.
\end{proof}

\begin{definition}\label{prox_def}
Let us consider the \ref{BILP-SF} problem. Let $z^*$ be an optimal solution of the problem and $x^*$ be an optimal vertex-solution of the LP relaxation. \emph{The $l_1$-proximity bound $\chi$} of the \ref{BILP-SF} problem is defined by the formula
$$
\chi = \max_{x^*}\min_{z^*} \|x^* - z^*\|_1.
$$

Due to Theorem \ref{BILP_SF_proximity_th},
\begin{equation}\label{chi_delta_ineq}
    \chi \leq (2m+1)^m \cdot \Delta \cdot \abs{\det(S)},
\end{equation}
where $\Delta = \Delta(A)$.

It is proved in \cite{SteinitzILP} for the classical bounded \ref{CLASSIC-ILP-SF} that
\begin{equation*}
    \chi \leq m \cdot (2 m \cdot \Delta_1 + 1)^m,\quad\text{where $\Delta_1 = \Delta_1(A)$.}
\end{equation*}
The last bound can be easily generalized for the problem \ref{BILP-SF}, see Remark \ref{ILP_SF_proximity_rm}:
\begin{equation}\label{chi_delta1_ineq}
  \chi \leq m \cdot (2 m \cdot \Delta_1 + 1)^m \cdot \abs{\det(S)}.  
\end{equation}

Finally, we note that the very recent result due to \cite{ModularDiffColumns} states that
\begin{equation}\label{chi_diffcol_delta_ineq}
 \chi \leq m \cdot (m+1)^2\cdot \Delta^3 + (m+1) \cdot \Delta = O(m^3 \cdot \Delta^3).    
\end{equation}
\end{definition}
The proof of the next theorem is based on the classical dynamic programming principle for ILP problems in the standard form, due to \cite{Papadimitriou}, with some additional modifications, due to the papers \cite{SteinitzILP} and \cite{KnapsackDPTrick}.
\begin{theorem}\label{BILP_SF_complexity_th}
The \ref{BILP-SF} problem can be solved by an algorithm with the following arithmetic complexity bound:
\begin{equation*}
    T_{LP} + n \cdot O(\chi + m)^{m+1} \cdot \log^2(\chi+m) \cdot \Delta \cdot \abs{\det(S)} \cdot \bigl( m + \log\abs{\det(S)}\bigr).
\end{equation*}

The previous complexity bound can be slightly improved in terms of $\chi$:
\begin{gather*}
    T_{LP} + n \cdot O(\log m)^{m^2} \cdot (\chi + m)^m \cdot \Delta \cdot \abs{\det(S)} \cdot \bigl( m + \log\abs{\det(S)}\bigr).
\end{gather*}
\end{theorem}

\begin{proof}
\GribanovAdd{Assume that the relaxed LP problem is bounded.} Let $x^*$ be an optimal vertex solution of the LP relaxation of the \ref{BILP-SF} problem. After a standard change of coordinates $x \to x - \lfloor x^* \rfloor$ the original \ref{BILP-SF} problem transforms to an equivalent ILP with different lower and upper bounds on variables and a different r.h.s. vector $b$. 

\GribanovAdd{In the case, when the relaxed LP problem is unbounded, we just need to distinguish between two possibilities: 1) $\PC_I \not= \emptyset$, 2) $\PC_I = \emptyset$. Let us choose any vertex $x^*$ of $\PC$ and any objective function $c'^\top x$ that attains its maximum on $x^*$. Clearly, $\PC_I \not= \emptyset$ iff the ILP problem $\max\{c'^\top x \colon x \in \PC_I\}$ is feasible. So, we just reduced the considered case to the case, when the relaxed LP problem is bounded.}


Any optimal vertex solution of the LP problem has at most $m$ non-zero coordinates, so we have the following bound on the $l_1$-norm of an optimal ILP solution $z^* - \lfloor x^* \rfloor$ of the new problem:
\begin{equation*}
    \|z^* - \lfloor x^* \rfloor\|_1 \leq \|x^* - z^*\|_1 + \|x^* - \lfloor x^* \rfloor\|_1 \leq \chi + m.
\end{equation*}

Additionally, let us denote the transformed problem by
\begin{gather}
    c^\top x \to \min\notag\\
    \begin{cases}
    A x = b \\
    G x \equiv g \pmod{S \cdot \ZZ^n}\\
    u_l \leq x \leq u_r\\
    x \in \ZZ^n.
    \end{cases}\tag{New-BILP}\label{T-BILP-SF}
\end{gather}

{\bf The First Complexity Bound}
Let us consider the weighted digraph $G = (V,E)$, whose vertices are quadruples $(k,b^\prime,g^\prime,h)$, for $k \in \intint n$, $h \in \intint[0]{(\chi + m)}$, $b^\prime \in \{A x \colon \|x\|_1 \leq h \} \cap \ZZ^m$ and $g^\prime \in \ZZ^{n-m} \bmod\, S$. The vertex $(k,b^\prime,g^\prime,h)$ corresponds to the problem
\begin{gather*}
    c_{1:k}^\top x \to \min\\
    \begin{cases}
    A_{1:k} x = b^\prime\\
    G_{1:k} x \equiv g^\prime \pmod{S \cdot \ZZ^n}\\
    (u_l)_{1:k} \leq x \leq (u_r)_{1:k}\\
    \|x\|_1 \leq h\\
    x \in \ZZ^k.
    \end{cases}
\end{gather*}

Using Corollary \ref{DP_width_cor}, we bound the number of vertices $\abs{V}$ by 
$$
n \cdot O(\chi + m)^{m+1} \cdot \Delta \cdot \abs{\det(S)}.
$$ By definition, any vertex $(k,b^\prime,g^\prime,h)$ has an in-degree at most $\min\{(u_r)_k - (u_l)_k, h\}+1$. More precisely, for any integer $t \in \intint[(u_l)_k]{(u_r)_k} \cap [-h,h]$ there is an arc from $(k-1, b^\prime - A_k t, (g^\prime - G_k t) \bmod S, h - t)$ to $(k, b^\prime, g^\prime, h)$, this arc is weighted by $c_k t$. Note that vertex $(k-1, b^\prime - A_k t, (g^\prime - G_k t) \bmod S, h - t)$ exists only if $\abs{t} \leq h$. Additionally, we add to $G$ a starting vertex $v_s$, which is connected with all the vertices of the first level $(1,*,*,*)$, weights of this arcs correspond to solutions of $1$-dimensional sub-problems. Clearly, the number of arcs can be estimated by 
$$
\abs{E} = O(\abs{V} \cdot (\chi+m)) = n \cdot O(\chi + m)^{m+2} \cdot \Delta \cdot \abs{\det(S)}.
$$

The \eqref{T-BILP-SF} problem is equivalent to searching of the shortest path starting from the vertex $v_s$ and ending at the vertex $(n, b, g, \chi + m)$ in $G$. Since the graph $G$ is acyclic, the shortest path problem can be solved by an algorithm with the complexity bound 
$$
O(\abs{V} + \abs{E}) = n \cdot O(\chi + m)^{m+2} \cdot \Delta \cdot \abs{\det(S)}.
$$ Here, we assume that the complexity of any vertex or edge is $O(1)$.

We note that during the shortest path problem solving, the graph $G$ must be evaluated on the fly. In other words, the vertices and arcs of $G$ are not known in advance, and we build them online. To make constant-time access to vertices we can use a hash-table data structure with constant-time insert and search operations.

Finally, using the binarization trick, described in the work \cite{SteinitzILP}, we can significantly decrease the number of arcs in $G$. Let $[\alpha_k, \beta_k] = [(u_l)_k, (u_r)_k] \cap [-h, h]$, the idea of the trick is that any integer $t \in [\alpha_k, \beta_k]$ can be uniquely represented, using at most $O(\log^2(h)) = O(\log^2 (\chi+m))$ bits. More precisely, for any interval $[\alpha_k, \beta_k]$, there exist at most $O(\log^2 (\chi+m))$ integers $s(i,k)$, such that any integer $t \in [\alpha_k, \beta_k]$ can be uniquely represented as 
\begin{gather*}
    t = \sum_i s(i,k) x_i,\quad\text{ where $x_i \in \{0,1\}$, and }\\
    \sum_i s(i,k) x_i \in [\alpha_k, \beta_k], \quad \text{for any $x_i \in \{0,1\}$}.
\end{gather*}
Using this idea, we replace the part of the graph $G$ connecting vertices of the levels $(k-1,*,*,*)$ and $(k,*,*,*)$ by an auxiliary graph, whose vertices correspond to the quadruples $(i,b^\prime,g^\prime,h)$, where $i \in \{0,1, \dots, O(\log^2 (\chi+m))\}$, and any quadruple $(i,b^\prime,g^\prime,h)$ has in-degree at most two. More precisely, the vertex $(i,b^\prime,g^\prime,h)$ is connected with at most two vertices: 
$$
(i-1, b^\prime,g^\prime,h) \quad\text{and}\quad (i-1, b^\prime - s(i,k) A_k, (g^\prime - s(i,k) G_k) \bmod S, h - s(i,k)).
$$
The first edge has the weight $0$ and the second has the weight $s(i,k) c_k$. The resulting graph has $O(\log^2 (\chi + m) \abs{V})$ vertices and arcs, where $\abs{V}$ corresponds to the original graph. The total arithmetic complexity of the shortest path problem solving can be estimated as 
$$
n \cdot O(\chi + m)^{m+1} \cdot \log^2(\chi+m) \cdot \Delta \cdot \abs{\det(S)}.
$$
Finally, since any column $A_k$ has $m$ components and $G_k \bmod S$ has at most $\log_2\abs{\det(S)}$ components, the arithmetic complexity of an edge or a vertex traversal is $m + \log_2\abs{\det(S)}$. Hence, the total arithmetic complexity of the algorithm can be roughly estimated as 
$$
n \cdot O(\chi + m)^{m+1} \cdot \log^2(\chi+m) \cdot \Delta \cdot \abs{\det(S)} \cdot \bigl( m + \log\abs{\det(S)}\bigr).
$$

{\bf The Second Complexity Bound}

Let us consider the weighted digraph $G = (V, E)$, whose vertices are triples $(k,b^\prime,g^\prime)$, for $k \in \intint n$, $b^\prime \in \MC := \{A x \colon \|x\|_1 \leq \chi + m \} \cap \ZZ^m$ and $g^\prime \in \ZZ^{n-m} \bmod\, S$. The edges of $G$ have the same structure as in the graph from the previous Subsection. More precisely, for any $t \in \intint[(u_l)_k]{(u_r)_k}$, we put an arc from $(k-1, b^\prime - t A_k, (g^\prime - t G_k) \bmod S)$ to $(k, b^\prime, g^\prime)$, if such vertices exist in $V$, the arc is weighted by $c_k t$. Let us again assume that the complexity of any edge or vertex is $O(1)$. We will erase this assumption at the end of the proof. 

We compute all the vertices of $G$ directly, using Corollary \ref{DP_width_cor}. The arithmetic complexity of this step is bounded by 
$$
n \cdot {O(\log m)}^{m^2} \cdot (\chi + m)^m \cdot \Delta \cdot \abs{\det(S)}.
$$ Due to Corollary \ref{DP_width_cor}, 
$$
\abs{V} = n \cdot \abs{M} \cdot \abs{\det(S)} = n \cdot O(\chi + m)^m \cdot \Delta \cdot \abs{\det(S)}
$$ and, since an in-degree of any vertex in $G$ is bounded by $\chi + m +1$,
$$ 
\abs{E} = n \cdot O(\chi + m)^{m+1} \cdot \Delta \cdot \abs{\det(S)}.
$$

Let us fix some vertex-level $(k, *, *)$ of $G$, for some $k \in \intint n$, and consider an auxiliary graph $F_k$, whose vertices are pairs $(b^\prime,g^\prime)$ for $b^\prime \in M$ and $g^\prime \in \ZZ^{n - m} \bmod\, S$. For two vertices $(b_1^\prime,g_1^\prime)$ and $(b_2^\prime,g_2^\prime)$, we put an arc from $(b_1^\prime,g_1^\prime)$ to $(b_2^\prime,g_2^\prime)$ if $b^\prime_2 - b^\prime_1 = A_k$ and $g^\prime_2 - g^\prime_1 \equiv G_k \pmod{S \cdot \ZZ^n}$. Since "in" and "out" degrees of any vertex in $F_k$ are at most one, the graph is the disjoint union of paths and cycles. More precisely, if $A_k = \BZero$, then $F_k$ is the disjoint union of cycles, and if $A_k \not= \BZero$, then $F_k$ is the disjoint union of paths. Observe, that this decomposition can be computed by an algorithm with the complexity 
$$ 
O(\abs{V(F_k)}) = O(\abs{M} \cdot \abs{\det(S)}) = O(\chi + m)^m \cdot \Delta \cdot \abs{\det(S)}.
$$

Let $((b^\prime_0, g^\prime_0), \dots, (b^\prime_{l-1}, g^\prime_{l-1}))$ be some path or cycle of the length $l$ in the decomposition and $\shortest(k,b^\prime, g^\prime)$ be the value of the shortest path in the original graph $G$, starting at $v_s$ and ending at $(k,b^\prime,g^\prime)$. Next, we are going to show how to find the values of $\shortest(k,\cdot,\cdot)$ along $((b^\prime_0, g^\prime_0), \dots, (b^\prime_{l-1}, g^\prime_{l-1}))$, for the cyclic and path cases in $O(l)$ time.

{\bf Case 1}: $A_k = \BZero$ and $((b^\prime_0, g^\prime_0), \dots, (b^\prime_{l-1}, g^\prime_{l-1}))$ is a cycle. Since $A_k = \BZero$ and $c_k \geq 0$, we have $\lfloor x^*_k \rfloor = x^*_k = 0$, and $[(u_l)_k, (u_r)_k] = [0, (u_r)_k]$. Again, since $c_k \geq 0$, we can assume that $(u_r)_k < l$.  Clearly, for any $i \in \intint[0]{l-1}$, the value of $\shortest(k,b^\prime_i, g^\prime_i)$ can be computed by the formula
\begin{equation}\label{shortest_cycle}
   \shortest(k,b^\prime_i, g^\prime_i) = \min\limits_{t \in [0, (u_r)_k]} \shortest(k-1,b^\prime_{(i-t) \bmod l}, g^\prime_{(i - t) \bmod l}) + c_k t.
\end{equation}

Let us consider a queue $Q$ with the operations: $Enqueue(Q,x)$ that puts an element $x$ into the tail of $Q$, $Dequeue(Q)$ that removes an element $x$ from the head of $Q$, $GetMax(Q)$ that returns maximum of elements of $Q$. It is a known fact that a queue can be implemented, such that all the given operations will have the amortized complexity $O(1)$. Now, we compute $\shortest(k,b^\prime_i, g^\prime_i)$, for $i \in \intint[0]{(l-1)}$, using the following algorithm:
\begin{algorithm}[H]
\caption{Compute $\shortest(k,\cdot,\cdot)$: the cyclic case}
\begin{algorithmic}[1]
\State Create an empty queue $Q$;
\For{$t := (u_r)_k$ {\bf down to} $0$}
    \State $pos := (0 - t) \bmod l$;
    \State $Enqueue(Q, \shortest(k-1,b^\prime_{pos}, g^\prime_{pos}) + c_k t)$;
\EndFor
\For{$i := 0$ {\bf to} $l-1$}
    \State $\shortest(k,b^\prime_i,g^\prime_i) := GetMax(Q) + c_k i$;
    \State $Dequeue(Q)$;
    \State $new\_pos := (i + 1) \bmod l$;
    \State $Enqueue(Q, \shortest(k-1,b^\prime_{new\_pos}, g^\prime_{new\_pos}) - c_k (i + 1))$;
\EndFor
\end{algorithmic}
\end{algorithm}

The correctness of the algorithm follows from the formula \eqref{shortest_cycle}. The algorithm's complexity is $O(l)$.

{\bf Case 2}: $A_k \not= \BZero$ and $((b^\prime_0, g^\prime_0), \dots, (b^\prime_{l-1}, g^\prime_{l-1}))$ is a path. Clearly, for any $i \in \intint[0]{l-1}$, the value of $\shortest(k,b^\prime_i, g^\prime_i)$ can be computed by the formula
\begin{equation}\label{shortest_path}
    \shortest(k,b^\prime_i, g^\prime_i) = \min\limits_{t \in [(u_l)_k, (u_r)_k] \cap ( i - l, i]} \shortest(k-1,b^\prime_{i-t}, g^\prime_{i - t}) + c_k t.
\end{equation}

Again, let us consider the same queue $Q$. Now, we compute $\shortest(k,b^\prime_i, g^\prime_i)$, for $i \in \intint[0]{(l-1)}$ using the following algorithm:
\begin{algorithm}[H]
\caption{Compute $\shortest(k,\cdot,\cdot)$: the path case}
\begin{algorithmic}[1]
\State Create an empty queue $Q$;
\For{$t := 0$ {\bf to} $\min\{l-1, \abs{(u_l)_k}\}$}
    \State $Enqueue(Q, \shortest(k-1,b^\prime_{t}, g^\prime_{t}) - c_k t)$;
\EndFor
\For{$i := 0$ {\bf to} $l-1$}
    \State $\shortest(k,b^\prime_i,g^\prime_i) := GetMax(Q) + c_k i$;
    \If{$i \geq (u_r)_k$}
        \State $Dequeue(Q)$;
    \EndIf
    \State $new\_pos := (i + \abs{(u_l)_k} + 1)$;
    \If{$new\_pos \leq l - 1$}
        \State $Enqueue(Q, \shortest(k-1,h_{new\_pos}, f_{new\_pos}) - c_k new\_pos)$;
    \EndIf
\EndFor
\end{algorithmic}
\end{algorithm}

The correctness of the algorithm follows from the formula \eqref{shortest_path}. The algorithm's complexity is $O(l)$.

Let us estimate the total arithmetic complexity of the whole procedure. It consists from the following parts:
\begin{enumerate}
    \item Enumerating of points in the set $M$. Due to Corollary \ref{DP_width_cor}, the complexity of this part is $O(\log m)^{m^2} \cdot (\chi+m)^m \cdot \Delta \cdot \abs{\det(S)}$; 
    \item Constructing the graphs $F_k$, for each $k \in \intint n$. The number of edges and vertices in $F_k$ can be estimated as $O(\abs{M} \cdot \abs{\det(S)})$. Hence, due to Corollary \ref{DP_width_cor}, the complexity of this part can be estimated as $$O(n \cdot \abs{M} \cdot \abs{\det(S)}) = O(n \cdot O(\chi+m)^m \cdot \Delta \cdot \abs{\det(S)}).$$
    \item For each $F_k$, compute a path or a cycle decomposition of $F_k$. For each path or cycle in the decomposition, apply one of the algorithms above. The complexity of this part is clearly the same as in the previous step. 
\end{enumerate}

Therefore, in the assumption that the complexity of any edge or vertex is $O(1)$, the total complexity of the shortest path computation is 
$$
n \cdot O(\log m)^{m^2} \cdot (\chi + m)^m \cdot \Delta \cdot \abs{\det(S)}.
$$

Again, since any column $A_k$ has $m$ components and and $G_k \bmod S$ has at most $\log_2\abs{\det(S)}$ components, the arithmetic complexity of an edge or a vertex traversal is $m + \log_2\abs{\det(S)}$. Hence, the total arithmetic complexity of the algorithm can be roughly estimated as 
$$
n \cdot O(\log m)^{m^2} \cdot (\chi + m)^m \cdot \Delta \cdot \abs{\det(S)} \cdot \bigl( m + \log\abs{\det(S)}\bigr).
$$
\end{proof}

Again, using Lemma \ref{ILPCF_to_ILPSF_lm} we can formulate analogues result for the \ref{BILP-CF} problem.
\begin{corollary}\label{BILP_CF_complexity_cor}
The \ref{BILP-CF} problem can be solved by an algorithm with the following arithmetic complexity bound:
\begin{equation*}
    T_{LP} + n \cdot O(\chi + m)^{m+1} \cdot \log^2(\chi+m) \cdot \Delta \cdot \bigl( m + \log(\Delta_{\gcd})\bigr).
\end{equation*}

The previous complexity bound can be slightly improved in terms of $\chi$:
\begin{gather*}
    T_{LP} + n \cdot O(\log m)^{m^2} \cdot (\chi + m)^m \cdot \Delta \cdot \bigl( m + \log(\Delta_{\gcd})\bigr).
\end{gather*}

Here $\Delta = \Delta(A)$ and $\Delta_{\gcd} = \Delta_{\gcd}(A)$.
\end{corollary}

\subsection{Algorithms For The Unbounded Problems}\label{unbounded_complexity_subs}

The paper \cite{DiscConvILP} presents an elegant way how to apply results of the discrepancy theory for the \ref{CLASSIC-ILP-SF} problem parameterized by $m$ and $\Delta_1(A)$. In the current Subsection, we generalize the original algorithm of \cite{DiscConvILP} to work with the \ref{ILP-CF} and \ref{ILP-SF} problems.

\begin{definition}
For a matrix $A \in \RR^{m \times n}$ we define \emph{its discrepancy and its hereditary discrepancy} by the formulas
\begin{gather*}
\disc(A) = \min_{z \in \{-1/2,\, 1/2\}^n} \left\| A z  \right\|_\infty,\\
\herdisc(A) = \max_{\IC \subset \intint n} \disc(A_{* \IC}).
\end{gather*}

Due to the works \cite{HerDisc} and \cite{SixDeviations_Spencer} it is known that
\begin{equation}\label{disc_estimate}
\herdisc(A) \leq 2 \disc(A) \leq \eta_m \cdot \Delta_1(A), \quad\text{where}\quad \eta_m \leq 12 \cdot \sqrt{m}.
\end{equation}

The following key Lemma, due to \cite{DiscConvILP}, connects results of the discrepancy theory with the ILP problem.

\begin{lemma}[\cite{DiscConvILP}]\label{disc_lm} Let $x \in \ZZ^n_{\geq 0}$. Then, there exists a vector $z \in \ZZ^n_{\geq 0}$ with $z_i \leq x_i$, $\frac{1}{6} \cdot \|x\|_1 \leq \|z\|_1 \leq \frac{5}{6} \cdot \|x\|_1$, for each $i \in \intint n$, and
$$
\left\|A(z - \frac{x}{2})\right\|_\infty \leq 2 \cdot \herdisc(A).
$$
\end{lemma}

\begin{remark}\label{n_delta_bound_rm}
Before presenting an algorithm for the \ref{ILP-SF} problem, parameterized by $m$, $\Delta$, and $\abs{\det(S)}$, we can assume that $n \leq (m^2 + m) \cdot \Delta^2 \cdot \abs{\det(S)}$. Definitely, due to the main result of \cite{ModularDiffColumns}, the number of columns of the matrix $A$, in the assumption that there are no pairs of columns $u,v$ with $u = \pm v$, is bounded by $\frac{1}{2}(m^2 + m)\cdot \Delta$. Consequently, if $n > (m^2 + m) \cdot \Delta^2 \cdot \abs{\det(S)}$, then the system will contain a pair of equal columns (the last coordinates of the columns are equal modulo $S$). 
Using any $O(n \log n)$-sorting algorithm, we can make all the columns to be unique.
\end{remark}

\begin{theorem}\label{ILP_SF_complexity_th}
Let us consider the \ref{ILP-SF} problem. Let $\Delta = \Delta(A)$ and $\rho \in \ZZ_{>0}$ be the value such that $\|z^*\|_1 \leq (6/5)^\rho$ for some optimal integer solution $z^*$ of the problem. Additionally, let $B$ be some non-degenerate $n \times n$ sub-matrix of $A$ and $\kappa = \Delta / \abs{\det(B)}$. Then, the problem can be solved by an algorithm with the complexity bound
$$
\rho \cdot m \cdot O(\kappa \cdot \eta_m)^{2m} \cdot \Delta^2 \cdot \abs{\det(S)}^2 \cdot \log \abs{\det(S)}.
$$
\end{theorem}
\begin{proof} 
W.l.o.g. we can assume that $A = B \bigl(I \, U\bigr)$, where the matrix $U$ has the property that all the minors of $U$ are at most $\kappa$ by an absolute value. Particularly, $\Delta_1(U) \leq \kappa$.  Define $$
\MC(i,\gamma) = \ZZ^m \cap \{ B x \colon \|x - 2^{i - r} B^{-1} b\|_\infty \leq \gamma \}.
$$ 

Let $\mu = \eta_m \cdot \Delta_1(U) \leq \kappa \cdot \eta_m$. For every $i \in \intint[0]{\rho}$, $b^\prime \in \MC(i, 4 \mu)$ and every $g^\prime \in \ZZ^{n-m} \bmod\, S$, we solve
\begin{gather}
c^\top x \to \min\notag\\
\begin{cases}
A x = b^{\prime}\\
G x \equiv g^{\prime}\\
\|x\|_1 \leq \left(\frac{6}{5}\right)^i\\
x \in \ZZ^n_{\geq 0}.
\end{cases}\tag{$DP(i,b^\prime,g^\prime)$}\label{norm_prob}
\end{gather}

As in \cite{DiscConvILP}, we iteratively derive solutions for $i$, using pairs of solutions for $i-1$. Finally, we will compute an optimal solution for $i := \rho$, $b^\prime := b$ and $g^\prime := g$.

The case $i = 0$ is trivial, since $\|x\|_1 \leq 1$, and such a solution corresponds exactly to the columns of $\binom{A}{G}$. Let us fix some $i > 0$, $b^\prime \in \MC(i, 4 \mu)$, $g^\prime \in \ZZ^{n-m} \bmod\, S$, and let $x^*$ be an optimal solution of \ref{norm_prob}. The equality $A x^* = b^\prime$ can be rewritten as $\bigl(I \; U\bigr) x^* = B^{-1} b^\prime$. 

By Lemma \ref{disc_lm}, there exists a vector $0 \leq z \leq x^*$ with $\|\bigl(I \, U\bigr) z - \frac{1}{2} B^{-1} b^\prime\|_{\infty} \leq 2 \mu$ and 
$$
\|z\|_1 \leq \frac{5}{6} \cdot \|x^*\|_1 \leq \frac{5}{6} \cdot \left(\frac{6}{5}\right)^i = \left(\frac{6}{5}\right)^{i -1},
$$
if $\|x^*\|_1 > 1$, or $\|z\|_1 \leq \|x^*\|_1 \leq 1 \leq (6/5)^{i-1}$, otherwise. The same holds for $x^* - z$. Therefore, $z$ is an optimal solution for $DP(i-1,b^{\prime\prime}, g^{\prime\prime})$, where $b^{\prime\prime} = A z$ and $g^{\prime\prime} = G z \bmod S$. Likewise, $x^* - z$ is an optimal solution for $DP(i-1,b^\prime - b^{\prime\prime}, g^\prime - g^{\prime\prime})$. We claim that $b^{\prime\prime} \in \MC(i-1,4 \mu)$ and $b^{\prime} - b^{\prime\prime} \in \MC(i-1, 4 \mu)$. This implies that we can look up solutions for $DP(i-1, b^{\prime\prime}, g^{\prime\prime})$ and $DP(i-1,b^\prime - b^{\prime\prime}, g^\prime - g^{\prime\prime})$ in the dynamic table and their sum is a solution for $DP(i,b^\prime, g^\prime)$.

Let us prove the claim: \begin{multline*}
    \left\|B^{-1} b^{\prime\prime} - 2^{(i-1)-r} B^{-1} b\right\|_{\infty} = \left\|\bigl(I \; U\bigr) z - \frac{1}{2} B^{-1} b^{\prime} + \frac{1}{2} B^{-1} b^{\prime} - 2^{(i-1)-r} B^{-1} b\right\|_{\infty} \leq \\
    \leq \left\|\bigl(I \; U\bigr) z - \frac{1}{2} B^{-1} b^{\prime}\right\|_{\infty} + \left\|\frac{1}{2} B^{-1} b^{\prime} - 2^{(i-1)-r} B^{-1} b\right\|_{\infty} \leq \\
    \leq 2 \cdot \mu + \frac{1}{2} \left\|B^{-1} b^{\prime} - 2^{i - r} B^{-1} b\right\|_{\infty} \leq 4 \cdot \mu.
\end{multline*}

Let us estimate the algorithm complexity. By Lemma \ref{par_enum_lm}, $\abs{\MC(i,4 \mu)} \leq O(\mu)^m \cdot \abs{\det(B)}$ and the elements of $\MC(i,4 \mu)$ can be enumerated by an algorithm with the complexity bound 
$$
m \cdot O(\mu)^m \cdot \abs{\det(B)} \cdot \min\{\log \abs{\det(B)}, m\} = m^2 \cdot O(\mu)^m \cdot \abs{\det(B)}.
$$

Let us fix $i > 1$. For any $b^\prime \in \MC(i,4 \mu)$ and $g^\prime \in \ZZ^{n-m} \bmod\, S$, we compute $DP(i,b^\prime,g^\prime)$ by the formula
\begin{equation}\label{DP_formula_unb}
DP\bigl(i,b^\prime,g^\prime\bigr) = \min_{
\substack{
b^{\prime\prime} \in \MC(i-1,4 \mu)\\
g^{\prime\prime} \in \ZZ^{n-m} \bmod\, S
}
} \{DP\bigl(i,b^{\prime\prime},g^{\prime\prime}\bigr) + DP\bigl(i,b^\prime-b^{\prime\prime},(g^\prime-g^{\prime\prime}) \bmod S\bigr)\}.
\end{equation}

Let us assume that work with $b^\prime, g^\prime, b^{\prime\prime}, g^{\prime\prime}$ and access to the dynamic table needs $O(1)$ time. Then, the algorithm complexity consists from the following components:
\begin{enumerate}
    \item Construction of the sets $\MC(i, 4 \mu)$, for $i \in \intint[0]{r}$. As it was noted, the complexity of this step is $\rho \cdot m^2 \cdot O(\mu)^m \cdot \abs{\det(B)}$.
    \item Solving of $DP(0,*,*)$ problems. As it was noted, solutions of these sub-problems correspond to the columns of the matrix $\binom{A}{G}$. So, in our assumptions, the complexity of this step can be estimated as $O(n)$.
    \item For each $k \in \intint \rho$, $b^\prime \in \MC(k,4\mu)$, $g^\prime \in \ZZ^{n-m} \bmod\, S$, computation of $DP(k,b^\prime, g^\prime)$ using Formula \eqref{DP_formula_unb}. The complexity of this component can be estimated as $\rho \cdot O(\mu)^{2m} \cdot \Delta^2 \cdot \abs{\det(S)}^2$.
\end{enumerate}

Assuming that $m^2 \cdot O(\mu)^m = O(\mu)^{2 m}$, the dominating therm of all the components is
$$
\rho \cdot O(\mu)^{2m} \cdot \Delta^2 \cdot \abs{\det(S)}^2.
$$

Finally, since any column $A_k$ has $m$ components and $G_k \bmod S$ has at most $\log_2\abs{\det(S)}$ components, the total arithmetic complexity can be estimated as
$$
\rho \cdot m \cdot O(\mu)^{2m} \cdot \Delta^2 \cdot \abs{\det(S)}^2 \cdot \log \abs{\det(S)}.
$$
\end{proof}

\begin{corollary}\label{ILP_SF_complexity_cor}
Let us consider the \ref{ILP-SF} problem. Let $\Delta = \Delta(A)$ and $\|x^*\|_1 \leq (\frac{6}{5})^\rho$, for some optimal integer solution of the problem. Then, the  problem  can  be  solved  by  an  algorithm  with the arithmetic complexity bound 
$$
O(\log m)^{2 m^2} \cdot m^{m+1} \cdot \Delta^2 \cdot \abs{\det(S)}^2 \cdot \log \abs{\det(S)} \cdot \rho.
$$ 

Using a polynomial reformulation, $\rho$ can be chosen as
$$
\rho = O\bigl( \log(m \cdot \Delta \cdot \abs{\det(S)}) \bigr).
$$
\end{corollary}
\begin{proof}
Let us assume that the relaxed LP problem is bounded, \GribanovAdd{the unbounded case can be handled by the same way as in the proof of Theorem \ref{BILP_SF_complexity_th}}. We search for a non-singular $n \times n$ sub-matrix $B$ of $A$, such that $\Delta/\abs{\det(B)} = O(\log m)^m$. Due to the main result of \cite{SubdeterminantApprox}, it can be done by a polynomial-time algorithm. Additionally, due to formula \eqref{disc_estimate}, we have $\kappa \cdot \eta_m = O(\log m)^m \cdot \sqrt{m}$. 

Connecting all the things together, the arithmetic complexity bound becomes
$$
O(\log m)^{2 m^2} \cdot m^{m+1} \cdot \Delta^2 \cdot \abs{\det(S)}^2 \cdot \log \abs{\det(S)} \cdot \rho.
$$

Now, let estimate $\rho$, using Theorem \ref{ILP_SF_proximity_cor}. It states that there exists an optimal vertex solution $x^*$ of the LP relaxation and an optimal integer solution $z^*$, such that 
$$
\|z^* - x^*\|_1 \leq \chi := O\bigl(\Delta^* \cdot m^2 \cdot \log(2 \sqrt[m]{\Delta^*})\bigr),
$$ where $\Delta^* = \Delta \cdot \abs{\det(S)}$. Following to \cite{DiscConvILP}, we change the variables $x^\prime = x - y$, where $y_i = \min\{0, \lceil x^*_i \rceil - \chi\}$. Note that $z^* \geq y$, and since $x^*$ has at most $m$ non-zero components, we have $\|z^* - y\|_1 = O(m \chi)$. Since $z^* \geq y$, after the change of variables $x^\prime = x - y$, the original problem transforms to a similar problem, but with
$$
\rho = O\bigl(\log(m \chi)\bigr) = O\bigl(\log(m \cdot \Delta \cdot \abs{\det(S)}) \bigr).
$$
\end{proof}



Again, using Lemma \ref{ILPCF_to_ILPSF_lm} we can formulate a similar result for the \ref{BILP-CF} problem.
\begin{corollary}\label{ILP_CF_complexity_cor}
Let us consider the \ref{ILP-CF} problem. Let $\Delta = \Delta(A)$, $\Delta_{\gcd} = \Delta_{\gcd}(A)$ and $\|x^*\|_1 \leq (\frac{6}{5})^\rho$ for some optimal integer solution of the problem. Then, the problem can be solved by an algorithm with the arithmetic complexity bound
$$
O(\log m)^{2 m^2} \cdot m^{m+1} \cdot \Delta^2 \cdot \log(\Delta_{\gcd}) \cdot \rho.
$$

Using a polynomial reformulation, $\rho$ can be chosen as
$$
\rho = O\bigl( \log(m \cdot \Delta) \bigr).
$$
\end{corollary}

\end{definition}

\section{The Case $m = 0$ And Applications To The Knapsack Problem, Subset-Sum Problem, And An Average Complexity Of The ILP Problems}\label{partial_cases_sec}

\subsection{The ILP Problems With A Square ($m = 0$) Constraints Matrix And Their Relation To Group Minimization Problems}\label{group_subs}

In this Section, we consider the \ref{BILP-CF} and \ref{ILP-CF} problems with $m = 0$.
\begin{equation*}
\begin{gathered}
c^\top x \to \max\\
\begin{cases}
b_l \leq A x \leq b_r\\
x \in \ZZ^n
\end{cases}
\end{gathered} 
\quad    
\begin{gathered}
c^\top x \to \max\\
\begin{cases}
A x \leq b\\
x \in \ZZ^n.
\end{cases}
\end{gathered} 
\end{equation*}
We denote $\Delta = \abs{\det(A)}$. Clearly, $\Delta = \Delta_{\gcd}(A)$, for the case $m = 0$.

Due to Lemma \ref{ILPCF_to_ILPSF_lm}, the last problems are equivalent to the problems
\begin{gather*}
w^\top x \to \min\notag\\
\begin{cases}
G x \equiv g \pmod{S \cdot \ZZ^n}\\
x \in \ZZ^{n}_{\geq 0}
\end{cases}\label{un_ILP_group_prob}\tag{U-GROUP-ILP}
\end{gather*} 

\begin{gather*}
w^\top x \to \min\notag\\
\begin{cases}
G x \equiv g \pmod{S \cdot \ZZ^n}\\
0 \leq x \leq u\\
x \in \ZZ^{n},\label{bn_ILP_group_prob}\tag{B-GROUP-ILP}
\end{cases}
\end{gather*} 
where $S$ is the SNF of $A$ and $\abs{\det(S)} = \Delta$.

Both problems \ref{un_ILP_group_prob} and \ref{bn_ILP_group_prob} are special cases of the general group minimization problems \ref{un_group_prob} and \ref{bn_group_prob}, respectively. The following known theorem, due to R.~Gomory, gives an algorithm to solve unbounded group problems, and, additionally, gives a bound on the size of an optimal point.
\begin{theorem}[ R.~Gomory \cite{GomoryRelation}, see also Chapter~19 of \cite{HuBook} ]\label{Gomory_th}
Let $\GC$ be a finite Abelian group and $g_0, g_1, \dots, g_n \in \GC$. Let us consider the polyhedron $\PC$, defined as the convex hull of solutions of the following system
\begin{equation}\label{group_eq}
\begin{cases}
 \sum_{i = 1}^n x_i g_i = g_0 \\
 x \in \ZZ_{\geq 0}^n.
\end{cases}
\end{equation}

Then, for any vertex $v$ of $\PC$, the following inequality holds 
\begin{equation}\label{Gomory_ineq}
\prod_{i=1}^s (1 + v_i) \leq \abs{\GC}.    
\end{equation}

Moreover, for $c \in \ZZ^n_{\geq 0}$, the problem $\max\{c^\top x \colon x \in \PC\}$ can be solved by an algorithm with the group-operations complexity 
$$
O(\min\{n, \abs{\GC}\} \cdot \abs{\GC}).
$$
\end{theorem}

Since any group operation in the \ref{un_ILP_group_prob} problem has the complexity $\log_2 \abs{\det(S)} = \log_2(\Delta)$ and $\abs{\GC} = \abs{\det(S)} = \Delta$, the arithmetic complexity of the resulting minimization algorithm is
$$
O(n \cdot \Delta \cdot \log(\Delta)).
$$ Here, we assume that $n \leq \Delta$, because in the opposite case we have multiple elements in the group constraint that can be erased.

Additionally, from \eqref{Gomory_ineq} it holds  
$$
\|v\|_1 \leq \Delta-1,
$$ for some optimal integer solution $v$ of the problem.

Our Theorem \ref{BILP_SF_complexity_th} gives an algorithm with the same arithmetic complexity 
$$
O(n \cdot \Delta \cdot \log(\Delta))
$$ for both \ref{un_group_prob} and \ref{bn_group_prob} problems. Additionally, the norm of an optimal solution $v$ for any of this problems can be estimated using Theorem \ref{BILP_SF_proximity_th} of this paper (see Remark \ref{ILP_SF_proximity_rm}), it gives the same bound as \eqref{Gomory_ineq} for the \ref{un_group_prob} problem:
$$
\|v\|_1 \leq \Delta-1.
$$

In the following Theorem, we generalize Theorem \ref{Gomory_th} of R.~Gomory to characterise lower dimensional faces of $\PC$.
\begin{theorem}\label{faces_Gomory_th}
Let $\PC$ be the polyhedron from the definition of Theorem \ref{Gomory_th}. Let $\FC$ be a $d$-dimensional face of $\PC$ and $p \in \FC \cap \ZZ^n$, then there exists a set of indices $\JC \subseteq \intint n$ such that $\abs{\JC} \geq n - d$ and
$$
\prod_{i \in \JC} (1 + p_{i}) \leq \Delta.
$$
\end{theorem}
\begin{proof}
We call a point $v \in \PC \cap \ZZ^n$ as \emph{$k$-independent} if there exist $k$ linearly independent integer vectors $\{r_1,r_2, \dots, r_k\}$, which satisfy to the system 
\begin{equation}\label{k_indep_system}
\begin{cases}
G x = \BZero \\
0 \leq x \leq v\\
x \in \ZZ^n,
\end{cases}
\end{equation}
but existence of $k+1$ such vectors is impossible. Here $G x$ is a short notation for the sum $\sum_{i = 1}^n x_i g_i$.

It can be seen that $p$ is $k$-independent, for some $k \leq d$. Definitely, if there exist $d+1$ vectors $r_i$ with the described properties, then for each $i \in \intint n$ we have $G (p \pm r_i) = g_0$ and $(p \pm r_i) \geq 0$ that contradicts to the fact $p \in \FC$.

Now, let $\SC$ be the linear space induced by solutions of the system \eqref{k_indep_system} and assume that there exist a vector $p^\prime \in \ZZ^n_{\geq 0}$ such that $p^\prime \leq p$. Clearly, $p - p^\prime \in \SC$. Since $p$ is $k$-independent, there exist a matrix $R \in \ZZ^{n \times k}$ of rank $k$ such that $\SC = \linh(R)$. Consequently, there exist a vector $\alpha \in \QQ^k$ such that
$$
p^\prime = p - R \alpha.
$$
After a change of the $\alpha$ variables, we can declare that there exists a vector $\alpha \in \QQ^k$, such that
$$
p^\prime = p - \binom{I_{k \times k}}{Q} \alpha, 
$$ where $I_{k \times k}$ is the $k \times k$ identity matrix and $Q \in \QQ^{(n-k) \times k}$. Additionally, we have assumed that $R$ contains non-degenerate sub-matrix in the first $k$ rows.

Since $p,p^\prime$ are integer and $0 \leq p^\prime \leq p$ the vector $\alpha$ can be chosen by at most 
$
\prod_{i = 1}^k (p_i + 1)
$ variants.

Consequently, if $G p^\prime = G p = g_0$ for $0 \leq p^\prime \leq p$, then there are at most $\prod_{i = 1}^k (p_i + 1)$ variants to chose $p^\prime$. Hence,
$$
\prod_{i = 1}^n (p_i + 1) \leq \Delta \cdot \prod_{i = 1}^k (p_i + 1)
$$ that proves the theorem.

\end{proof}

Finally, we note that if the group $\GC$ is cyclic, then the \ref{un_group_prob} problem can be solved faster. For the  \ref{un_ILP_group_prob} problem, it holds, for example, when $\Delta$ is prime.
\begin{theorem}\label{cyclic_group_complexity_th}
Let $\GC$ be a finite Abelian cyclic group and $g_0, g_1, \dots, g_n \in \GC$. Let us consider the following problem
\begin{gather*}
c^\top x \to \min\\
\begin{cases}
 \sum_{i = 1}^n x_i g_i = g_0 \\
 x \in \ZZ_{\geq 0}^n,
\end{cases}
\end{gather*}
where $c \in \ZZ_{\geq 0}^n$.

Then, \GribanovAdd{all the problems with the all possible r.h.s. values $g_0 \in \GC$} can be solved by an algorithm with the group-operations complexity 
$$
O\bigl( T_{(\min,+)}(2 \abs{\GC}) \cdot \log \abs{\GC} + \abs{\GC} \cdot \log \abs{\GC} \bigr),
$$ where $T_{(\min,+)}(\cdot)$ is the complexity of the $(\min,+)$-convolution problem.
\end{theorem}
\begin{proof}
Due to Theorem \ref{Gomory_th}, we have $\|v\|_1 \leq \abs{\GC}-1$, for some optimal solution $v$ of the problem. For any $k \in \intint{\lceil \log_{3/2} \abs{\GC} \rceil}$ and $g^\prime \in \GC$, let us consider the problem $DP(k,g^\prime)$:
\begin{gather*}
    c^\top x \to \min\\
    \begin{cases}
     \sum_{i = 1}^n x_i g_i = g^\prime\\
     \|x\|_1 \leq (3/2)^k\\
     x \in \ZZ_{\geq 0}^n.
    \end{cases}
\end{gather*}
Clearly, the original problem is equivalent to $DP(\lceil \log_{3/2} \abs{\GC} \rceil, g_0)$.

Note that any vector $x \in \ZZ_{\geq 0}^n$ with $\|x\|_1 \leq (3/2)^k$ can be represented as the sum $x = x^\prime + x^{\prime\prime}$, for $x^\prime,x^{\prime\prime} \in \ZZ_{\geq 0}^n$ with $\|\cdot\|_1 \leq (3/2)^{k-1}$. Hence, the value of $DP(k,g^\prime)$, for $k \geq 2$, can be computed by the formula
$$
DP(k,g^\prime) = \min\limits_{g^{\prime\prime} \in \GC} \{DP(k-1,g^\prime-g^{\prime\prime}) + DP(k-1,g^{\prime\prime})\}.
$$ 

Since any Abelian cyclic group is isomorphic to $\ZZ_r$ (the ring of the numbers $\{0,1,\dots,r-1\}$ modulo $r$), we can assume that $\GC = \ZZ_r$, where $r = \abs{\GC}$. Fix $k \geq 2$ and define $\alpha_j = DP(k,j)$, for $j \in \intint[0]{r-1}$. Let $\beta$ be the result of the $(\min,+)$-convolution of the sequences $\alpha\alpha$ and $\alpha\alpha$, where $\alpha\alpha$ is the concatenation of $\alpha$ with itself. Then, for $s \in \intint[0]{r-1}$, we have
\begin{multline*}
    \beta_{s + r} = \min\limits_{i \in [0,s+r]}\{(\alpha\alpha)_i + (\alpha\alpha)_{s+r-i}\} = \\
    = \min\limits_{i \in \ZZ_r} \{ DP(k,i) + DP(k,s-i)\}.
\end{multline*}

Hence, $\beta_{s + r} = DP(k,s)$ and the computation of $DP(k \geq 2,\cdot)$ can be reduced to the $(\min,+)$-convolution. 

Let us consider the case $k = 1$. Clearly, for any $g^\prime \in \GC\setminus\{0\}$, the problem $DP(1,g^\prime)$ is equivalent to the problem of searching $g^\prime$ in the list $\{g_1,g_2, \dots, g_n\}$. After the preprocessing that takes $O(n \cdot \log n)$ time, it can be done by the binary search in $O(\log n)$ time. In assumption that $n \leq \Delta$, the computation of $DP(1,\cdot)$ takes $O(\Delta \cdot \log(\Delta))$ time.


Clearly, the full computational complexity is the same that is stated in the Theorem.

\end{proof}

Due to \cite{3SumViaAdditioveComb} and \cite{APSPViaCircuitComplexity}, $T_{(\min,+)}(n) = n^2/2^{\Omega(\sqrt{\log n})}$, so the problem \ref{un_group_prob} with a cyclic group can be solved by an algorithm with the group-operations complexity bound
$$
\Delta^2 \cdot \frac{1}{2^{\Omega(\sqrt{\log (\Delta)})}},
$$ which is faster, than $O(n \cdot \Delta \cdot \log(\Delta))$, due to Theorem \ref{Gomory_th}, for example, when $n = \Theta(\Delta)$.

\subsection{The Class Of \emph{Local Problems}}\label{local_subs}

\begin{definition} 
Let us consider the \ref{ILP-CF} problem. Let $\BC$ be some optimal base of the corresponding relaxed LP problem. The \ref{ILP-CF} problem is \emph{local} if the set of its optimal solutions coincides with the set of optimal solutions of the following local problem 
\begin{gather}
c^\top x \to \max\notag\\
\begin{cases}
A_{\BC} x \leq b_{\BC}\label{local_prob}\tag{LOCAL-ILP}\\
x \in \ZZ^n
\end{cases}
\end{gather}
that is induced from the original problem by omitting all the inequalities, except $A_{\BC} x \leq b_{\BC}$.
\end{definition}

\begin{lemma}\label{long_dist_lm}
Consider the \ref{ILP-CF} problem, let $\Delta = \Delta(A)$, $\BC$ and $v = A_{\BC}^{-1} b_{\BC}$ be the corresponding LP base and optimal LP vertex, and $\NotBC = \intint{m+n} \setminus \BC$. Let, additionally, $z_{loc}$ be any optimal vertex of $\PC_I(A_{\BC},b_{\BC})$. 

Then, $\|A_{\NotBC} (z_{loc} - v)\|_\infty \leq \Delta-1$. 
\end{lemma}
\begin{proof}
Let $y = b_{\BC}-A_{\BC} z_{loc}$.  Due to Lemma \ref{ILPCF_to_ILPSF_lm}, the \ref{local_prob} problem is equivalent to the group minimization problem \ref{un_ILP_group_prob}. Hence, by Theorem \ref{Gomory_th}, 
$$
\prod_{i = 1}^n (1 + y_i) \leq \Delta,
$$ 
and, consequently, $\|y\|_1 \leq \Delta-1$. Finally, for any $i$,
$$
\abs{(A_{\NotBC})_{i\,*} (z_{loc}-v)} = \abs{(A_{\NotBC})_{i\,*}A_{\BC}^{-1} y} \leq \frac\Delta\delta \BUnit^\top y < \Delta.
$$
\end{proof}

\begin{corollary}\label{long_dist_cor}
If $b_{\NotBC} - A_{\NotBC} v \geq (\Delta-1) \BUnit$, then the \ref{ILP-CF} problem is local for any $c \in \cone(A_{\BC}^\top)$.
\end{corollary}
\begin{proof}
Let $z_{loc}$ be any optimal vertex of $\PC_I(A_{\BC},b_{\BC})$ with respect to $c$. We are going to show that $z_{loc}$ satisfies the inequalities $A_{\NotBC} x \leq b_{\NotBC}$. It implies that we can solve the \ref{local_prob} problem instead of the  \ref{ILP-CF} problem and that the \ref{ILP-CF} problem is local.

Due to previous Lemma \ref{long_dist_lm}, we have $\abs{A_{\NotBC} (v-z_{loc})} \leq (\Delta-1) \BUnit$. Due to the Corollary assumptions, we have $$b_{\NotBC} - A_{\NotBC} z_{loc} = A_{\NotBC} (v - z_{loc}) + b_{\NotBC} - A_{\NotBC} v \geq \BZero.$$
\end{proof}

\begin{theorem}\label{local_properties_th}
Let $z$ be an integer optimal vertex of the local \ref{ILP-CF} problem, $\BC$ be an optimal base of the relaxed problem, $\NotBC = \intint{n+m} \setminus \BC$, $v = A_{\BC}^{-1} b_{\BC}$, $\Delta = \Delta(A)$ and $\Delta_{\BC} = \abs{\det(A_{\BC})}$. Then, the following statements hold:
\begin{enumerate}
\item $\|A(v - z)\|_0 \leq \log_2(\Delta_{\BC}) + m$, \\$\|A(v - z)\|_1 \leq (\Delta_{\BC}-1) + m(\Delta-1)$,\\ $\|A(v - z)\|_\infty \leq \Delta-1$;
\item $\|b - A z\|_0 \leq \log_2(\Delta_{\BC}) + m$, \\$\|b - A z\|_1 \leq (\Delta_{\BC}-1) + m(\Delta-1) + \|b_{\NotBC} - A_{\NotBC} v\|_1$,\\
$\|b - A z\|_\infty \leq \Delta-1 + \|b_{\NotBC} - A_{\NotBC} v\|_\infty$;
\item The point $z$ lies on a face of $\PC$, whose dimension is bounded by $\log_2(\Delta_{\BC})$;
\item The problem can be solved by an algorithm with the arithmetic complexity bound 
$$
O(n \cdot \Delta_{\BC} \cdot \log(\Delta_{\BC})).
$$
\end{enumerate}

If, additionally, the system $A x \leq b$ is $v$-normalized, then

\begin{enumerate}
\item $\|z\|_0 \leq 2 \log_2(\Delta_{\BC})$, $\|z\|_1 < 2 \Delta_{\BC}^2 \log_2(\Delta_{\BC})$, $\|z\|_\infty < \Delta_{\BC}^2$;
\item $\|v-z\|_0 \leq 2 \log_2(\Delta_{\BC})$, $\|v-z\|_1 < \Delta_{\BC}^2 \log_2(\Delta_{\BC})$, $\|v-z\|_\infty < \Delta_{\BC}^2/2$.
\end{enumerate}
\end{theorem}
\begin{proof}
Let $y = b_{\BC} - A_{\BC} z = A_{\BC} (v - z)$. By the inequality \eqref{Gomory_ineq}, we have $\|y\|_0 \leq \log_2(\Delta_{\BC})$ and $\|y\|_1 \leq \Delta_{\BC}-1$.

Since $A(v-z) = \dbinom{y}{A_{\NotBC} (v-z)}$, the inequalities of Statement 1 follow from Lemma \ref{long_dist_lm}.

The inequalities from Statement 2 follow from the equalities $b - A z = A (v - z) + b - A v$ and $b_B - A_B v = \BZero$.

Statement 3 follows from the inequality $\|y\|_0 \leq \log_2(\Delta_{\BC})$ and from the locality property.

Now, assume that the system $A x \leq b$ is $v$-normalized. Due to Lemma \ref{adj_lm}, $\|v-z\|_0 = \|A_{\BC}^{-1} y\|_0 \leq 2 \log_2(\Delta_{\BC})$. Note that $\|A_{\BC}^{-1} b_{\BC}\|_0 \leq \log_2(\Delta_{\BC})$, so $\|z\|_0 = \|A_{\BC}^{-1} (b_{\BC} - y)\|_0 \leq 2 \log_2(\Delta_{\BC})$.

Since $\|b_{\BC} - y\|_1 \leq 2 (\Delta_{\BC}-1)$, by Lemma \ref{adj_lm}, we have
\begin{gather*}
\|z\|_\infty \leq \|A_{\BC}^{-1} (b_{\BC} - y)\|_\infty \leq (\Delta_{\BC}-1) \Delta_{\BC} < \Delta_{\BC}^2 \text{ and }\\
\|v-z\|_\infty \leq \|A_{\BC}^{-1} y\|_\infty  \leq \frac{(\Delta_{\BC} - 1) \Delta_{\BC}}{2} < \Delta_{\BC}^2/2.
\end{gather*} Now, the inequalities $\|z\|_1 < 2 \Delta_{\BC}^2 \log_2(\Delta_{\BC})$ and $\|v-z\|_1 < \Delta_{\BC}^2 \log_2(\Delta_{\BC})$ are trivial.


\end{proof}

\subsection{Applications For The Simplex Integer Feasibility, Unbounded Knapsack, and Unbounded Subset-Sum Problems}\label{group_app_subs}

Let $A \in \ZZ^{(n+1) \times n }$, $b \in \ZZ^{n+1}$, and $\PC = \PC(A,b)$ be a simplex. Let $\BC$ be some base of $A$, $v_{\BC} = A_{\BC}^{-1} b_{\BC}$ be the corresponding vertex and $\Delta_{\BC} = \abs{\det(A_{\BC})}$. Consider the remaining inequality $a^\top x \leq a_0$ of the system $A x \leq b$. Then, clearly, $\PC \cap \ZZ^n \not= \emptyset$ if and only if, then the optimum value of the problem
\begin{gather*}
a^\top x \to \max\\
\begin{cases}
A_{\BC} x \leq b_{\BC}\\
x \in \ZZ^n
\end{cases}
\end{gather*}
is at most $a_0$.

Moreover, if  $\PC \cap \ZZ^n \not= \emptyset$, then the minimum of $a^\top x$ attains at some vertex of $\PC_I(A_{\BC},b_{\BC})$. Consequently, from results of Subsection \ref{group_subs}, we have
\begin{theorem}\label{simplex_feasibility_th}
The feasibility problem for the intersection of the simplex $\PC$ with the integer lattice $\ZZ^n$ can be solved by an algorithm with the arithmetic complexity bound
$$
O(n \cdot \Delta_{\min} \cdot \log(\Delta_{\min})),
$$ where $\Delta_{\min}$ is minimal $n \times n$ minor of $A$ by an absolute value.

Moreover, for any vertex $v_{\BC}$ of $\PC$ and the corresponding base $\BC$ there exist a vertex $z$ of $\PC_I = \PC_I(A,b)$ such that:
\begin{enumerate}
\item The vertex $z$ lies on a $d$-dimensional face of $\PC$ with $d \leq \log_2(\Delta_{\BC})$;
\item $\|A (v_{\BC} - z)\|_0 \leq 1 + \log_2(\Delta_{\BC})$, \\$\|A (v_{\BC} - z)\|_1 \leq \Delta_{\BC} + \Delta - 2$, \\$\|A (v_{\BC} - z)\|_\infty \leq \Delta-1$;
\item $\|b - A z\|_0 \leq 1 + \log_2(\Delta_{\BC})$.
\end{enumerate}

In particular, $\PC_I$ is always contains a vertex $z$ with $\|b - A z\|_0 \leq 1 + \log_2(\Delta_{\min})$ and etc.

Additionally, if the system $A x \leq b$ is $v$-normalized, then
\begin{enumerate}
\item $\|z\|_0 \leq 2 \log_2(\Delta_{\BC})$, $\|z\|_1 < 2 \Delta_{\BC}^2 \log_2(\Delta_{\BC})$, $\|z\|_\infty < \Delta_{\BC}^2$,
\item $\|v_{\BC}-z\|_0 \leq 2 \log_2(\Delta_{\BC})$, $\|v_{\BC}-z\|_1 < \Delta_{\BC}^2 \log_2(\Delta_{\BC})$, $\|v_{\BC}-z\|_\infty < \Delta_{\BC}^2/2$.
\end{enumerate}

\end{theorem}

Now, let us consider the unbounded Knapsack and Subset-Sum problems:
\begin{gather}
c^\top x \to \max \label{un_knapsack}\tag{U-KNAP}\\
\begin{cases}
w^\top x = W\\
x \in \ZZ^n_{\geq 0},
\end{cases}\label{un_subset}\tag{U-SUBSET-SUM}
\end{gather}
where $c,w \in \ZZ^n_{>0}$ and $W \in \ZZ_{>0}$. Let, additionally, $w_{\min} = \min_i\{w_i\}$, $w_{\max} = \max_i\{w_i\}$. Let additionally $w_{opt}$ be the weight of an item with an optimal relative cost:
$$
\frac{c_j}{w_j} = \min_{i \in \intint{n}} \left\{ \frac{c_i}{w_i} \right\} \quad\text{and}\quad w_{opt} := w_j,
$$ or, in other words, $w_{opt}$ corresponds to the optimal solution of the relaxed LP Knapsack problem.

We are going to show two facts with very short proofs:
\begin{theorem}\label{subset_sum_th}
The \ref{un_subset} problem can be solved by an algorithm with the following arithmetic complexity bound
$$
O\bigl(w^2_{\min} \cdot \frac{1}{2^{\Omega(\sqrt{\log(w_{\min})})}}\bigr).
$$

Moreover, the \ref{un_subset} problem's polyhedron contains a vertex $z \in \ZZ^n_{\geq 0}$ with
$$
\|z\|_0 \leq \log_2(w_{\min}) + 1, \quad \|z\|_1 \leq w_{\min} + w_{\max} - 2, \quad \|z\|_{\infty} \leq w_{\max} - 1.
$$
\end{theorem}

\begin{theorem}\label{knapsack_th}
If $W \geq w^2_{opt}$, then the \ref{un_knapsack} problem can be solved by algorithms with the arithmetic complexity bounds
$$
O(n \cdot w_{opt}) \quad\text{and}\quad O\bigl(w^2_{opt} \cdot \frac{1}{2^{\Omega(\sqrt{\log(w_{opt})})}}\bigr).
$$

Consequently, any \ref{un_knapsack} problem can be solved by an algorithm with the arithmetic complexity $O(n \cdot w^2_{opt})$.
\end{theorem}

Proof of Theorem \ref{subset_sum_th}:
\begin{proof} Assume that $w_{\min} = w_1$. Then, the original \ref{un_subset} problem is equivalent to the group minimization problem
\begin{gather*}
x_1 \to \max\\
\begin{cases}
\sum_{i = 2}^n w_i x_i \equiv W \pmod{w_1}\\
x \in \ZZ^{n}_{\geq 0}.
\end{cases}
\end{gather*}

The last problem is clearly equivalent to the group minimization problem
\begin{gather*}
-W + \sum_{i = 2}^n w_i x_i \to \min\\
\begin{cases}
\sum_{i = 2}^n w_i x_i \equiv W \pmod {w_1}\\
x \in \ZZ^{n-1}_{\geq 0}.
\end{cases}
\end{gather*}

Since the group is cyclic, due to Theorem \ref{cyclic_group_complexity_th}, the problem \ref{un_subset} can be solved for the all r.h.s. $W \in \intint[0]{w_1-1}$ by an algorithm with the arithmetic complexity $O\bigl(w^2_{1} \cdot \frac{1}{2^{\Omega(\sqrt{\log(w_{1})})}}\bigr)$.

Finally, due to Lemma \ref{ILPSF_to_ILPCF_lm}, the original problem corresponds to the integer feasibility problem on simplex, defined by a system $A x \leq b$, for $A \in \ZZ^{n \times (n-1)}$ and $b \in \ZZ^{n-1}$.

Due to Theorem \ref{perp_matricies_th}, there exists a base $\BC$ with $\abs{\det(A_{\BC})} = w_{\min}$. Hence, by Theorem \ref{simplex_feasibility_th}, there exists an integer vertex $z$ with slacks $y = b - A z$, such that  
$$
\|y\|_0 \leq 1 + \log_2(w_{\min}), \quad \|y\|_1 \leq w_{\min} + w_{max} -2, \quad \|y\|_{\infty} \leq w_{\max}-1.
$$

Since the slacks $y$ correspond to variables of the original \ref{un_subset} problem, the proof is finished.
\end{proof}

Proof of Theorem \ref{knapsack_th}.
\begin{proof}
Assume, that $w_1 = w_{opt}$. Again, due to Lemma \ref{ILPSF_to_ILPCF_lm}, the original problem is equivalent to the problem 
\begin{gather}
\hat c^\top x \to \max\notag\\
\begin{cases}
A x \leq b\\
x \in \ZZ^{n-1},
\end{cases}\label{reduced_simplex_prob}
\end{gather}
where $A \in \ZZ^{n \times (n-1)}$, $b \in \ZZ^n$, and $\hat c \in \ZZ^{n-1}$. Due to Theorem \ref{perp_matricies_th}, there exists a base $\BC$, such that $\abs{\det(A_{\BC})} = w_{1}$. Let $a^\top x \leq a_0$ be the remaining inequality of the system and $v = A_{\BC}^{-1} b_{\BC}$. Then, by Corollary \ref{long_dist_cor}, if $a_0 - a^\top v \geq w_{1}-1$, then the problem \eqref{reduced_simplex_prob} is local.

Now, let us consider the optimal LP vertex $x^*$ of the original problem \ref{un_knapsack}. Due to correspondence between the problems \ref{un_knapsack} and \eqref{reduced_simplex_prob}, we have $a_0 - a^\top v = x^*_{1}$. Hence, if $W \geq w^2_{1}$, then $a_0 - a^\top v = x^*_1 \geq w_{1}$, and, consequently, the constrained $x_1 \geq 0$ is redundant.

So, by analogy with the previous proof, the original problem is equivalent to the group minimization problem
\begin{gather*}
\sum_{i = 2}^n w_i x_i \to \min\\
\begin{cases}
\sum_{i = 2}^n w_i x_i \equiv W \pmod {w_1}\\
x \in \ZZ^{n-1}_{\geq 0}.
\end{cases}
\end{gather*}
So, due to Theorems \ref{Gomory_th} and \ref{cyclic_group_complexity_th}, the problem can be solved by algorithms with the appropriate complexity bounds.

Finally, if $W \leq w^2_{1}$, then, due to \cite{KnapsackDPTrick}, we can apply an algorithm with the arithmetic complexity bound $O(n \cdot W)$ to achieve the guarantied complexity $O(n \cdot w^2_{opt})$. 
\end{proof}

\subsection{The Number Of Empty Simplices}\label{empty_number_subs}

We say that two simplices $\SC$ and $\SC^\prime$ are \emph{integer-equivalent} if there exists an affine transformation $T(x) = Q x + q$, where $Q$ is a unimodular $n \times n$ matrix and $q \in \ZZ^n$ is an integer translation, such that $\SC^\prime = T(\SC)$.

We say that a $n$-dimensional simplex $\SC$ is \emph{$\Delta$-modular} if $\SC$ is defined in the following way: $\SC = \PC(A,b)$ for $A \in \ZZ^{(n+1) \times n}$, $b \in \ZZ^{n+1}$, $\rank(A) = n$, and $\Delta = \Delta(A)$. We say that a simplex $\SC$ is \emph{empty} if $\SC \cap \ZZ^n = \emptyset$.
\begin{theorem}\label{empty_number_th}
Let $\Delta \leq n$, then the number of $n$-dimensional $\Delta$-modular simplicies modulo the integer-equivalence relation is bounded by
$$
O\left(\frac{n}{\Delta}\right)^{\Delta-1} \cdot \Delta^{\log_2(\Delta) + 2},
$$ which is a $(\Delta-1)$-degree polynomial on $n$, for fixed $\Delta$.
\end{theorem}
\begin{proof}
Let $\SC = \PC(A,b)$ be an arbitrary $n$-dimensional $\Delta$-modular simplex, where $A \in \ZZ^{(n+1)\times n}$, $\rank(A) = n$, and $b \in \ZZ^{n+1}$. W.l.o.g. we can assume that the system $A x \leq b$ is $\BC$-normalized, where $\BC = \intint n$ and $\abs{\det(A_{\BC})} = \Delta$. 

Due to properties of normalised systems (see Subsection \ref{normalization_subs}), we have that $0 \leq b_{\BC} \leq \diag(A_{\BC})$, $A_{\BC}$ is reduced to the HNF, and has the form
\begin{equation*}
A_{\BC} = \begin{pmatrix}
I_{s \times s} & \BZero_{s \times t}\\
H & T\\
\end{pmatrix},\quad\text{where}\quad
T = \begin{pmatrix}
A_{s+1\,s+1} & 0           & \dots & 0\\
A_{s+2\,s+1} & A_{s+2\,s+2} & \dots & 0\\
\hdotsfor{4}\\
A_{n \, s+1} &   A_{n\,s+2}        & \dots  & A_{n\,n}\\
\end{pmatrix}.
\end{equation*}
Here, the columns of $H$ are lexicographically sorted, $s$ is the number of $1$-s on the diagonal of $A_{\BC}$, and $t$ is the number of elements that are not equal $1$. Clearly, $s + t = n$ and $t \leq \log_2(\Delta)$.

Since any column $h$ of $H$ has the property $0 \leq h \leq \diag(A_{\BC})$ and columns of $H$ are lexicographically sorted, we have that the number of different matrices $H$ is bounded by
$$
\binom{s+\Delta-1}{\Delta-1} \leq \left( \frac{e\cdot(s + \Delta - 1)}{\Delta - 1} \right)^{\Delta-1} = O\left(\frac{n}{\Delta}\right)^{\Delta-1}.
$$
In a similar way, the number of different matrices $T$ can be roughly estimated as $\Delta^{t -1} \leq \Delta^{\log_2(\Delta) -1}$.

Hence, the total number of sub-systems $A_{\BC} x \leq b_{\BC}$ can be roughly estimated as
$$
O\left(\frac{n}{\Delta}\right)^{\Delta-1} \cdot \Delta^{\log_2(\Delta)}.
$$

Next, we need to estimate the number of ways, how to chose the last line $A_{n+1} x \leq b_{n+1}$ of the system $A x \leq b$. Denote $c = A_{n+1}^\top$ and $c_0 = b_{n+1}$. Since $\SC$ is a simplex, we have that $-c = A_{\BC}^\top \alpha$, for some $\alpha \in \RR^n_{>0}$. Since $\Delta(A) = \Delta$, we have that $\alpha \in (0,1]^n$. Hence, $-c$ is an element of the set 
$$
\MC \cap \ZZ^n,\quad\text{where } \MC = \{A^\top_{\BC} x \colon x \in (0,1]^n\}.
$$

Clearly, $\abs{\MC \cap \ZZ^n} \leq \Delta$, so there are at most $\Delta$ ways to chose $c$.

Finally, we need to estimate the number of ways to chose $c_0$. Let $v = A_{\BC}^{-1} b_{\BC}$ be a vertex of $\SC$ corresponding to the base $\BC$. Clearly, $c_0$ can be represented as $c_0 = \lceil c^\top v \rceil + \tau$, for $\tau \in \ZZ_{\geq 0}$. Since $\SC \cap \ZZ^n = \emptyset$, due to Corollary \ref{long_dist_cor}, $\tau \leq \Delta$.

Consequently, the number of different systems
$$
\begin{pmatrix}
A_{\BC}\\
c^\top
\end{pmatrix} x \leq 
\begin{pmatrix}
b_{\BC}\\
c_0
\end{pmatrix}
$$
defining the simplex $\SC$, is bounded by
$$
O\left(\frac{n}{\Delta}\right)^{\Delta-1} \cdot \Delta^{\log_2(\Delta) + 2}.
$$
\end{proof}

\subsection{The Average Case Analysis}\label{average_subs}

In this Subsection, we follow to \cite{DistributionsILP}. For functions $g,h \colon \RR_{>0} \to \RR_{>0}$, we write
$$
g \sim h \text{ if } \lim_{t \to \infty} \frac{g(t)}{h(t)} = 1 \quad \text{ and } \quad
g \lesssim h \text{ if } \limsup_{t \to \infty} \frac{g(t)}{h(t)} \leq 1.
$$ For a $n$-dimensional set $\PC \subseteq \RR^n$, we denote the $n$-dimensional Lebesgue measure by $\vol_n(\PC)$.

The next Lemma is given in \cite{DistributionsILP}, and it is a variation of classical known results in the Ehrhart theory, see, for instance, \cite[Theorem 7]{LatticeInvariantValuations} and \cite[Theorem 1.2]{NoteOnEhrhardCoeff}.
\begin{lemma}\label{Ehrhart_lm}
Let $\PC \subseteq \RR^n$ be a $m$-dimensional rational polytope and $\Lambda \subseteq \ZZ^n$ be a $n$-dimensional affine lattice. There exists a constant $\eta_{\PC,\Lambda} > 0$, such that $\abs{\Lambda \cap t \cdot \PC} \lesssim \eta_{\PC,\Lambda} \cdot t^m$. If $m = n$, then $\eta_{\PC,\Lambda} = \vol_m(\PC)/\det\Lambda$ and $\abs{\Lambda \cap t \cdot \PC} \sim \eta_{\PC,\Lambda} \cdot t^m$.
\end{lemma}

\subsubsection{The ILP Problems In The Canonical Form}

Let us fix a matrix $A \in \ZZ^{(n+m) \times n}$ of rank $n$ and consider the polyhedron, related to the \ref{ILP-CF} problem, parameterized by $b \in \ZZ^m$:
$$
\PC(b) = \{x \in \RR^n: A x \leq b\} \quad\text{and}\quad \PC_I(b) = \conv\bigl(\PC(b) \cap \ZZ^n \bigr).
$$
Let again $\Delta = \Delta(A)$. For any feasible base $\BC$, we denote $\Delta_{\BC} = \abs{\det(A_{\BC})}$, $\NotBC = \intint{n+m} \setminus \BC$ and 
$$
\PC(b,\BC) = \{x \in \RR^n \colon A_{\BC} x \leq b_{\BC}\} \quad\text{and}\quad \PC_I(b,\BC) = \conv\bigl(\PC(b,\BC) \cap \ZZ^n \bigr).
$$ Here, $\PC(b,\BC)$ is the corner polyhedron, related to $\BC$, and $\PC_I(b,\BC)$ is the convex hull of its integer points. Additionally, denote $v_{\BC} = A_{\BC}^{-1} b_{\BC}$ for the unique vertex of $\PC(b,\BC)$.

Presenting the next definition, we follow to the works \cite{DistributionsILP,SparsityAverage,IntegralityNumber}.
\begin{definition}\label{probability_def}
Let $\LC$ be an arbitrary $n$-dimensional sublattice of $\ZZ^n$ and $\Omega_{\LC,t} = \{b \in \LC\colon \|b\|_\infty \leq t\}$. Then, for $\AC \subseteq \ZZ^n$, we define
$$
\prob_{\LC,t}(\AC) = \cfrac{\abs{\AC \cap \Omega_{\LC,t}}}{\abs{\Omega_{\LC,t}}} \text{ and } \prob_{\LC}(\AC) = \liminf\limits_{t \to \infty} \prob_{\LC,t}(\AC).
$$

For $\LC = \ZZ^n$, we simply denote $\prob(\AC) = \prob_\Lambda(\AC)$.

The conditional probability of $\AC$ with respect to $\GC$ is denoted by the formula 
$$
\condprobL{\AC}{\GC} = \frac{\prob_{\LC}(\AC \cap \GC)}{\prob_{\LC}(\GC)}.
$$
\end{definition}

As it was noted in \cite{IntegralityNumber}, the functional $\prob(\AC)$ is not formally a probability measure, but rather a lower density function found from number theory.

\begin{theorem}\label{ILP_CF_exp_th}
\begin{gather*}
\text{Let }\FC = \{b \in \ZZ^{n + m} \colon \PC(b) \not= \emptyset \}\text{ and }\\
\GC =  \{ b \in \FC \colon \forall \BC\text{-- feasible base, } b_{\NotBC}-A_{\NotBC} v \geq (\Delta_{\BC}-1)\cdot \BUnit \}.
\end{gather*}

Then $\condprob{\GC}{\FC} = 1$ and, for any $b \in \GC$ and $c \in \ZZ^n$, the ILP problem $\max\{c^\top x \colon x \in \PC(b) \cap \ZZ^n\}$ is local. 

Consequently, for any $b \in \GC$ and for any vertex $z$ of $\PC_I(b)$, there exists a base $\BC$, such that $z$ is a vertex of $\PC_I(b,\BC)$, and 
\begin{enumerate}
\item $\|A(v_{\BC} - z)\|_0 \leq \log_2(\Delta_{\BC}) + m$, 
\\$\|A(v_{\BC} - z)\|_1 \leq (\Delta_{\BC}-1) + m\cdot(\Delta-1)$,\\ 
$\|A(v_{\BC} - z)\|_\infty \leq \Delta-1$;
\item $\|b - A z\|_0 \leq \log_2(\Delta_{\BC}) + m$, \\$\|b - A z\|_1 \leq (\Delta_{\BC}-1) + m\cdot(\Delta-1) + \|b_{\NotBC} - A_{\NotBC} v\|_1$,\\
$\|b - A z\|_\infty \leq \Delta-1 + \|b_{\NotBC} - A_{\NotBC} v\|_\infty$;
\item The point $z$ lies on a face of $\PC(b)$, whose dimension is bounded by $\log_2(\Delta_{\BC})$;
\item For any $c \in \ZZ^n$ the \ref{ILP-CF} problem $\max\{c^\top x \colon x \in \PC_I(b)\}$ can be solved by an algorithm with the arithmetic complexity bound 
$
O\bigl(n \cdot \Delta \cdot \log(\Delta)\bigr).
$
\end{enumerate}

\medskip
Additionally, if the system $A x \leq b$ is $\BC$-normalized, then the following statements hold for $z$:
\begin{enumerate}
\item $\|z\|_0 \leq 2 \log_2(\Delta_{\BC})$, $\|z\|_1 < 2 \Delta_{\BC}^2 \log_2 (\Delta_{\BC})$, $\|z\|_\infty < \Delta_{\BC}^2$,
\item $\|v_{\BC}-z\|_0 \leq 2 \log_2(\Delta_{\BC})$, $\|v_{\BC}-z\|_1 < (\Delta_{\BC})^2 \log_2(\Delta_{\BC})$, $\|v_{\BC}-z\|_\infty < (\Delta_{\BC})^2/2$.
\end{enumerate}
\end{theorem}
\begin{proof}
Due to Corollary \ref{long_dist_cor}, if $b \in \GC$, then the \ref{ILP-CF} problem is local. Since all the properties follows from Theorem \ref{local_properties_th}, we just need to prove that $$\condprob{\GC}{\FC} = 1.$$

Let $\LC = \ZZ^{n + m}$ and $\NotGC = \LC \setminus \GC$, it follows from the definition of $\GC$ that
\begin{equation}\label{G_inverted}
\NotGC \cap \FC \subseteq \bigcup_{\substack{\BC \subseteq \intint{n+m} \\ \BC \text{ -- base} }}
\bigcup_{j \in \NotBC } \bigcup_{r = 0}^{\Delta_{\BC}\cdot(\Delta-2)} \{b \in \FC \colon \text{$\BC$ is feasible and } \Delta_{\BC} b_j = A_j A_{\BC}^* b_{\BC} + r\}.
\end{equation}

We are going to prove that
\begin{equation}\label{barG_def}
\condprobL{\NotGC}{\FC} = \frac{\prob_{\LC}(\NotGC \cap \FC)}{\prob_{\LC}(\FC)} = \liminf_{t \to \infty} \frac{\abs{\Omega_{\LC,t} \cap \NotGC \cap \FC}}{\abs{\Omega_{\LC,t} \cap \FC}} = 0.
\end{equation}

Here, we assume that
\begin{equation}\label{nonzero_prob_assumption}
\prob_{\LC}(\FC) = \liminf_{t \to \infty} \frac{\abs{\Omega_{\LC,t} \cap \FC}}{\abs{\Omega_{\LC,t}}} > 0.
\end{equation} 

The correctness of this assumption will be shown later.

The formula \eqref{G_inverted} implies that the r.h.s. part of \eqref{barG_def} is at most
\begin{equation}\label{barG_prob}
\sum_{\substack{\BC \subseteq \intint m \\ \BC \text{ -- basis} }}
\sum_{j \in \NotBC} \sum_{r = 0}^{\Delta_{\BC}\cdot(\Delta-2)} \frac{\abs{\{b \in \Omega_{\LC,t} \cap \FC \colon \text{$\BC$ is feasible and } \Delta_{\BC} b_j = A_j A_{\BC}^* b_{\BC} + r\}}}
{\abs{\Omega_{\LC,t} \cap \FC}}.
\end{equation}

Let us denote 
$$
\PC_t = \{b \in \RR^{n+m} \colon \|b\|_{\infty} \leq t,\, A_{\NotBC} v_{\BC} \leq b_{\NotBC}\}.
$$
Then, $\abs{\Omega_{\LC,t} \cap \FC} \geq \abs{\LC \cap \PC_t}$, so the value of \eqref{barG_prob} is at most
\begin{equation}\label{barG_prob_simple}
\sum_{\substack{\BC \subseteq \intint m \\ \BC \text{ -- basis} }}
\sum_{j \in \NotBC} \sum_{r = 0}^{\Delta_{\BC}\cdot(\Delta-2)} \frac{\abs{\{b \in \LC \cap \PC_t \colon \Delta_{\BC} b_j = A_j A_{\BC}^* b_{\BC} + r\}}}
{\abs{\LC \cap \PC_t}}.
\end{equation}
Clearly, $\PC_t = t \cdot \PC_1$ and $\dim(\PC_t) = n + m$, for sufficiently large $t$, because $\dim(\PC) = n+m$. Consequently, $\LC \cap \PC_t \not= \emptyset$, for sufficiently large $t$.

Let us fix a base $\BC$, an index $j \in \NotBC$, a value $r \in \intint[0]{\Delta_{\BC}\cdot(\Delta-2)}$ and consider the fraction in the r.h.s. of \eqref{barG_prob_simple}: 
\begin{equation}\label{barG_ratio}
\frac{\abs{\LC \cap t\cdot\PC_1 \cap \{b \in \RR^{n+m} \colon \Delta_{\BC} b_j = A_j A_{\BC}^* b_{\BC} + r\}}}{\abs{\LC \cap t \cdot \PC_1}}.
\end{equation}

Clearly, the intersection $\LC \cap \{b \in \RR^{n+m} \colon \Delta_{\BC} b_j = A_j A_{\BC}^* b_{\BC} + r\}$ induces a new $(n + m -1)$-dimensional affine lattice $\LC^\prime$. Hence, due to Lemma \ref{Ehrhart_lm}, the ratio \eqref{barG_ratio} is at most 
$$
\frac{\abs{\LC^\prime \cap t \cdot \PC_1}}{\abs{\LC \cap t \cdot \PC_1}} \lesssim \const \cdot \frac{1}{t},
$$ where $\const = \frac{\eta_{\PC_1,\LC^\prime}}{\eta_{\PC_1, \LC}}$.
 
Additionally, as it was shown before,
$$
\frac{\abs{\Omega_{\LC,t} \cap \FC}}{\abs{\Omega_{\LC,t}}} \geq \frac{\abs{\LC \cap t\cdot\PC_1}}{\abs{\LC \cap t\cdot\BB_\infty}} > 0,\text{ for sufficiently large $t > 0$,}
$$ where $\BB_\infty$ denotes the unit ball with respect to the $l_\infty$-norm. The last fact proves the correctness of the assumption \eqref{nonzero_prob_assumption}.

Finally, we have
$$
\frac{\abs{\Omega_{\LC,t} \cap \NotGC \cap \FC}}{\abs{\Omega_{\LC,t} \cap \FC}} \lesssim m \cdot \binom{n+m}{n} \cdot \Delta^2 \cdot \frac{\eta_{\PC_1,\LC^\prime}}{\eta_{\PC_1, \LC}} \cdot \frac{1}{t} = O(1/t),
$$ and consequently
$$
\condprobL{\NotGC}{\FC} = 0.
$$
\end{proof}

\begin{remark}\label{probability_proof_stength_rm}
Proving this theorem, we partially followed to the work \cite{IntegralityNumber}. The main difference is that the paper \cite{IntegralityNumber} uses a wider set $\GC$, given by
$$
\GC =  \{ b \in \ZZ^{n+m} \colon \text{Either } \PC(b) = \emptyset \text{ or } \forall \BC\text{-- feasible base, } b_{\NotBC}-A_{\NotBC} v \geq (n \Delta)^2\cdot \BUnit \}.
$$
Additionally, we consider the conditional probability $\prob(\GC \mid \FC)$ that assumes the feasibility of $\PC(b)$, while the paper \cite{IntegralityNumber} proofs only $\prob(\GC) = 1$.

We note that the paper \cite{DistributionsILP} gives a much stronger analysis than \cite{IntegralityNumber}, but for the narrow class of the \ref{CLASSIC-ILP-SF} problem.
\end{remark}

\subsubsection{The ILP Problems In The Standard Form}

A similar result can be easily proved for the generalized ILP problem in the standard form \ref{ILP-SF}, due to the reduction Lemma \ref{ILPSF_to_ILPCF_lm}. 

Again, let us fix matrices $A \in \ZZ^{m \times n}$ of rank $m$, $\GC \in \ZZ^{(n-m)\times n}$, and $S \in \ZZ^{(n-m)\times(n-m)}$ be some matrix reduced to the SNF. For $b \in \ZZ^m$ and $g \in \ZZ^{n-m}$, we consider the parametric polyhedron
\begin{gather*}
 \PC(b) = \{x \in \RR^n_{\geq 0} \colon A x = b\} \quad\text{and}\\
 \PC_I(b,g) = \conv\bigl(\PC(b) \cap \{x \in \ZZ^n \colon G x \equiv g \pmod{S \cdot \ZZ^n}\}\bigr).
\end{gather*}

Additionally, we will use the following symbols $\Delta^* = \Delta(A) \cdot \abs{\det(S)}$, $v_{\BC} = A_{\BC}^{-1} b$, $\Delta_{\BC} = \abs{\det(A_{\BC})}$ and $\Delta^*_{\BC} = \Delta_{\BC} \cdot \abs{\det(S)}$.

\begin{theorem}\label{ILP_SF_exp_th}
\begin{gather*}
\text{Let }\FC = \left\{\binom{b}{g} \in \ZZ^{n} \colon \PC(b) \not= \emptyset \right\}\text{ and }\\
\GC =  \left\{ \binom{b}{g} \in \FC \colon \forall \BC\text{-- feasible base of $\PC(b)$, } v_{\BC} \geq (\Delta^*_{\BC}-1)\cdot\BUnit \right\}.
\end{gather*}

Then $\prob(\GC \mid \FC) = 1$. And, for any $\binom{b}{g} \in \GC$ and any vertex $z$ of $\PC_I(b,g)$, there exists a feasible base $\BC$ of $\PC(b)$, such that
\begin{enumerate}
    \item $\|v_{\BC}-z\|_0 \leq \log_2(\Delta^*_{\BC}) + m$,\\ 
         $\|v_{\BC}-z\|_1 \leq (\Delta^*_{\BC} - 1) + m \cdot (\Delta^* - 1)$,\\ 
         $\|v_{\BC}-z\|_\infty \leq \Delta^*-1$,
    
    \item $\|z\|_0 \leq \log_2(\Delta^*_{\BC}) + m$,
    
    \item for any $c \in \ZZ^n$, the \ref{ILP-SF} problem can be solved by an algorithm with the arithmetic complexity bound 
    $
    O\bigl((n-m) \cdot \Delta^* \cdot \log_2(\Delta^*)\bigr).
    $
\end{enumerate}
\end{theorem}
\GribanovAdd{
\begin{proof}
Due to the proof of Lemma \ref{ILPSF_to_ILPCF_lm}, the original problem is equivalent to the problem 
\begin{gather*}
    \hat c^\top \hat x \to \max \\
    \begin{cases}
    -P^{-1} \dbinom{\BZero}{S} \hat x \leq P^{-1}\dbinom{b}{g}\\
    \hat x \in \ZZ^d,
    \end{cases}
\end{gather*}
where $P = \dbinom{A}{G}$. Now, the proof follows to the proof of Theorem \ref{ILP_CF_exp_th} with respect to the lattice $\LC = \inth(P^{-1}) = \ZZ^{n}$.
\end{proof}
}

Due to Remark \ref{classic_ILP_rm}, the \ref{CLASSIC-ILP-SF} problem can be easily reduced to the \ref{ILP-SF} problem with $S = I$ and $\abs{\det(S)} = 1$. Consequently, we automatically obtain the following Corollary:

\begin{corollary}\label{classic_exp_cor}
\begin{gather*}
\text{Let }\FC = \left\{b \in \ZZ^{m} \colon \PC(b) \not= \emptyset \right\}\text{ and }\\
\GC =  \left\{ b \in \FC \colon \forall \BC\text{-- feasible base of $\PC(b)$, } v_{\BC} \geq (\Delta_{\BC}-1)\cdot\BUnit \right\}.
\end{gather*}

Then $\prob(\GC \mid \FC) = 1$. And, for any $b \in \GC$ and any vertex $z$ of $\PC_I(b)$, there exists a feasible base $\BC$ of $\PC(b)$, such that
\begin{enumerate}
    \item $\|v_{\BC}-z\|_0 \leq \log_2(\Delta_{\BC}) + m$,\\ 
         $\|v_{\BC}-z\|_1 \leq (\Delta_{\BC} - 1) + m \cdot (\Delta - 1)$,\\ 
         $\|v_{\BC}-z\|_\infty \leq \Delta-1$,
    
    \item $\|z\|_0 \leq \log_2(\Delta_{\BC}) + m$,
    
    \item for any $c \in \ZZ^n$, the \ref{ILP-SF} problem can be solved by an algorithm with the arithmetic complexity bound 
    $
    O\bigl((n-m) \cdot \Delta \cdot \log_2(\Delta)\bigr).
    $
\end{enumerate}
\end{corollary}

Results of the last Corollary are not new, they have been previously proven in work \cite{DistributionsILP}. Moreover, the paper \cite{DistributionsILP} gives probability distributions of the related sparsity and proximity functions.

\section*{Acknowledgments}

Results of the Sections \ref{sparsity_and_proximity_sec}, \ref{algorithms_sec} were prepared within the framework of the Basic Research Program at the National Research University Higher School of Economics (HSE). Results of the Section \ref{partial_cases_sec} were prepared under financial support of Russian Science Foundation grant No 21-11-00194.

The authors would like to thank J.~Paat for useful discussions during preparation of this article.

\GribanovAdd{
Additionally, the authors would like to thank anonymous reviewers for their useful comments and remarks.
}

\section*{Appendix}\label{appendix_sec}

Proof of the Lemma \ref{adj_lm}. 
\begin{proof}
The structure of $A_{\BC}^*$ directly follows from the triangular structure of $A_{\BC}$ and from the definition of $A_{\BC}^*$.

Let the matrix $H$ be obtained from $A_{\BC}$ by deleting any row and any column. The value of $\abs{\det(H)}$ corresponds to some element of $A_{\BC}^*$. It is easy to see that $H$ is a lower triangular matrix with at most one additional diagonal. We can expand the determinant of $H$ by the first row, using the Laplace rule.  Then, $\abs{\det(H)} \leq 2^{k-1} d_1 d_2 \dots d_{k}$, where $k$ is the number of diagonal elements in $A_{\BC}$ that are not equal to $1$, and $(d_1, d_2, \dots, d_{k})$ is a sequence of diagonal elements. Since $k \leq \log_2(\delta)$, we have $\abs{\det(H)} \leq \delta^2/2$.

Let us prove the second part of this Lemma. Columns of the matrix $A_{\BC}^{-1}$ have at most $\log_2(\delta) + 1$ non-zero components. More precisely, the last $\lfloor \log_2(\delta) \rfloor$ columns of $A_{\BC}^{-1}$ have at most $\lfloor \log_2 \delta \rfloor$ non-zero components that are concentrated in the last $\lfloor \log_2(\delta) \rfloor$ coordinates. Consequently, $\|A_{\BC}^{-1} y\|_0 \leq \alpha + \log_2(\delta)$. The second inequality trivially follows from the first part of this lemma.
\end{proof}

Proof of the Lemma \ref{tight_delta_ineq}.
\begin{proof}
Let $y = \frac{2 n }{m}$, then 
$$
y \leq 2 + \log_2(y) + \log_2 \frac{e}{2} + \frac{2 \log_2(\Delta)}{m}.
$$
Let $d = 1 + \log_2 e + \frac{2 \log_2(\Delta)}{m}$, then 
$$
y \leq \log_2(y) + d.
$$
Let $h$ be defined by the equality 
$$
y = d + \log_2(d) + h,
$$ then
$$
d + \log_2(d) + h \leq \log_2(d + \log_2(d) + h) + d,
$$ and 
$$
h \leq \log_2\frac{d + \log_2(d) + h}{d} = \log_2 \left(1 + \frac{\log_2(d) + h}{ d} \right).
$$

The function in the r.h.s. of the previous inequality attains its maximum in the point $d = \frac{e}{2^h}$, hence
$$
h \leq \log_2 \left( 1 + \frac{2^h \cdot \log_2 e}{e} \right).
$$ Next, 
$$
2^h \leq 1 + \frac{2^h \cdot \log_2 e}{e}\quad\text{and}\quad 2^h \leq \frac{e}{e-\log_2e},
$$ hence
$$
h \leq \bar c := \log_2 \frac{e}{e - \log_2 e}.
$$
Returning to $y$, we have $y \leq d + \log_2(d) + \bar c$ and
$$
\frac{2 n}{m} \leq 1 + \log_2 e + \frac{2 \log_2(\Delta)}{m} + \log_2\left(1 + \log_2 e + \frac{2 \log_2(\Delta)}{m} \right) + \bar c.
$$
Finally,
$$
n \leq \frac{\bar c+1 + \log_2(2e)}{2} \cdot m + \log_2(\Delta) + \frac{m}{2} \cdot \log_2 \left( \frac{1}{2} + \frac{1}{2} \log_2 e + \frac{\log_2(\Delta)}{m} \right).
$$
\end{proof}

Proof of the Lemma \ref{rough_delta_ineq}.
\begin{proof}
Let $g = \frac{\log_2(\Delta)}{m}$. We are going to use the inequality 
$$
f(x + h) \leq f(x) + f^\prime(x) \cdot h,
$$ which is true for concave functions, such as $\log_2(x)$. Taking $h = g - k$ and $x = c_2 + k$, we have
$$
\log_2(c_2 + g) \leq \log_2(c_2 + k) + \frac{1}{(c_2 + k)\cdot \ln 2} \cdot (g - k).
$$

We substitute the last inequality to the inequality in the Lemma's definition:
$$
n \leq \left(c_1 + \log_2\sqrt{c_2 + k} - \frac{k}{(c_2 + k)\cdot \ln 4} \right) \cdot m + \frac{1}{(c_2 + k)\cdot \ln 4} \cdot \log_2(\Delta).
$$

Finally, we note that 
$$
\frac{k}{(c_2 + k)\cdot \ln 4} \geq \frac{1}{\ln 4}.
$$ 
\end{proof}

\bibliography{grib_biblio}

\end{document}

%% file: grib_defs.tex
\DeclareMathOperator{\RR}{\mathbb{R}}
\DeclareMathOperator{\QQ}{\mathbb{Q}}
\DeclareMathOperator{\ZZ}{\mathbb{Z}}
\DeclareMathOperator{\BB}{\mathbb{B}}

\DeclareMathOperator{\GC}{\mathcal{G}}

\DeclareMathOperator{\BC}{\mathcal{B}}
\DeclareMathOperator{\NC}{\mathcal{N}}

\DeclareMathOperator{\PC}{\mathcal{P}}
\DeclareMathOperator{\DC}{\mathcal{D}}
\DeclareMathOperator{\SC}{\mathcal{S}}
\DeclareMathOperator{\ZC}{\mathcal{Z}}
\DeclareMathOperator{\LC}{\mathcal{L}}
\DeclareMathOperator{\HC}{\mathcal{H}}
\DeclareMathOperator{\MC}{\mathcal{M}}
\DeclareMathOperator{\AC}{\mathcal{A}}
\DeclareMathOperator{\FC}{\mathcal{F}}
\DeclareMathOperator{\KC}{\mathcal{K}}

\DeclareMathOperator{\JC}{\mathcal{J}}
\DeclareMathOperator{\IC}{\mathcal{I}}

\DeclareMathOperator{\BHC}{\overline{\HC}}

\DeclareMathOperator{\NotBC}{\overline{\BC}}
\DeclareMathOperator{\NotGC}{\overline{\GC}}

\DeclareMathOperator{\rank}{rank}
\DeclareMathOperator{\cone}{cone}
\DeclareMathOperator{\conv}{conv\!.\!hull}
\DeclareMathOperator{\affh}{aff.\!hull}
\DeclareMathOperator{\linh}{span}
\DeclareMathOperator{\inth}{\Lambda}

\DeclareMathOperator{\vertex}{vert}
\DeclareMathOperator{\lcm}{lcm}

\DeclareMathOperator{\poly}{poly}

\DeclareMathOperator{\const}{const}
\DeclareMathOperator{\sgn}{sgn}

\DeclareMathOperator{\vol}{vol}
\DeclareMathOperator{\diag}{diag}

\DeclareMathOperator{\disc}{disc}
\DeclareMathOperator{\herdisc}{herdisc}
\DeclareMathOperator{\shortest}{shortest}

\DeclareMathOperator{\BUnit}{\mathbf 1}
\DeclareMathOperator{\BZero}{\mathbf 0}

\DeclareMathOperator{\prob}{\Pr}

\DeclareMathOperator{\supp}{supp}
\DeclareMathOperator{\zeros}{zeros}

\newcommand*{\intint}[2][1]{\{#1, \dots, #2\}}

\newcommand{\GribanovAdd}[1]{{\color{blue}#1}}

\newcommand{\condprobL}[2]{\prob_{\Lambda}(#1 \mid #2)}
\newcommand{\condprob}[2]{\prob(#1 \mid #2)}

\newcommand{\emptyBinom}[2]{\genfrac{}{}{0pt}{}{#1}{#2}}

\newcommand{\trinom}[3]{\begin{pmatrix} #1 \\ #2 \\ #3 \end{pmatrix}}